\documentclass[a4paper,fleqn,numbers]{cas-dc}

\def\tsc#1{\csdef{#1}{\textsc{\lowercase{#1}}\xspace}}
\tsc{WGM}
\tsc{QE}

\def\nn{\nonumber}
\def\beq{\begin{equation}}
\def\eeq{\end{equation}}
\def\bea{\begin{eqnarray}}
\def\eea{\end{eqnarray}}
\def\ba{\begin{array}}
\def\ea{\end{array}}

\def\bitem{\begin{itemize}}
\def\eitem{\end{itemize}}
\def\ben{\begin{enumerate}}
\def\een{\end{enumerate}}

\def\eg{{\it e.g., \/}}

\def\ie{{\it i.e.,\ \/}}

\newcommand{\mbbE}{\mathbb{E}}

\def\Pbf{{\bm P}}

\newcommand{\E}{\mathbb{E}}

\usepackage{bm, amsfonts}
\usepackage{enumitem}
\setlist[itemize]{noitemsep, topsep=2pt}
\usepackage{setspace}
\usepackage{threeparttable}

\let\OLDthebibliography\thebibliography
\renewcommand\thebibliography[1]{
  \OLDthebibliography{#1}
  \setlength{\parskip}{1pt}
  \setlength{\itemsep}{1pt plus 0.3ex}
}

\def\ben{\begin{enumerate}[itemsep=0pt, parsep=0pt, topsep=2pt]}
\usepackage{tabularx}

\newcommand{\reg}[1]{\mathcal{R}_{#1}}

\newtheorem{theorem}{Theorem}

\newtheorem{proposition}{Proposition}
\newdefinition{rmk}{Remark}
\newproof{proof}{Proof}
\begin{document}
\let\WriteBookmarks\relax
\def\floatpagepagefraction{1}
\def\textpagefraction{.001}

\shorttitle{}    

\shortauthors{}  

\title [mode = title]{Renewable-Colocated Green Hydrogen Production: Optimal Scheduling and Profitability}




\author[1]{Siying Li}[orcid=0000-0001-7118-7659]
\ead{sl2843@cornell.edu}
\affiliation[1]{organization={School of Electrical and Computer Engineering, Cornell University},
            city={Ithaca},
            postcode={14853}, 
            state={NY},
            country={USA}}

\author[1]{Lang Tong}[orcid=0000-0003-3322-2681]
\ead{lt35@cornell.edu}

\author[2]{Timothy D. Mount}%[<options>]
\ead{tdm2@cornell.edu}
\affiliation[2]{organization={Dyson School of Applied Economics and Management, Cornell University},
            city={Ithaca},
            postcode={14853}, 
            state={NY},
            country={USA}}

\author[3]{Kanchan Upadhyay}%[<options>]
\ead{KUpadhyay@nyiso.com}
\affiliation[3]{organization={New York Independent System Operator (NYISO)},
            city={Rensselaer},
            postcode={12144}, 
            state={NY},
            country={USA}}

\author[3]{Harris Eisenhardt}%[]
\ead{HEisenhardt@nyiso.com}

\author[3]{Pradip Kumar}%[]
\ead{PKumar@nyiso.com}

\nonumnote{This work was supported in part by the National Science Foundation under Grants 2218110 and 2412776, and by the NYISO Grid Development Research Project.}

\begin{abstract}
We study the optimal green hydrogen production and energy market participation of a renewable-colocated hydrogen producer (RCHP) that utilizes onsite renewable generation for both hydrogen production and grid services. Under deterministic and stochastic profit-maximization frameworks, we analyze RCHP's multiple market participation models and derive closed-form optimal scheduling policies that dynamically allocate renewable energy to hydrogen production and electricity export to the wholesale market. Analytical characterizations of the RCHP's operating profit and the optimal sizing of renewable and electrolyzer capacities are obtained. We use real-time renewable generation and electricity price data from three independent system operators to evaluate the impacts of market prices and environmental policies on RCHP's profitability.
\end{abstract}

%\item Optimize real-time operations for renewable-colocated hydrogen plants.
%\item Quantify optimal capacity sizing for renewables and electrolyzers.
%\item Prosum er models maximize profits for renewable-colocated hydrogen plants.
%\item Grid participation boosts solar-hydrogen profits by up to 216\% in New York.
%\item Wind systems achieve 12\% higher electrolyzer utilization than solar in New York.

%\nocite{*}

% Keywords

\maketitle

% Main text
\section{Introduction}\label{sec:intro}
% !TEX root = LiEtal26AE_extend.tex
Green hydrogen, also known as renewable hydrogen, refers to hydrogen certified as being produced with lifecycle greenhouse gas emissions below a specified threshold. The specific certification criteria vary by country and region and continue to evolve. One common pathway for green hydrogen classification is to demonstrate that production is powered directly by colocated\footnote{Although \textit{collocate} is the standard spelling, we adopt the IT-oriented variant \textit{colocate} to emphasize the joint operation and co-optimization between renewable and hydrogen producers.} renewable energy sources or, if grid electricity is used, that the electricity is certified as renewable \cite{EUCMDR:23, DOE2023}. Both the European Union and the U.S. have established tradable certification mechanisms, such as Guarantees of Origin (GO) and Renewable Energy Certificates (RECs), to verify that grid-imported electricity originates from renewable sources.

This work examines the optimal production of green hydrogen by grid-connected Renewable-Colocated Hydrogen Producers (RCHPs), as illustrated in Fig.~\ref{fig:schematic}.
Our study is motivated by several emerging trends. First, green hydrogen is a dispatchable, emission-free resource that can complement intermittent renewable sources, enhancing system reliability and resource adequacy. Second, global hydrogen demand reached about 100 million tonnes in 2024, representing an increase of over 2\% from 2023, and is expected to continue rising steadily over the coming decade. However, demand for low-emissions hydrogen accounted for less than 1\% of the total, underscoring the significant room for growth \cite{iea_global_hydrogen_review_2025}.
% Second, hydrogen demand is projected to grow at an annual rate of 9.3\% through 2030, yet in 2023, only 5\% of the \$170 billion hydrogen market was classified as green. 
If RCHPs prove profitable, they could accelerate the clean energy transition of a promising energy market. Third, recent market trends favor RCHPs. Curtailment of renewable generation has increased sharply in high-penetration regions: in 2024, approximately 10\% of wind generation in Britain and 30\% in Northern Ireland were curtailed \cite{FT2025}. Similarly, total renewable curtailment in CAISO exceeded 3.4 TWh, highlighting the untapped potential of green hydrogen production.
Finally, grid-connected RCHPs can enhance grid reliability by dynamically adjusting their operations, maximizing hydrogen production during oversupply conditions, and prioritizing renewable power delivery to the grid when supply is constrained. This dual role of RCHPs is a focus of this work.

The potential of green hydrogen remains a topic of debate, with valid concerns about its underlying assumptions and economic viability. A central issue is the cost and profitability of green hydrogen production, as current electrolyzer technology faces significant challenges, including high energy demand, low production efficiency, and limited utilization due to the intermittency of renewable power.

While this article does not claim to resolve these challenges, it provides an analytical framework, a cost-effective and optimized production strategy, and empirical evidence on the short-run profitability of green hydrogen. Specifically, we address the following key questions:
\begin{itemize}
\item What is the profit-maximizing hydrogen production level given colocated renewable generation?
\item What is the expected profitability of an RCHP under different market participation models?
\item How do prices in the electricity and hydrogen markets and environmental policies on subsidies affect RCHP's profitability?
\item How do electrolyzer and renewable generation capacities impact profitability, and what are their optimal sizes given a fixed cost budget?
\end{itemize}
\begin{figure}[!htb]
\centering
\includegraphics[scale=0.3]{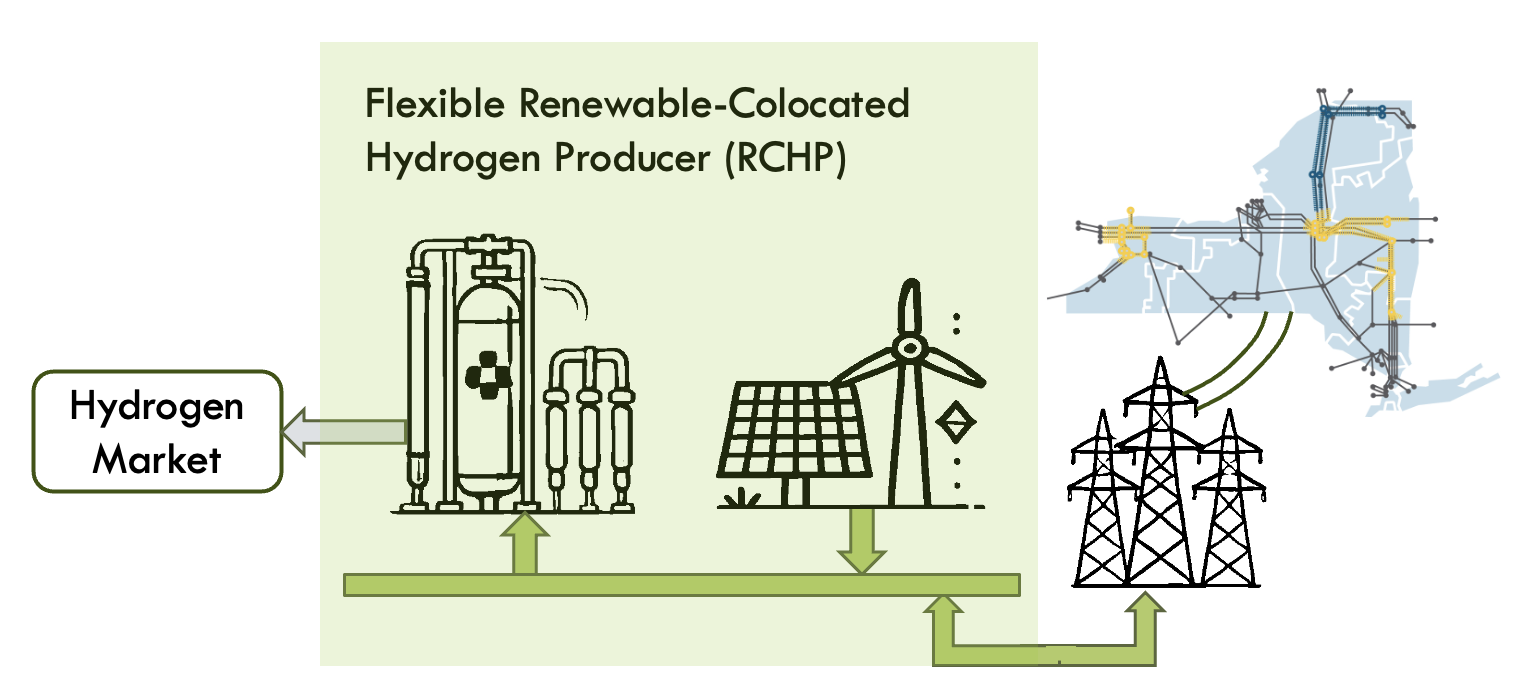}
\vspace{-0.5em}
\caption{\scriptsize Schematic of a flexible RCHP.}\label{fig:schematic}
\vspace{-2em}
\end{figure}

Fig.~\ref{fig:schematic} shows a generic {\em flexible RCHP model} schematic. By flexible RCHP, we mean that the RCHP not only produces hydrogen but also participates in a wholesale electricity market, exporting surplus renewable power to the grid and procuring certified renewable energy for hydrogen production.
To highlight the functional and economic differences between existing RCHP models and the flexible RCHP model proposed here, we define four RCHP configurations based on the RCHP interface with the wholesale electricity market: M0 for an RCHP with no interconnection with the grid, M1-x for a unidirectional interface with the wholesale market, either as a producer (M1-p) or as a consumer (M1-c). Finally, M2 is a prosumer model that has a bidirectional interface with the market. Specifically, the four RCHP market participation models are as follows:
\begin{itemize}
\item \underline{\it Standalone Hydrogen Producer (M0):} An RCHP under M0 produces hydrogen exclusively from colocated renewable. M0 is a benchmark for comparisons.

\item \underline{\it Renewable Producer (M1-p):} A renewable energy producer in the electricity market; the RCHP produces hydrogen while exporting surplus renewable energy.

\item \underline{\it Price-elastic Consumer (M1-c):} A flexible demand in the electricity market, importing certified renewable energy to supplement onsite renewable power for hydrogen production.

\item \underline{\it Flexible Prosumer (M2):} A prosumer, the RCHP can purchase certified renewable energy from the market to supplement onsite renewables for hydrogen production or sell surplus renewable power to the grid.
\end{itemize}
This article focuses on the optimal scheduling and {\em short-run} profitability analysis of flexible RCHP under M2 with results applied to other models as special cases.

\subsection{Related Literature}
Before discussing specific models, it is important to clarify our fundamental perspective. While a substantial body of literature evaluates renewable-hydrogen integration from a centralized system operator's perspective---focusing on network constraints, optimal power flow, and system cost minimization \cite{Zhou&etal:22IJHE,Clegg:16IET}---our work adopts the perspective of the RCHP as an independent market participant, focusing exclusively on the resource's profit-driven decision making.

We classify relevant literature based on the market participation models outlined above.

\paragraph{1) RCHP as a Renewable Producer.}
The power producer model M1-p, often referred to as a hybrid renewable hydrogen system, has been extensively studied over the past decade \cite{Glenk_Reichelstein_2019,Jiang&etal:19IJHE,McDonagh&etal:20AE,LUCAS&etal:22AE}. Among these studies, the techno-economic analysis by Glenk and Riechelstein \cite{Glenk_Reichelstein_2019} is most closely related to our work. Under the M1-p model, Glenk and Riechelstein derive conditions on electricity and hydrogen prices under which an RCHP system would be economically viable, showing that, in Texas, an RCHP under M1-p is competitive when the hydrogen price exceeds \$3.53/kg. Our results are consistent with conclusions of \cite{Glenk_Reichelstein_2019,Jiang&etal:19IJHE,McDonagh&etal:20AE,LUCAS&etal:22AE} and complement these existing analysis with the most general market interaction model. In addition, our approach differs from these existing methodologies by focusing on short-run profitability and providing closed-form profit-maximizing solutions.

\paragraph{2) RCHP as a Price-elastic Consumer.}
Under the flexible demand model M1-c, an RCHP can import electricity from the grid when it enhances profitability but cannot export power. Compared to M0, M1-c introduces a different cost structure, where electricity prices directly impact hydrogen production costs. Moreover, to meet green hydrogen certification standards, the costs of purchasing RECs for grid-imported electricity must be included. These costs are mostly ignored in the literature.

Extensive research has been conducted on hydrogen production as a demand-side participant in the wholesale market, particularly under the Power-to-Gas (PtG) framework \cite{VanLeeuwen&Mulder:18AE,Li&etal:22TSE, Li&etal:24AE}. Gahleitner provided a comprehensive review of physical PtG plants from 1990 to 2012 \cite{Gahleitner:13IJHE}. 
% Studies such as \cite{Zhang&etal:23TSG} focus on grid-powered hydrogen production. 
% A key insight from \cite{Wu&etal:22E} is the importance of incorporating revenues from grid services, including demand response and downstream grid services.
Colocated renewable power is considered in \cite{VanLeeuwen&Mulder:18AE}, where Van Leeuwen and Mulder introduce the concept of willingness to pay (WTP) for electricity used in hydrogen production. Alongside electricity prices, WTP is crucial in defining an RCHP's strategy for importing electricity. The efficiency characterization of such PtG systems is explored in \cite{Frank&etal:18AE}.

\paragraph{3) RCHP as a Flexible Prosumer.}
We find no prior work that directly addresses the prosumer model under M2, though partial overlaps exist in \cite{Alkano&etal:18TSG, Pavic&etal:22AE, Li&etal:25TSE}, which consider systems with electrolyzers, hydrogen storage, and fuel cells capable of bidirectional power exchange with the grid. Unlike these studies, our model restricts the ``producer'' role to the direct sale of surplus renewable generation, excluding fuel cell-based power-to-hydrogen-to-power (P2H2P) conversion from the real-time energy market framework. This design choice is driven by the prevailing economic realities of energy arbitrage: the relatively low round-trip efficiency of P2H2P systems makes them economically uncompetitive for frequent cycling \cite{Zhang&etal:20FER}. 
We provide a detailed numerical example in Appendix \ref{appendix:fuel_cell} to quantitatively justify this setup under historical market conditions.

Among the aforementioned works, \cite{Alkano&etal:18TSG} investigates decentralized coordination among multiple PtG facilities. \cite{Pavic&etal:22AE} integrates power, hydrogen, and gas systems in a large-scale multi-energy framework. \cite{Li&etal:25TSE} studies the coordinated dispatch of renewable and hydrogen systems within a grid-connected microgrid. The complexity of these models limits their analytical tractability and the ability to obtain closed-form insights.

The economic viability of RCHPs as prosumers critically depends on electricity and hydrogen prices and their market coupling. Li and Mulder \cite{Li&Mulder:21AE} explore this coupling, suggesting that PtG can reduce price volatility, improve social welfare (including carbon costs), and support grid operations. They conclude that the high investment costs and the displacement of cheaper energy carriers outweigh these benefits.

To clearly show the scope and novelty of our work, a comparative summary of our approach relative to prior research is presented in Table~\ref{tab:comparison}.
\begin{table*}[htbp]
\centering
\footnotesize
\begin{threeparttable}
\caption{\small Comparison of the proposed RCHP framework with existing literature.}
\vspace{-1em}
\label{tab:comparison}
\begin{tabular}{@{}llllll@{}}
\toprule
Reference & Grid Interface & Optimization & Solution Form & Profitability Analysis & Green Subsidies \\ 
\midrule
\cite{Glenk_Reichelstein_2019} & M1-p & Deterministic & Closed-form & Yes & Renewable  \\
\cite{Jiang&etal:19IJHE} & M1-p & Stochastic & CCP\tnote{a} & Yes& No  \\
\cite{McDonagh&etal:20AE} & M1-p & Deterministic & MILP\tnote{b} & Yes & No\\
\cite{LUCAS&etal:22AE} & M1-p & Deterministic & Simulation & Yes & No\\
\cite{VanLeeuwen&Mulder:18AE} & M1-c & Deterministic & Simulation & Yes & No \\ 
\cite{Li&etal:22TSE} & M1-c & Stochastic & DP\tnote{c}  & No & No \\ 
\cite{Li&etal:24AE} & M1-c & Robust & Two-layer algorithm  & No & No \\ 
\cite{Alkano&etal:18TSG} & M2 & Deterministic & MPC\tnote{d} & No & No \\ 
\cite{Pavic&etal:22AE} & M1-c, M2 & Deterministic & MIP\tnote{e} & Yes & No \\ 
\cite{Li&etal:25TSE} & M2 & Deterministic & MILP & No & No \\ 
\cite{Li&Mulder:21AE} & M2 & Deterministic & MCP\tnote{f} & Yes & Renewable \& H2 \\ 
This Paper & M0, M1-p, M1-c, M2 & Deterministic \& Stochastic & Closed-form & Yes & Renewable \& H2 \\ \bottomrule
\end{tabular}
\begin{tablenotes}[para, flushleft]
\footnotesize
\item[a] CCP: Chance-Constrained Programming.
\item[b] MILP: Mixed-Integer Linear Programming.
\item[c] DP: Dynamic Programming. 
\item[d] MPC: Model Predictive Control. 
\item[e] MIP: Mixed-Integer Programming. 
\item[f] MCP: Mixed Complementarity Problem.
\end{tablenotes}
\end{threeparttable}
\vspace{-1.5em}
\end{table*}

\subsection{Summary of Contributions}
\paragraph{1) Optimal Production Plan.}
We develop a framework to maximize profit in RCHP real-time operations, integrating four market participation models into a unified structure. Although the RCHP's bidirectional grid participation generally introduces non-convexity due to differing costs as a producer or consumer, we demonstrate that the optimization can be exactly relaxed into a convex program under practical market conditions. We leverage this property to derive a closed-form solution, resulting in an efficient threshold policy that maps renewable generation and real-time locational marginal price (LMP) to optimal electrolyzer inputs and grid transactions.

Fig.~\ref{fig:Optimality} (a) shows the structure of the optimal hydrogen production plan under M2 when the electrolyzer capacity $Q_{\mbox{\tiny\sf H}}$ is smaller than the renewable power capacity $Q_{\mbox{\tiny\sf R}}$, where the plane of LMP $\pi^{\mbox{\tiny\sf LMP}}$ and renewable generation $Q$ are partitioned into four distinct regions $\reg{1}$-$\reg{4}$. When the LMP $\pi^{\mbox{\tiny\sf LMP}}$ is high in region $\reg{3}$, the RCHP produces no hydrogen and exports all its renewable power to the grid. In region $\reg{4}$ where the colocated renewable power $Q$ is abundant but LMP $\pi^{\mbox{\tiny\sf LMP}}$ is below threshold $\overline{\pi}^{\mbox{\tiny\sf LMP}}$, the RCHP maximizes its hydrogen production and exports the surplus renewable power. When the LMP $\pi^{\mbox{\tiny\sf LMP}}$ and renewable energy level $Q$ are both low in the region $\reg{1}$, the RCHP imports power and maximizes its hydrogen production. Perhaps the most intriguing is when the renewable energy level $Q$ is low, but the LMP falls between the lower and upper price thresholds, $\pi^{\mbox{\tiny\sf LMP}} \in (\underline{\pi}^{\mbox{\tiny\sf LMP}}, \overline{\pi}^{\mbox{\tiny\sf LMP}})$.  In this case, using only its renewable power for hydrogen production is optimal, making the RCHP a net-zero producer. This net-zero region arises from the discrepancy between RECs' purchasing and selling prices. See Sec.~\ref{sec:optimalplan}. 
% \tcb{Note that, threshold-based decision rules have been extensively explored in energy storage and demand response studies, where thresholds are often obtained numerically \cite{Shi&etal, hegde2011optimal}. The thresholds in the proposed solution can be computed a priori, allowing for simple threshold comparisons without numerical computation.}

\vspace{-0.5em}
\begin{figure}[h]
\centering
\includegraphics[scale=0.3]{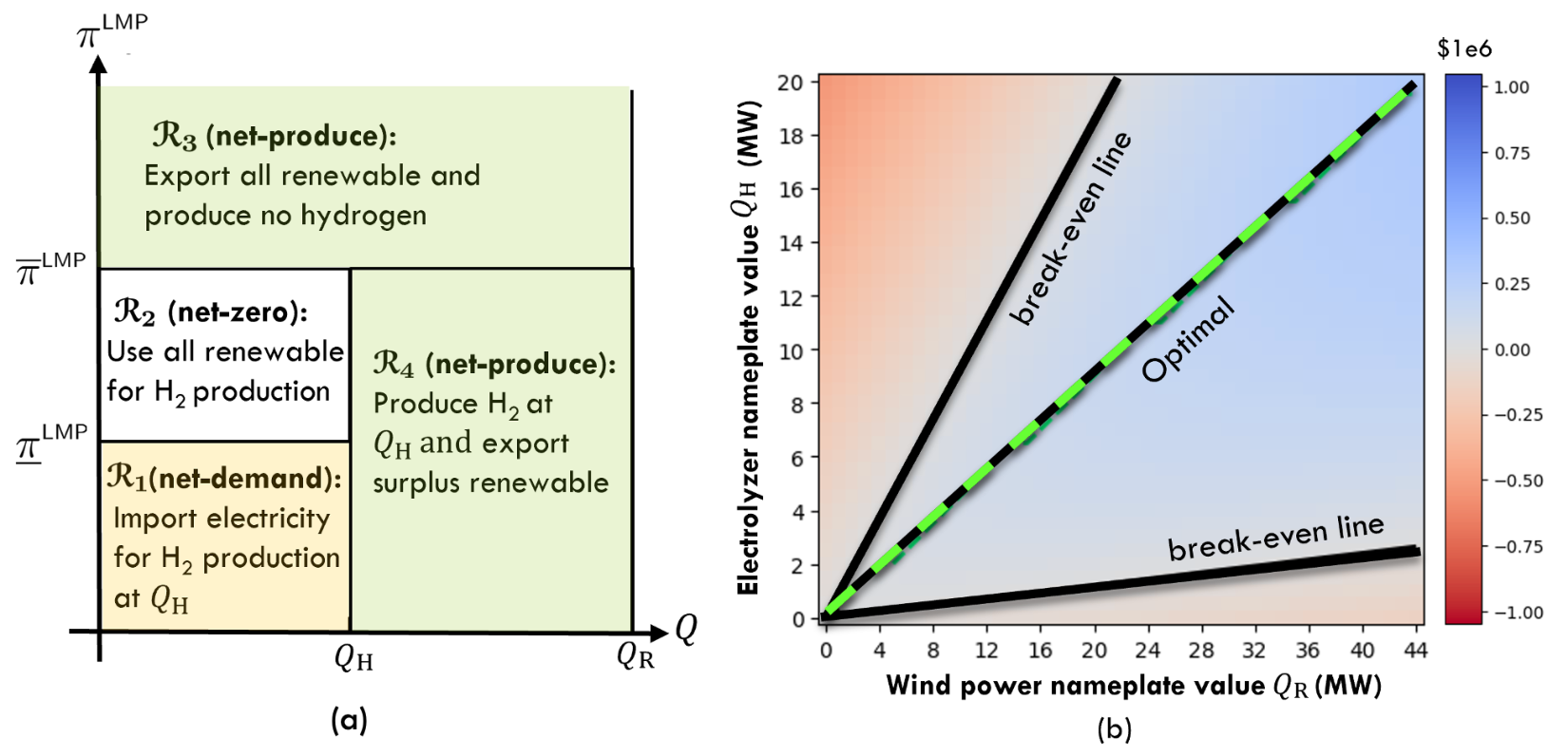}
\vspace{-1em}
\caption{\scriptsize (a) Optimal hydrogen production policy for flexible  RCHP (M2).  (b) Operating profit heatmap as the function of the electrolyzer and renewable nameplate capacities. }\label{fig:Optimality}
\vspace{-2em}
\end{figure}

\paragraph{2) Short-run Profit and Optimal Nameplate Capacities.}
We analyze the profitability of the four RCHP models under the optimal hydrogen production policy, deriving closed-form expressions for expected operating profit as functions of renewable and electrolyzer capacities as shown in Fig.~\ref{fig:Optimality} (b). We show that the profitable region (in blue) and the deficit region (in orange) are convex cones with linear boundaries (in black) being the break-even region. Within the convex cones, profits or deficits grow in proportion to linearly increasing capacities.

The profit function suggests an optimal design trade-off: oversized electrolyzers lead to underutilization, while undersized electrolyzers waste renewable power, both reducing profitability. Theorem \ref{Thm:profitability} presents the optimal matched electrolyzer capacity as a linear function of renewable capacity, as shown in Fig.~\ref{fig:Optimality} (b). As a by-product, our analysis also produces an easily computable profit forecasting tool. Specifically, the expressions given in Proposition \ref{Prop:linear} can be used to predict future profits using historical renewable power and LMP statistics.

\paragraph{3) Insights from Empirical Studies.}
Empirical studies based on actual LMP and renewable generation profiles from three independent system operators (NYISO, CAISO, and MISO) are presented in Sec.~\ref{sec:simulation}, revealing key insights into RCHP operation and profitability. Sensitivity analyses indicate that the hydrogen price and environmental subsidies play a pivotal role in shaping RCHP profitability and determining the preferred participation model. The profitability gap between the producer/consumer models (M1-p and M1-c) and the prosumer model (M2) varies with market conditions, with close alignment at low or high hydrogen prices.

We also observe pronounced cross-regional and resource-specific patterns. Distinct LMP and renewable generation profiles across ISOs lead to different profitability outcomes for RCHP participation models. For instance, in MISO, high wind availability enhances revenues from renewable sales, making the producer model more profitable than the consumer model, whereas the opposite trend emerges in NYISO and CAISO. Across all regions, wind-based RCHPs generally achieve higher utilization rates for hydrogen production than solar-based ones, as the concentrated generation peaks of solar led to more curtailment and surplus market sales.

For future reference, key variables and system parameters are summarized in Table \ref{tab:symbols}.
\begin{table}[htbp]
\begin{center}
\caption{\small RCHP decision variables and system parameters.}\label{tab:symbols}
%\vspace{-0.5em}
\begin{tabularx}{\columnwidth}{@{}lX@{}}
\toprule
\multicolumn{2}{@{}l}{Indices and Time Parameters} \\ [0.5ex] 
$t$ & Index for scheduling intervals \\ 
$\Delta T$ & Duration of each scheduling interval\\ [0.5ex] 
\hline
\multicolumn{2}{@{}l}{Exogenous variables} \\ [0.5ex] 
$\eta_t$ & Capacity factor of renewable generation at $t$ \\
$\pi^{\mbox{\tiny\sf H}} $  & Hydrogen market price \\
$\pi^{\mbox{\tiny\sf LMP}}_t$  & Real-time locational marginal price at $t$\\ [0.5ex]
\hline
\multicolumn{2}{@{}l}{Decision variables} \\ [0.5ex]
$H_t$ & Hydrogen production at $t$ \\
$\Pbf_t$ & Decision variables in the profit maximization at $t$\\
$P^{\mbox{\tiny\sf EX}}_t$ & Power exported to the grid at $t$\\
$P^{\mbox{\tiny\sf H}}_t$ & Electrolyzer input power at $t$\\
$P^{\mbox{\tiny\sf IM}}_t$ & Power imported from the grid at $t$\\ [0.5ex]
\hline
\multicolumn{2}{@{}l}{System parameters} \\ [0.5ex]
$\alpha^{\mbox{\tiny\sf H}}$, $\alpha^{\mbox{\tiny\sf R}}$ & Per-unit cost of electrolyzer/renewable capacity \\
$c^{\mbox{\tiny\sf W}}$ & Non-electricity marginal cost  of hydrogen supply\\
$\gamma$& Electrolyzer efficiency factor\\
$Q_{\mbox{\tiny\sf H}}$, $Q_{\mbox{\tiny\sf R}}$& Electrolyzer/renewable nameplate capacity\\[0.5ex]
\hline
\multicolumn{2}{@{}l}{Production credits} \\ [0.5ex]
$\tau^{\mbox{\tiny\sf H}}$, $\tau^{\mbox{\tiny\sf R}}$ & Per-unit hydrogen/renewable production credit\\
$\tau^{\mbox{\tiny\sf EX}}_{\mbox{\tiny\sf REC}}$ & REC prices for exported renewable\\ 
$\tau^{\mbox{\tiny\sf IM}}_{\mbox{\tiny\sf REC}}$ & REC prices for imported renewable\\ [0.5ex]
\bottomrule
%$J_\theta^{\mbox{\tiny\sf GP}}(\Pbf_t)$ & Gross profit in interval $t$ under M2\\
%$J^{\mbox{\tiny\sf OP}}_n(Q_{\mbox{\tiny\sf R}},Q_{\mbox{\tiny\sf H}})$ & Average operating profit in $n$ intervals\\
\end{tabularx}
 \end{center}
 \vspace{-2em}
 \end{table}

%\subsubsection{Notations and Nomenclature}
%The notations used in this paper follow standard conventions. We use $x$ to represent a scalar and $\xbf$ for a vector. The expression $(x)^+$ denotes the positive part of $x$, \ie $(x)^+ \coloneqq \max\{x, 0\}$. A list of the major symbols is provided in Table \ref{tab:symbols}.
%\begin{table}[htbp]
%\caption{Major symbols}\label{tab:symbols}
%\begin{center}
%\vspace{-1em}
%\begin{tabular}{ll}
%\hline
%$\eta_t \in \mathbb{R}_{+}$& Capacity factor of renewable generation at time $t$ \\
%$\gamma \in \mathbb{R}_{++}$& Conversion rate of the electrolyzer \\
%$h_t\in \mathbb{R}_+$& Hydrogen production at time $t$ \\
%$P_{{\rm ex},t}\in \mathbb{R}_+$& Power sold to the grid at time $t$ \\
%$P_{{\rm im},t}\in \mathbb{R}_+$& Power purchased from the grid at time $t$ \\
%$P_{h,t}\in \mathbb{R}_+$& Power used for hydrogen production at time $t$ \\
%$\pi_{e,t} \in \mathbb{R}$& Real-time electricity price at time $t$ \\
%$\pi_{h} \in \mathbb{R}_{++}$& Hydrogen price \\
%$\tau_{\rm ptc} \in \mathbb{R}_{+}$& Production tax credit per unit of renewable electricity \\
%$Q_{H}\in \mathbb{R}_+$& Capacity of the electrolyzer \\
%$Q_{R}\in \mathbb{R}_+$& Capacity of the renewable generation plant \\
%$\tau_{h} \in \mathbb{R}_+$& Green hydrogen credit \\
%$\tau_{rb} \in \mathbb{R}_+$& Price for buying renewable energy certificates\\
%$\tau_{rs} \in \mathbb{R}_+$& Price for selling renewable energy certificates\\
%\hline
%\end{tabular}
%\end{center}
%\end{table} 

\section{Production Model, Cost, and Profit}\label{sec:formulation}
% !TEX root = LiEtal26AE_extend.tex
\paragraph{1) Production Model.} We adopt the standard linear hydrogen production model for an electrolyzer \cite{Glenk_Reichelstein_2019}:
 \begin{equation}
H_{t} = \gamma P_t^{\mbox{\tiny\sf H}}\Delta T, \label{eq:linear_prod}
\end{equation}
where $H_{t}$ represents the hydrogen produced (kg) in interval $t$, $P_t^{\mbox{\tiny\sf H}}$ the power used for hydrogen production (kW), $\gamma$ the electrolyzer's efficiency factor (kg/kWh), and  $\Delta T$ the scheduling interval duration, assumed to be aligned with the wholesale market pricing interval.   Without loss of generality, we assume $\Delta T=1$.

The linear model \eqref{eq:linear_prod} provides a good approximation of the electrolyzer's production behavior while maintaining analytical tractability. It is important to note that $\gamma$ represents the effective system-wide conversion efficiency, and can be readily scaled down to subsume the auxiliary energy consumption of localized hydrogen compression \cite{ArmijoPhilibert:20IHE, McDonag&etal:20AE}. Furthermore, this linear model can be extended to a piecewise linear formulation to better capture nonlinear operational characteristics. The main analytical results remain valid under this extension, although the optimal scheduling policy then involves an exponentially growing number of regions. See Appendix \ref{appendix:piecewise_prod} for details.

The input power of an electrolyzer with nameplate capacity $Q_{\mbox{\tiny\sf H}}$ is bounded, satisfying the following equation:
\begin{equation}
0 \leq P_t^{\mbox{\tiny\sf H}} \leq Q_{\mbox{\tiny\sf H}}. \label{eq:EL_cons}
\end{equation}

\paragraph{2) RCHP Fixed Costs.}
An RCHP's production cost includes both fixed and variable costs.
The fixed operating cost $C^{\mbox{\tiny\sf F}}$ covers the amortized capital investment, management, maintenance, insurance, and equipment degradation. It is typically assumed to be linear with respect to the nameplate capacities of the renewable plant and electrolyzer, denoted by $Q_{\mbox{\tiny\sf R}}$ and $Q_{\mbox{\tiny\sf H}}$, respectively (we discuss the relaxation to nonlinear cost structures in Appendix \ref{appendix:nonlinear_cost}.)
\bea \label{eq:CF}
    C^{\mbox{\tiny\sf F}}(Q_{\mbox{\tiny\sf R}}, Q_{\mbox{\tiny\sf H}})&=&
    \alpha^{\mbox{\tiny\sf R}}Q_{\mbox{\tiny\sf R}} + \alpha^{\mbox{\tiny\sf H}}Q_{\mbox{\tiny\sf H}},
\eea
where $\alpha^{\mbox{\tiny\sf R}}$ and $\alpha^{\mbox{\tiny\sf H}}$ represent the annual fixed costs per unit capacity of renewable and electrolyzer facilities, respectively \cite{Glenk_Reichelstein_2019, Reichelstein:15EE}.

\paragraph{3) RCHP Variable Costs.}
The marginal cost of hydrogen supply includes (a) the marginal cost of renewable energy, assumed negligible, (b) the marginal cost of grid-imported power, which is the sum of the time-varying LMP $\pi^{\mbox{\tiny\sf LMP}}_t$ and its associated REC price $\tau_{\mbox{\tiny\sf REC}}^{\mbox{\tiny\sf IM}}$, and (c) the non-electricity marginal cost $c^{\mbox{\tiny\sf W}}$, which accounts for consumable inputs such as water and the downstream logistics costs of hydrogen storage and transport to the point of sale \cite{IEA:19}.

We assume that the start-up and shut-down costs of the electrolyzer are negligible, as modern Proton Exchange Membrane (PEM) systems---the industry standard for integrating intermittent renewables---can undergo cold starts and rapid ramping with minimal thermal stress \cite{ButtlerSpliethoff:18RSER, IRENA:20}. Furthermore, operation-induced degradation effects can be readily incorporated either as a constant adder to the non-electricity marginal cost $c^{\mbox{\tiny\sf W}}$, or as a static derating factor applied to the electrolyzer efficiency $\gamma$ \cite{Campana&etal:25AE}.

\paragraph{4) RCHP's Revenue.}
We model the RCHP as a competitive price-taker in both the hydrogen and wholesale electricity markets, which is a standard modeling approach \cite{Dadkhah&etal:21IJHE, AkhavanMohsenian:14TSG}. This assumption is practically justified by the facility's negligible share in the overall markets, coupled with strict regulatory oversight by Independent Market Monitors (IMMs) \cite{Potomac:24} in wholesale electricity markets, which compel utility-scale resources to bid their true marginal values rather than exercising market power.

The revenue of an RCHP consists of income from (a) selling hydrogen at the market price $\pi^{\mbox{\tiny\sf H}}$ and receiving per-kg green hydrogen production credits $\tau^{\mbox{\tiny\sf H}}$, (b) exporting surplus renewable power to the energy market at the LMP $\pi^{\mbox{\tiny\sf LMP}}_t$ and earning the associated renewable energy certificates $\tau_{\mbox{\tiny\sf REC}}^{\mbox{\tiny\sf EX}}$, and (c) obtaining the renewable production tax credits $\tau^{\mbox{\tiny\sf R}}$.

\paragraph{5) RCHP's Gross Profit.}
Let vector $\theta$ denote the set of techno-economic system parameters:
$$\theta:=(\pi^{\mbox{\tiny\sf H}},\tau^{\mbox{\tiny\sf H}},\tau^{\mbox{\tiny\sf R}},\tau_{\mbox{\tiny\sf REC}}^{\mbox{\tiny\sf EX}},\tau_{\mbox{\tiny\sf REC}}^{\mbox{\tiny\sf IM}},\gamma,c^{\mbox{\tiny\sf W}},Q_{\mbox{\tiny\sf R}},Q_{\mbox{\tiny\sf H}}),$$
and let $\Pbf_t=\Big[P^{\mbox{\tiny\sf H}}_{t}, P^{\mbox{\tiny\sf EX}}_t, P^{\mbox{\tiny\sf IM}}_t \Big]$ be the vector of power dispatch variables: the electrolyzer input $P^{\mbox{\tiny\sf H}}_{t}$, the power exported to the grid  $P^{\mbox{\tiny\sf EX}}_t$, and the grid-imported power $P^{\mbox{\tiny\sf IM}}_t$. While the hydrogen price $\pi^{\mbox{\tiny\sf H}}$, environmental attributes $\tau^{\mbox{\tiny\sf H}},\tau^{\mbox{\tiny\sf R}},\tau_{\mbox{\tiny\sf REC}}^{\mbox{\tiny\sf EX}},\tau_{\mbox{\tiny\sf REC}}^{\mbox{\tiny\sf IM}}$, and marginal cost $c^{\mbox{\tiny\sf W}}$ are modeled as constant within the real-time operational horizon, this merely reflects the practical reality that these values remain significantly more stable compared to the high-frequency fluctuations of LMPs. It is important to note that our framework does not require these parameters to be time-invariant; any shift in policy or market conditions can be seamlessly integrated into the model as a change in $\theta$, without altering the underlying problem structure.

Given realized renewable capacity factor $\eta_t$ and electricity LMP $\pi^{\mbox{\tiny\sf LMP}}_t$, the {\em gross profit} of an RCHP under M2 in interval $t$ as a function of $\Pbf_t$  is
\begin{align}
J^{\mbox{\tiny\sf GP}}_\theta(\Pbf_t)&= (\pi^{\mbox{\tiny\sf H}}+\tau^{\mbox{\tiny\sf H}})(\gamma P^{\mbox{\tiny\sf H}}_{t})+(\pi^{\mbox{\tiny\sf LMP}}_{t}+\tau^{\mbox{\tiny\sf EX}}_{\mbox{\tiny\sf REC}}) P^{\mbox{\tiny\sf EX}}_t\nonumber\\
&+\tau^{\mbox{\tiny\sf R}}\eta_tQ_{\mbox{\tiny\sf R}}-c^{\mbox{\tiny\sf W}}\gamma P^{\mbox{\tiny\sf H}}_{t}- (\pi^{\mbox{\tiny\sf LMP}}_t+\tau^{\mbox{\tiny\sf IM}}_{\mbox{\tiny\sf REC}})P_t^{\mbox{\tiny\sf IM}},\label{eq:objwREC}
\end{align}
where the first three terms on the right-hand side are revenues from selling produced hydrogen in the hydrogen market, exporting renewable energy to the energy market, and obtaining renewable production credits, respectively. The last two terms are the costs associated with consumable inputs and the import of certified renewables from the grid. Note that this model is conservative in the sense that it does not consider the possibility of the RCHP retaining self-generated RECs to reduce REC purchases.
%We ignore the start-up and shut-down costs since the prosumer model of RCHP operates continuously.

\paragraph{6) Hydrogen Storage and Delivery System.}
For analytical tractability and generalizability, we assume that the produced hydrogen can be stored or delivered without explicit capacity constraints. Given the diverse and emerging industrial paradigms, such as pipeline blending \cite{Topolski&etal:22NREL, CleggMancarella:16IET}, direct off-take \cite{ArmijoPhilibert:20IHE}, and truck transportation, modeling each of these methods is nontrivial, and each imposes different physical limits.
Consequently, specializing to a single storage and delivery mechanism would limit our framework's broad applicability. This unconstrained approach also aligns with existing literature \cite{Glenk_Reichelstein_2019}. To quantitatively justify this simplification, we provide an ex-post numerical evaluation in Appendix \ref{appendix:sto_limits}, which demonstrates that finite storage limits result in marginal profit losses (<3.1\%) across all market participation models.

\section{Profit-Maximizing Production}\label{sec:optimalplan}
% !TEX root = LiEtal26AE_extend.tex
This section presents the optimal RCHP production plan in closed form, from which we gain intuitions and insights into the operational strategies of the RCHP. For presentation simplicity, we assume positive LMPs. The general solution involving negative LMPs can be found in the Appendix.

\subsection{Profit Maximization}
At each time interval, the RCHP maximizes its gross profit by setting the optimal hydrogen production level and the amount of renewable power to trade in the market.
The profit maximization program under the prosumer model M2 is formulated as follows.
\begin{subequations}\label{eq:maxprofit}
\begin{align}
& \hspace{-1em} \underset{\Pbf_t=(P^{\mbox{\tiny\sf H}}_{t}, P^{\mbox{\tiny\sf EX}}_t, P^{\mbox{\tiny\sf IM}}_t)}{\rm maximize}&&
 J^{\mbox{\tiny\sf GP}}_\theta(\Pbf_t)\\ \nn
& \hspace{-1em} \mbox{subject to constraints}  && \\
% &\mbox{(production function)} && H_t=\gamma P^{\mbox{\tiny\sf H}}_t,\label{eq:cons1}\\
& \hspace{-1em} \mbox{(power balance)} &&  0\le P^{\mbox{\tiny\sf H}}_t+P^{\mbox{\tiny\sf EX}}_t-P^{\mbox{\tiny\sf IM}}_t\le \eta_t Q_{\mbox{\tiny\sf R}},\label{eq:cons2}\\
& \hspace{-1em} \mbox{(I/O complementarity)} && P^{\mbox{\tiny\sf IM}}_tP^{\mbox{\tiny\sf EX}}_t=0,\label{eq:cons3}\\
% &\mbox{(I/C complementarity)} && P^{\mbox{\tiny\sf IM}}_tP^{\mbox{\tiny\sf CT}}_t=0,\label{}\\
% &\mbox{(O/C complementarity)} && P^{\mbox{\tiny\sf EX}}_tP^{\mbox{\tiny\sf CT}}_t=0,\label{}\\
& \hspace{-1em} \mbox{(electrolyzer input limit)} && 0 \leq P^{\mbox{\tiny\sf H}}_{t} \leq Q_{\mbox{\tiny\sf H}},\label{eq:cons4}\\
& \hspace{-1em} \mbox{(renewable export limit)} && 0 \leq  P^{\mbox{\tiny\sf EX}}_{t}  \leq \eta_tQ_{\mbox{\tiny\sf R}},\label{eq:cons5}\\
& \hspace{-1em} \mbox{(grid-import limit)} && 0 \leq  P^{\mbox{\tiny\sf IM}}_{t}  \leq Q_{\mbox{\tiny\sf H}}.\label{eq:cons6}
% &\mbox{(curtailment limit)} && 0 \leq  P^{\mbox{\tiny\sf CT}}_{t}  \leq \eta_tQ_{\mbox{\tiny\sf R}}\label{}.
\end{align}
\end{subequations}
\vspace{-1em}
\begin{rmk}[Exact Convex Relaxation]\label{rmk:convexity}
    While the complementarity constraint \eqref{eq:cons3} is non-convex, it can be exactly relaxed under the no risk-free arbitrage condition for RECs, \ie$\tau_{\mbox{\tiny\sf REC}}^{\mbox{\tiny\sf IM}} > \tau_{\mbox{\tiny\sf REC}}^{\mbox{\tiny\sf EX}}$. Consequently, the profit maximization \eqref{eq:maxprofit} can be equivalently reformulated and solved as a convex program. 
\end{rmk}

Proof sketch in Appendix~\ref{sec:proof_rmk}.
The price spread $\tau_{\mbox{\tiny\sf REC}}^{\mbox{\tiny\sf IM}} > \tau_{\mbox{\tiny\sf REC}}^{\mbox{\tiny\sf EX}}$ is consistent with prevailing market structures. It ensures that simultaneous power injection and withdrawal is economically suboptimal, indicating that the optimal solution to the relaxed convex program inherently satisfies the original complementarity constraint \eqref{eq:cons3}.

This optimization framework further encompasses other RCHP market participation models as special cases. To ensure non-trivial operation across these cases, we assume that producing hydrogen using self-generated renewable power and selling it at least breaks even for RCHP, \ie $\pi^{\mbox{\tiny\sf H}}+\tau^{\mbox{\tiny\sf H}}-c^{\mbox{\tiny\sf W}}\geq 0$. This guarantees that the electrolyzer does not remain shut down. Accordingly, under the standalone model M0, $P^{\mbox{\tiny\sf IM}}_t=P^{\mbox{\tiny\sf EX}}_t=0$ and the optimal hydrogen production is $H_t^*=\gamma \min\{\eta Q_{\mbox{\tiny\sf R}}, Q_{\mbox{\tiny\sf H}}\}$. Under the producer model M1-p, $P^{\mbox{\tiny\sf IM}}_t=0$. Under the consumer model M1-c, $P^{\mbox{\tiny\sf EX}}_t=0$.

\subsection{Structure of Optimal Hydrogen Production} \label{sec:structure}
We first describe the structure of the optimal solution of \eqref{eq:maxprofit}
and present intuitions behind the solution, followed by the closed-form optimal production plan in  Theorem~\ref{Thm:solution}.

Fig. \ref{fig:Optimalplan_M2} shows the structure of the optimal production under the four RCHP models on the price-quantity plane: the x-axis represents the level of renewable generation and the y-axis the real-time LMP. Two pairs of thresholds define the optimal production plan: LMP thresholds  ($\underline{\pi}^{\mbox{\tiny\sf LMP}}$, $\overline{\pi}^{\mbox{\tiny\sf LMP}}$) on the y-axis and renewable thresholds $(Q_{\mbox{\tiny\sf H}}, Q_{\mbox{\tiny\sf R}})$ on the x-axis. Both threshold pairs are functions of known system and market parameters and set before real-time operations. In particular, the renewable power thresholds are nameplate capacities of the electrolyzer and renewable plant, and the LMP thresholds are given by
\begin{equation}\label{eq:thresholds}
\begin{array}{cl}
\underline{\pi}^{\mbox{\tiny\sf LMP}} &= \gamma(\pi^{\mbox{\tiny\sf H}}+\tau^{\mbox{\tiny\sf H}}-c^{\mbox{\tiny\sf w}})-\tau_{\mbox{\tiny\sf REC}}^{\mbox{\tiny\sf IM}},\\[1pt]
\overline{\pi}^{\mbox{\tiny\sf LMP}} &= \gamma(\pi^{\mbox{\tiny\sf H}}+\tau^{\mbox{\tiny\sf H}}-c^{\mbox{\tiny\sf w}})-\tau_{\mbox{\tiny\sf REC}}^{\mbox{\tiny\sf EX}}.\\
\end{array}
\end{equation}
The gap between the two thresholds is from the REC price spread discussed in Remark~\ref{rmk:convexity}. It follows that $\overline{\pi}^{\mbox{\tiny\sf LMP}}>\underline{\pi}^{\mbox{\tiny\sf LMP}}$, and without loss of generality, we assume that $\underline{\pi}^{\mbox{\tiny\sf LMP}}>0$. Because these closed-form thresholds are parameterized by the system parameters in $\theta$, the RCHP's optimal scheduling policy is inherently adaptive. For instance, if green hydrogen credits phase out, \ie $\tau^{\mbox{\tiny\sf H}} = 0$, or REC rules change, the thresholds will automatically shift to reflect the new profit-maximizing boundaries. This analytical property demonstrates the robustness of the proposed strategy across diverse and evolving policy and market landscapes. Furthermore, the proposed framework easily accommodates explicit grid interconnection and hydrogen market absorption limits. These constraints simply introduce localized saturation effects, such as excess curtailment, while preserving the strict convexity and analytical tractability of the optimal scheduling policy.

\begin{figure}[!htb]
    \centering
    \includegraphics[scale=0.32]{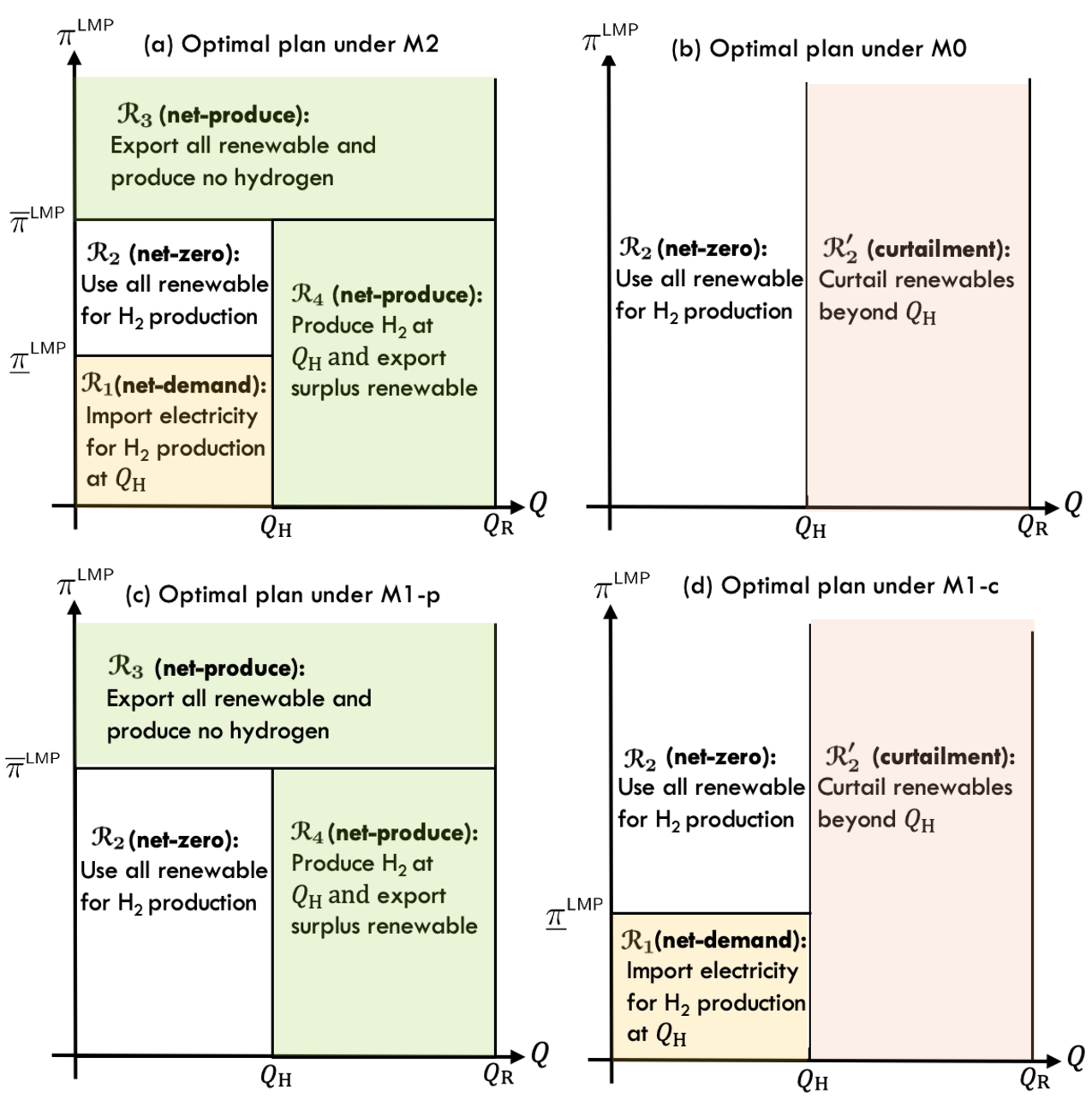}
    \caption{\scriptsize Optimal production plans of RCHP under different models when $Q_{\mbox{\tiny\sf H}}<Q_{\mbox{\tiny\sf R}}$. See Sec.~\ref{sec:Appendix_OptimalPlan} for the $Q_{\mbox{\tiny\sf H}}>Q_{\mbox{\tiny\sf R}}$ case.}
    \label{fig:Optimalplan_M2}
    \vspace{-1em}
\end{figure}

The optimal production plan has five operating regions $\reg{1}$-$\reg{4}$ and $\reg{2}'$.
To the market operator, the RCHP is a net consumer in $\reg{1}$, a net-zero participant in $\reg{2}$ and $\reg{2}'$ where the RCHP neither consumes from nor produces to the grid, and a renewable producer in $\reg{3}$ and $\reg{4}$.

Under the general prosumer model M2, the RCHP operates as follows in the $\reg{1}$-$\reg{4}$ regions, as shown in Fig. \ref{fig:Optimalplan_M2} (a).
\begin{itemize}
    \item[$\reg{1}$:]  $\reg{1}$ is the scenario of low LMP and limited renewable generation. The RCHP maximizes hydrogen production to the electrolyzer capacity $Q_{\mbox{\tiny\sf H}}$ by supplementing renewables with grid-imported power.

    \item[$\reg{2}$:] $\reg{2}$ is the scenario of moderate LMP and limited renewables. The RCHP uses all colocated renewable energy for hydrogen production without importing power from the grid.  The RCHP is a net-zero participant.

    \item [$\reg{3}$:] $\reg{3}$ is the high LMP scenario. The RCHP exports all renewable to the grid and produces zero hydrogen.

    \item [$\reg{4}$:] $\reg{4}$ is the high renewable scenario, where colocated renewable generation is beyond the electrolyzer capacity, and the LMP is moderate. RCHP maximizes hydrogen production to $Q_{\mbox{\tiny\sf H}}$ and sells the surplus renewable energy to the wholesale market.
\end{itemize}
The optimal production under other participation models are slight generalizations of  Fig. \ref{fig:Optimalplan_M2} (a) with new scenarios corresponding to renewable curtailment under M0 and M1-c.  Under the standalone model M0 shown in Fig. \ref{fig:Optimalplan_M2} (b), the RCHP has no connection to the grid.  It is optimal to use all renewables to produce hydrogen up to the electrolyzer capacity and curtail the over production.

Under the producer model M1-p, as shown in Fig. \ref{fig:Optimalplan_M2} (c), the RCHP behaves the same as under M2 in $\reg{3}$ and $\reg{4}$. Since it cannot import power, the net-zero region $\reg{2}$ extends downward to replace $\reg{1}$ in Fig. \ref{fig:Optimalplan_M2} (a).

Under the consumer model M1-c, as shown in Fig. \ref{fig:Optimalplan_M2} (d), the RCHP behaves the same as under M2 in $\reg{1}$. Since the RCHP cannot export renewable, the net-zero region $\reg{2}$ extends to replace $\reg{3}$ in Fig. \ref{fig:Optimalplan_M2} (a). Additionally, any renewable beyond the electrolyzer capacity must be curtailed.

% \vspace{-1em}
\subsection{Optimal Production in Closed Form} \label{sec:solution}
Theorem~\ref{Thm:solution} below validates the solution structure in
Fig.~\ref{fig:Optimalplan_M2} and provides explicit expressions for the solution to the profit maximization program \eqref{eq:maxprofit}. The proof is given in Appendix~\ref{sec:Appendix_OptimalPlan}.

\begin{theorem}\label{Thm:solution}
Under the prosumer model M2 (Fig.~\ref{fig:Optimalplan_M2} (a)) and positive LMP, the solution $\Pbf^*_t=\Big[P^{\mbox{\tiny\sf H}*}_t,P^{\mbox{\tiny\sf EX}*}_t,P^{\mbox{\tiny\sf IM}*}_t\Big]$ of \eqref{eq:maxprofit} in interval $t$ as a function of $\pi^{\mbox{\tiny\sf LMP}}_t$ and capacity factor $\eta_t$ is
{
\begin{eqnarray}\label{eq:optsol}
\mathbf{P}^*_t &=&F^{\mbox{\tiny\sf M2}}_\theta(\pi_t^{\mbox{\tiny\sf LMP}},\eta_t)\nonumber\\
&:=&
\left\{
\begin{aligned}
&\Big[Q_{\mbox{\tiny\sf H}},(\eta_t Q_{\mbox{\tiny\sf R}} - Q_{\mbox{\tiny\sf H}})^+, (Q_{\mbox{\tiny\sf H}} - \eta_t Q_{\mbox{\tiny\sf R}})^+\Big],
\pi^{\mbox{\tiny\sf LMP}}_t \le \underline{\pi}^{\mbox{\tiny\sf LMP}};\\[2pt]
&\Big[0, \eta_t Q_{\mbox{\tiny\sf R}}, 0\Big],
\pi^{\mbox{\tiny\sf LMP}}_t \ge \overline{\pi}^{\mbox{\tiny\sf LMP}};\\[2pt]
&\Big[\min\{\eta_t Q_{\mbox{\tiny\sf R}}, Q_{\mbox{\tiny\sf H}}\},(\eta_t Q_{\mbox{\tiny\sf R}} - Q_{\mbox{\tiny\sf H}})^+, 0\Big], \mbox{otherwise.}
\end{aligned}
\right.
\end{eqnarray}
}
Under the standalone model M0 (Fig.~\ref{fig:Optimalplan_M2} (b)),
\begin{equation}\label{eq:optsol0}
\mathbf{P}^*_t=F^{\mbox{\tiny\sf M0}}_\theta(\pi_t^{\mbox{\tiny\sf LMP}},\eta_t):=\Big[\min\{\eta_t Q_{\mbox{\tiny\sf R}}, Q_{\mbox{\tiny\sf H}}\},\ 0,\ 0\Big].
\end{equation}
Under the producer model M1-p (Fig.~\ref{fig:Optimalplan_M2} (c)),
\begin{eqnarray}\label{eq:optsol01p}
\mathbf{P}^*_t&=&F^{\mbox{\tiny\sf M1p}}_\theta(\pi_t^{\mbox{\tiny\sf LMP}},\eta_t)\\
&:=&\left\{
\begin{aligned}&\Big[0,\ \eta_t Q_{\mbox{\tiny\sf R}},\ 0\Big], \pi^{\mbox{\tiny\sf LMP}}_t \ge \overline{\pi}^{\mbox{\tiny\sf LMP}};\\[2pt]
&\big(\min\{\eta_t Q_{\mbox{\tiny\sf R}}, Q_{\mbox{\tiny\sf H}}\},(\eta_t Q_{\mbox{\tiny\sf R}} - Q_{\mbox{\tiny\sf H}})^+, 0\big), \mbox{otherwise}.\\
\end{aligned}
\right. \nonumber
\end{eqnarray}
Under the producer model M1-c (Fig.~\ref{fig:Optimalplan_M2} (d)),
\begin{eqnarray}\label{eq:optsol01c}
\mathbf{P}^*_t&=&F^{\mbox{\tiny\sf M1c}}_\theta(\pi_t^{\mbox{\tiny\sf LMP}},\eta_t)\\
&:=&
\left\{
\begin{aligned}
&\Big[Q_{\mbox{\tiny\sf H}},\ 0,\ (Q_{\mbox{\tiny\sf H}} - \eta_t Q_{\mbox{\tiny\sf R}})^+\Big], \pi^{\mbox{\tiny\sf LMP}}_t \le \underline{\pi}^{\mbox{\tiny\sf LMP}};\\[2pt]
&\Big[\min\{\eta_t Q_{\mbox{\tiny\sf R}}, Q_{\mbox{\tiny\sf H}}\}, 0, 0\Big], \quad \text{otherwise}.
\end{aligned}
\right.\nonumber
\end{eqnarray}
\end{theorem}

Here, the notation $(x)^+ := \max\{0, x\}$ denotes the standard positive part operator.
Note that the thresholds are computed a priori.  Utilizing the derived solution in Theorem~\ref{Thm:solution}, the RCHP can operate in real time by directly mapping the LMP and renewable generation to hydrogen production and grid participation decisions, thereby eliminating the need to repeatedly solve the optimization problem \eqref{eq:maxprofit}.

\section{Profitability and Capacity Matching}\label{sec:nameplate}
% !TEX root = LiEtal26AE_extend.tex
With the real-time operational scheduling established in Sec.~\ref{sec:optimalplan}, Sec.~\ref{sec:nameplate} addresses the system's profitability and capacity sizing. Evaluating these investment decisions requires considering operations over a multi-period horizon and accounting for the stochastic nature of LMPs and renewable generation. We specifically focus on the {\em expected  operating profit (OP)}, defined as the expected gross profit minus other expenses beyond the cost of goods sold, including the fixed operating cost defined in (\ref{eq:CF}). Since there is no temporal coupling across scheduling periods, the optimal solution to the single-period profit maximization (\ref{eq:maxprofit}) directly applies to each interval within the evaluation horizon. Consequently, the multi-period expected gross profit of the RCHP can be accurately captured by aggregating the expected values of these optimized single-period outcomes. This bridges the operational and investment timescales, allowing us to answer the following questions in this section:
\ben
\item Will the RCHP achieve a non-negative expected operating profit and therefore be deemed profitable?
\item How do the nameplate capacities of the renewable and the electrolyzer, $(Q_{\mbox{\tiny\sf R}},Q_{\mbox{\tiny\sf H}})$, affect profitability?
\item Given a fixed cost budget, what are the optimal capacities for the RCHP's renewable generation and electrolyzer, $(Q_{\mbox{\tiny\sf R}},Q_{\mbox{\tiny\sf H}})$?
\een
The last question is particularly relevant in practice, as the definition of green hydrogen may require that electrolyzers and renewables be invested jointly \cite{EUCMDR:23}.

\subsection{Stochastic Profit Maximization}
Because renewable generation and LMPs are random processes, we formulate the stochastic profit maximization over an $n$-period horizon by taking the expectation over their joint trajectories. Building upon the temporal decoupling discussed above, it follows naturally that the single-period optimal production plan derived in Sec.~\ref{sec:optimalplan} serves as the optimal decision under each realization. Let $\Pbf^*_t(\pi^{\mbox{\tiny\sf LMP}}_t,\eta_t;Q_{\mbox{\tiny\sf R}},Q_{\mbox{\tiny\sf H}})$ be the solution to (\ref{eq:maxprofit}) provided in Theorem~\ref{Thm:solution}. The expected $n$-period operating profit, expressed as a function of the respective nameplate capacities, is then given by the expected maximum gross profit accumulated over $n$ periods, minus the amortized fixed costs.
\begin{align} \label{eq:CFn}
J^{\mbox{\tiny\sf OP}}_n(Q_{\mbox{\tiny\sf R}},Q_{\mbox{\tiny\sf H}})
:=& \sum_{t=1}^n
\mbbE\Big[J^{\mbox{\tiny\sf GP}}_\theta\Big(\Pbf^*_t(\pi^{\mbox{\tiny\sf LMP}}_t,\eta_t;Q_{\mbox{\tiny\sf R}},Q_{\mbox{\tiny\sf H}})\Big)\Big]\nn\\
&-(\alpha^{\mbox{\tiny\sf R}}_nQ_{\mbox{\tiny\sf R}}+\alpha^{\mbox{\tiny\sf H}}_nQ_{\mbox{\tiny\sf H}}),
\end{align}
where $(\alpha_n^{\mbox{\tiny\sf R}},\alpha_n^{\mbox{\tiny\sf H}})$ are the $n$-period amortized per-unit fixed costs computed from $(\alpha^{\mbox{\tiny\sf R}},\alpha^{\mbox{\tiny\sf H}})$ in (\ref{eq:CF}), with the computation details provided in Sec.~\ref{sec:Appendix_AmortizedCost}.

The structure of the optimal production plan allows us to derive a closed-form expression for the $n$-period operating profit via conditioning $(\pi_t^{\mbox{\tiny\sf LMP}},\eta_t)$ in regions $\reg{1}$-$\reg{4}$ in Fig.~\ref{fig:Optimalplan_M2}.

\begin{proposition}[Expected Operating Profit]\label{Prop:linear}
Denote the capacity ratio as $\kappa := Q_{\mbox{\tiny\sf H}}/Q_{\mbox{\tiny\sf R}}$. Let $P_{t,\kappa}^{(i)}$ be the probability that $(\pi_t^{\mbox{\tiny\sf LMP}},\eta_t) \in \reg{i}$, and
$\E_{t,\kappa}^{(i)}[\cdot]$ the conditional expectation operator (on $\reg{i}$) in interval $t$. The expected $n$-period operating profit is given by
\begin{equation}\label{eq:nCF1}
J^{\mbox{\tiny\sf OP}}_n(Q_{\mbox{\tiny\sf R}},Q_{\mbox{\tiny\sf H}})
=\big(\sum_{t=1}^nA^{\mbox{\tiny\sf R}}_{t,\kappa}-\alpha_n^{\mbox{\tiny\sf R}}\big)Q_{\mbox{\tiny\sf R}}+\big(\sum_{t=1}^nA^{\mbox{\tiny\sf H}}_{t,\kappa}-\alpha_n^{\mbox{\tiny\sf H}}\big)Q_{\mbox{\tiny\sf H}},
\end{equation}
where
\begin{align}
\nonumber
& \hspace{-2em} A^{\mbox{\tiny\sf R}}_{t,\kappa}=P_{t,\kappa}^{(1)}\Big((\tau^{\mbox{\tiny\sf IM}}_{\mbox{\tiny\sf REC}}+\tau^{\mbox{\tiny\sf R}})\E_{t,\kappa}^{(1)}[\eta_t]+\E_{t,\kappa}^{(1)}[\eta_t\pi_{t}^{\mbox{\tiny\sf LMP}}]\Big)\\ \nonumber
& +P_{t,\kappa}^{(2)}\Big(\big(\gamma(\pi^{\mbox{\tiny\sf H}}+\tau^{\mbox{\tiny\sf H}}-c^{\mbox{\tiny\sf W}})+\tau^{\mbox{\tiny\sf R}}\big)\E_{t,\kappa}^{(2)}[\eta_t] \Big)\\ \nonumber
& +P_{t,\kappa}^{(3)}\Big((\tau^{\mbox{\tiny\sf EX}}_{\mbox{\tiny\sf REC}}+\tau^{\mbox{\tiny\sf R}})\E_{t,\kappa}^{(3)}[\eta_t]+\E_{t,\kappa}^{(3)}[\eta_t\pi_t^{\mbox{\tiny\sf LMP}}] \Big)\\ \nonumber
& +P_{t,\kappa}^{(4)}\Big((\tau^{\mbox{\tiny\sf EX}}_{\mbox{\tiny\sf REC}}+\tau^{\mbox{\tiny\sf R}})\E_{t,\kappa}^{(4)}[\eta_t]+\E_{t,\kappa}^{(4)}[\eta_t\pi_t^{\mbox{\tiny\sf LMP}}] \Big),\\ \nonumber
& \hspace{-2em} A^{\mbox{\tiny\sf H}}_{t,\kappa}=P_{t,\kappa}^{(1)}\Big(\underline{\pi}^{\mbox{\tiny\sf LMP}}- \E_{t,\kappa}^{(1)}[\pi_t^{\mbox{\tiny\sf LMP}}]\Big)+P_{t,\kappa}^{(4)}\Big(\overline{\pi}^{\mbox{\tiny\sf LMP}}- \E_{t,\kappa}^{(4)}[\pi_t^{\mbox{\tiny\sf LMP}}]\Big).
\end{align}
\end{proposition}

A particularly useful application of Proposition~\ref{Prop:linear} is revenue and operating profit  forecasting.  By replacing theoretical probabilities and expectations with their respective empirical forms, we can estimate future profits based on historical or forecasted LMP and renewable trajectories. An example is given in Appendix \ref{sec:empirical_example}. Our numerical evaluations indicate that the accuracy of operating profit forecasts is comparable to that of renewable generation forecasts.

It is noteworthy that, under the proposed optimal production plan, the RCHP yields a higher expected operating profit than the configuration in which the electrolyzer and renewable energy source operate independently, as formalized in Proposition~\ref{Prop:colocation}.
\begin{proposition}[Colocation Profit Advantage]\label{Prop:colocation}
The expected OP of an RCHP exceeds the sum of expected OPs from separate operation of the electrolyzer and renewable source at identical capacities.
\end{proposition}
The proofs for Proposition~\ref{Prop:linear} and Proposition~\ref{Prop:colocation} are provided in the Appendix.

\subsection{Profitability and Matching Capacities}
We call an RCHP  {\em profitable} in an $n$-period operation if its (expected) operating profit is positive, $J^{\mbox{\tiny\sf OP}}_n(Q_{\mbox{\tiny\sf R}},Q_{\mbox{\tiny\sf H}})>0$. It is {\em in deficit} if $J^{\mbox{\tiny\sf OP}}_n(Q_{\mbox{\tiny\sf R}},Q_{\mbox{\tiny\sf H}})<0$, and {\em break-even} if $J^{\mbox{\tiny\sf OP}}_n(Q_{\mbox{\tiny\sf R}},Q_{\mbox{\tiny\sf H}})=0$.  This section characterizes the profitable, deficit, and break-even regions on the $Q_{\mbox{\tiny\sf H}}$-$Q_{\mbox{\tiny\sf R}}$ plane.  We are also interested in the optimal matching of the electrolyzer capacity $Q^*_{\mbox{\tiny\sf H}}$ to a given renewable capacity $Q_{\mbox{\tiny\sf R}}$.

\begin{theorem}[Profitability Characterization]\label{Thm:profitability} The nameplate capacity plane $Q_{\mbox{\tiny\sf H}}$ vs $Q_{\mbox{\tiny\sf R}}$ is partitioned into profitable and deficit regions with linear break-even boundaries.
\ben
\item The profitable (deficit) regions are convex cones with linearly growing (decreasing) operating profit away from the origin.
\item The break-even region is a union of linear lines.
\item The optimal matching of electrolyzer capacity $Q^*_{\mbox{\tiny\sf H}}$ to a given renewable capacity $Q_{\mbox{\tiny\sf R}}$ is linear, \ie  $Q^*_{\mbox{\tiny\sf H}}=\kappa Q_{\mbox{\tiny\sf R}}$ for some constant $\kappa$.
\een
\end{theorem}

See Fig.~\ref{fig:Optimality} (b) for an illustration, where the expected operating profit heatmap is partitioned by the black break-even lines, and the green dashed line represents the optimal electrolyzer capacities matched to given $Q_{\mbox{\tiny\sf R}}$'s. Notably, the deficit region may not be connected, as shown in Fig.~\ref{fig:Optimality} (b). The upper deficit region corresponds to RCHPs with high electrolyzer capacity but insufficient renewable generation capacity, whereas RCHPs in the lower region have high renewable generation capacity but insufficient electrolyzer capacity (especially for M0 and M1-c).

The intuition behind Theorem~\ref{Thm:profitability} follows from Proposition~\ref{Prop:linear}: for a fixed capacity ratio $\kappa$, the expected operating profit $J^{\mbox{\tiny\sf OP}}_n(Q_{\mbox{\tiny\sf R}},Q_{\mbox{\tiny\sf H}})$ is a linear homogeneous function of $(Q_{\mbox{\tiny\sf R}},Q_{\mbox{\tiny\sf H}})$. This scale-invariant property dictates that both the profitable and deficit regions are convex cones, and that break-even and optimal matching lines are linear. See Appendix~\ref{appendix:thm2} for the proof.

While the linear fixed cost assumption provides the scale-invariant properties in Theorem~\ref{Thm:profitability}, real-world deployments may exhibit nonlinearities in the fixed costs. Our analytical framework can readily accommodate such complexities. The real-time optimal scheduling policy remains unchanged, as it relies solely on operational marginal costs and volatile market prices rather than sunk fixed costs; therefore, we only need to adjust the amortized cost calculations. In Appendix \ref{appendix:nonlinear_cost}, we extend this baseline by employing piecewise linear cost functions to model nonlinear capital expenditures. As demonstrated therein, while the linear break-even boundaries and capacity matching rule dynamically transition into piecewise linear frontiers, the underlying capacity optimization remains highly tractable and our core operational insights are completely robust.

\subsection{Optimal Nameplate Capacities}
Theorem~\ref{Thm:profitability} characterizes the impact of nameplate capacities on RCHP profitability, specifically addressing the optimal electrolyzer capacity matching for a given renewable source. This result is highly pertinent when a new electrolyzer is to be colocated with an existing renewable facility. Next, we consider the joint optimization of both hydrogen and renewable capacities---a problem that arises when electrolyzers are integrated with new renewable installations. 

We formulate the joint optimization of electrolyzer and renewable capacities as a budget-constrained optimization problem:
\begin{equation}\label{eq:stochastic}
\begin{array}{lll}
\underset{(Q_{\mbox{\tiny\sf R}},Q_{\mbox{\tiny\sf H}})}{\rm maximize}&
J^{\mbox{\tiny\sf OP}}_n(Q_{\mbox{\tiny\sf R}},Q_{\mbox{\tiny\sf H}})\\
    \mbox{subject to}  & \alpha^{\mbox{\tiny\sf R}}_nQ_{\mbox{\tiny\sf R}}+\alpha^{\mbox{\tiny\sf H}}_nQ_{\mbox{\tiny\sf H}}=B_n,\\
    & Q_{\mbox{\tiny\sf R}},Q_{\mbox{\tiny\sf H}} \geq 0,\\
\end{array}
\end{equation}
where $B_n$ is the budget for the RCHP's amortized fixed cost over $n$ periods.

Theorem~\ref{Thm:nameplate} states the necessary condition for the optimality of \eqref{eq:stochastic}, and its proof can be found in Appendix~\ref{appendix_thm3}.

\begin{theorem}[Optimal Nameplate Capacity]\label{Thm:nameplate}
The optimal RCHP nameplate capacity values $(Q^*_{\mbox{\tiny\sf R}},Q^*_{\mbox{\tiny\sf H}})$ satisfy
\begin{equation}\label{eq:optcond}
    \frac{\sum_{t=1}^n A^{\mbox{\tiny\sf H}}_{t,Q^*_{\mbox{\tiny\sf H}}/Q^*_{\mbox{\tiny\sf R}}}}{\sum_{t=1}^n  A^{\mbox{\tiny\sf R}}_{t,Q^*_{\mbox{\tiny\sf H}}/Q^*_{\mbox{\tiny\sf R}}}}=\frac{\alpha_n^{\mbox{\tiny\sf H}}}{\alpha_n^{\mbox{\tiny\sf R}}},\quad \alpha^{\mbox{\tiny\sf R}}_nQ^*_{\mbox{\tiny\sf R}}+\alpha^{\mbox{\tiny\sf H}}_nQ^*_{\mbox{\tiny\sf H}}=B_n.
\end{equation}
\end{theorem}
Within the set of RCHP nameplate capacity values $(Q_{\mbox{\tiny\sf R}},Q_{\mbox{\tiny\sf H}})$ that satisfy the budget constraint, we seek a solution where the corresponding ratio $\sum_{t=1}^n  A^{\mbox{\tiny\sf H}}_{t,\kappa}/\sum_{t=1}^n  A^{\mbox{\tiny\sf R}}_{t,\kappa}$ matches $\alpha_n^{\mbox{\tiny\sf H}}/\alpha_n^{\mbox{\tiny\sf R}}$. Since this ratio monotonically decreases as $\kappa$ increases, the optimal nameplate capacities $(Q^*_{\mbox{\tiny\sf R}},Q^*_{\mbox{\tiny\sf H}})$ can be efficiently determined using a bisection search algorithm.

\section{Numerical Study}\label{sec:simulation}
% !TEX root = LiEtal26AE_extend.tex
\subsection{RCHP Profitability Evaluation}
We considered an RCHP  in the Central Zone (Zone C) of New York State. The renewable energy capacity factor profile utilized was derived from the 2023-2042 System \& Resource Outlook Data Document, which provided simulated hourly production profiles for land-based wind and solar resources across NYISO zones \cite{NYISO:TheOutlook}. The real-time electricity price data were collected from NYISO's Decision Support System \cite{NYISO:Data} with a 5-minute resolution. Due to the hourly granularity of the renewable generation dataset, RCHP operational decisions were modeled on an hourly basis, and the hourly LMPs were computed as the mean of the 5-minute intervals. Missing values in both datasets were addressed using linear interpolation. Other parameters, including credits, investment costs, and RCHP operational characteristics, are provided in Table \ref{tab:para}.
\begin{table}[!t]
\caption{\small Model parameters \cite{Glenk_Reichelstein_2019, itc_ptc_cheat_sheet, NRELreport, NREL91775, nyserda2024, DOE2024, DOEpvcost}. \label{tab:para}}
\centering
\scriptsize
\begin{tabular}{l l}
\toprule
% Annual interest rate for investment loan, $r$ &  0.564 \%\\
% \hline
Electrolyzer efficiency factor, $\gamma$ &  0.019 kg/kWh\\[0.5pt]
% \hline
Fixed annual operating cost for electrolyzer, $\alpha^{\mbox{\tiny\sf H}}$ & 101.25 \$/kW\\[0.5pt]
% \hline
Fixed annual operating cost for renewable plant, $\alpha^{\mbox{\tiny\sf R}}$ & 85.50 \$/kW\\[0.5pt]
% \hline
Green hydrogen credit, $\tau^{\mbox{\tiny\sf H}}$& 3.00 \$/kg\\[0.5pt]
% \hline
% Lifespan of RCHP, $Y$ & 30 years\\
% \hline
Renewable production tax credit, $\tau^{\mbox{\tiny\sf R}}$ & 27.50 \$/MWh\\[0.5pt]
% \hline
REC price for exported renewable, $\tau^{\mbox{\tiny\sf EX}}_{\mbox{\tiny\sf REC}}$ & 10.00 \$/MWh\\[0.5pt]
% \hline
REC price for imported renewable, $\tau^{\mbox{\tiny\sf IM}}_{\mbox{\tiny\sf REC}}$ & 31.80 \$/MWh\\[0.5pt]
% \hline
% System price of electrolyzer capacity, $\alpha_H$ & 2,500 \$/kW\\
% \hline
% System price of solar photovoltaic capacity, $\alpha_R$ & 980 \$/kW\\
% \hline
% System price of wind turbine capacity, $\alpha_R$ & 1,968 \$/kW\\
% \hline
Non-electricity marginal cost  of hydrogen, $c^{\mbox{\tiny\sf W}}$ & 0.10 \$/kg\\[0.5pt]
\bottomrule
\end{tabular}
\vspace{-1em}
\end{table}

We compared the annual operating profit of an RCHP with $Q_{\mbox{\tiny\sf H}}=20$ MW and $Q_{\mbox{\tiny\sf R}}=45$ MW—calculated using the proposed M2 model and its corresponding optimal operation plan—with results obtained from the models in \cite{Glenk_Reichelstein_2019} and \cite{lesniak2024}.
To evaluate the performance against input uncertainties, all profits were evaluated using 11 years of historical data (2012–2022) across varying hydrogen prices and renewable sources (solar and wind), as summarized in Table \ref{tab:comp}. By reporting the mean annual operating profit alongside its standard deviation, the table illustrates the inter-annual variability of each method. Furthermore, paired t-tests confirm that the economic advantage of our proposed model is highly statistically significant across all tested scenarios.
Under the prosumer model and with both renewable and green hydrogen credits integrated into the optimization, our approach yielded the highest operating profit through market participation. In contrast, other studies do not account for the bidirectional electricity market participation of the RCHP or overlook revenues from environmental subsidies. Moreover, the optimization of RCHP operation in \cite{lesniak2024} neglects the variable costs associated with hydrogen production, resulting in operational decisions that are suboptimal in practice.

\begin{table}[!t]
\caption{\small Cross-model comparison of annual operating profit (2012--2022). All values are in \$$10^6$.} \label{tab:comp}
\centering
\footnotesize
\begin{tabular}{c c c c}
\toprule
Method & Ref. \cite{Glenk_Reichelstein_2019} & Ref. \cite{lesniak2024} & This work \\
\midrule
$\pi^{\mbox{\tiny\sf H}}=\$1$/kg, solar & $2.58\pm 0.64$ & $1.30\pm 0.76$ & $4.60\pm 0.43$\rlap{\textsuperscript{*}} \\
% \hline
$\pi^{\mbox{\tiny\sf H}}=\$1$/kg, wind & $5.73\pm 0.84$ & $3.89\pm 0.89$ & $6.93\pm 0.66$\rlap{\textsuperscript{*}}  \\
% \hline
$\pi^{\mbox{\tiny\sf H}}=\$4$/kg, solar & $5.93\pm 0.52$ & $5.87\pm 0.46$ & $14.05\pm 0.68$\rlap{\textsuperscript{*}} \\
% \hline
$\pi^{\mbox{\tiny\sf H}}=\$4$/kg, wind & $11.06\pm 0.82$ & $11.00\pm 0.78$ & $16.39\pm 0.70$\rlap{\textsuperscript{*}}  \\
\bottomrule
\end{tabular}
{\raggedright \small \textsuperscript{*} Indicates statistical significance ($p < 0.001$) compared to both reference models via a paired t-test.\par}
\vspace{-1.5em}
\end{table}

\subsection{Effects of Renewable Generation}
To illustrate the operational and economic characteristics of the RCHP under the proposed method, Table \ref{tab:breakdown} presents the yearly revenue breakdown for the (45 MW, 20 MW) RCHP across different market models, all under the same 2022 electricity price and solar/wind generation realizations. The hydrogen price was set at \$4/kg.

Our experiment demonstrated the significance of using grid-imported renewable.  From Table \ref{tab:breakdown}, the standalone and consumer models had the identical colocated renewable utilization. M1-c was more profitable due to its ability to use grid-imported renewable. Likewise, the producer and prosumer models also had identical colocated renewable utilization. Again, M2 was more profitable because M2 used grid-imported renewable.

\begin{table*}[!t]
    \centering
    \footnotesize
    \renewcommand{\arraystretch}{1.5}
    \setlength{\tabcolsep}{2.5pt}
    \caption{\small RCHP revenue breakdown in 2022.}
    \label{tab:breakdown}
    \begin{tabular}{l c c c c c c c c }
        \toprule
        & \multicolumn{2}{c}{M0: Standalone}
        & \multicolumn{2}{c}{M1-p: Producer}  & \multicolumn{2}{c}{M1-c: Consumer} & \multicolumn{2}{c}{M2: Prosumer} \\
        \midrule
        Renewable Type & Solar & Wind & Solar & Wind & Solar & Wind & Solar & Wind \\
        % \hline
        Total renewable generation ( MWh) & 0.9043$\times 10^5$ & 1.2203$\times 10^5$ & 0.9043$\times 10^5$ & 1.2203$\times 10^5$ & 0.9043$\times 10^5$ & 1.2203$\times 10^5$ & 0.9043$\times 10^5$ & 1.2203$\times 10^5$\\
        % \hline
        Renewable in hydrogen production (\%) & 68.34 & 79.90 & 63.55 & 75.72 & 68.34 & 79.90 & 63.55 & 75.72 \\
        % \hline
        Hydrogen produced (kg) & 1.1741$\times 10^6$ & 1.8525$\times 10^6$ & 1.0918$\times 10^6$ & 1.7556$\times 10^6$ & 3.1349$\times 10^6$ & 3.1381$\times 10^6$ & 3.0525$\times 10^6$ & 3.0412$\times 10^6$\\
        % \hline
        Revenue from hydrogen sales (\$) & 8.2187$\times 10^6$ & 1.2968$\times 10^7$ & 7.6425$\times 10^6$ & 1.2289$\times 10^7$ & 2.1944$\times 10^7$ & 2.1967$\times 10^7$ & 2.1368$\times 10^7$ & 2.1288$\times 10^7$\\
        % \hline
        Renewable sold in the market (\%) & 0 & 0 & 36.45 & 24.28 & 0 & 0 & 36.45 & 24.28 \\
        % \hline
        Revenue from renewable sales (\$) & 0 & 0 & 2.7820$\times 10^6$ & 2.1301$\times 10^6$ & 0 & 0 & 2.7820$\times 10^6$ & 2.1301$\times 10^6$ \\
        % \hline
        Renewable curtailed (\%) & 31.66 & 20.10 & 0 & 0 & 31.66 & 20.10 & 0 & 0 \\
        % \hline
        Revenue lost due to curtailment (\$) & 1.9479$\times 10^6$ & 1.1113$\times 10^6$ & 0 & 0 & 1.9479$\times 10^6$ & 1.1113$\times 10^6$ & 0 & 0 \\
        % \hline
        Annual operating profit (\$) & 4.7155$\times 10^6$ & 1.0266$\times 10^7$ & 6.9295$\times 10^6$ & 1.1727$\times 10^7$ & 1.0662$\times 10^7$ & 1.3597$\times 10^7$ & 1.2876$\times 10^7$ & 1.5058$\times 10^7$ \\
        \bottomrule
    \end{tabular}
\end{table*}
\begin{figure*}[!htb]
    \centering
    \includegraphics[scale=0.56]{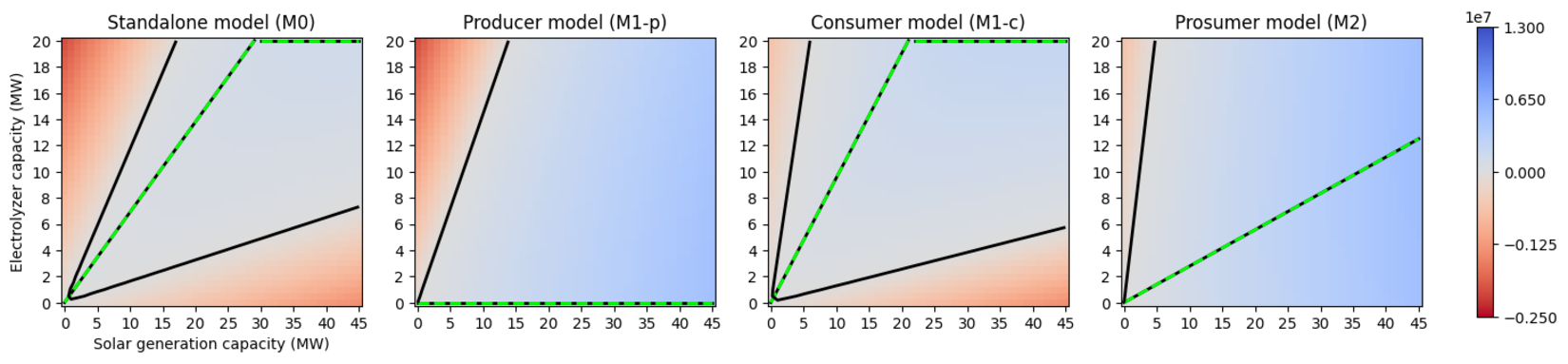}
    \vspace{-0.5em}
    \includegraphics[scale=0.56]{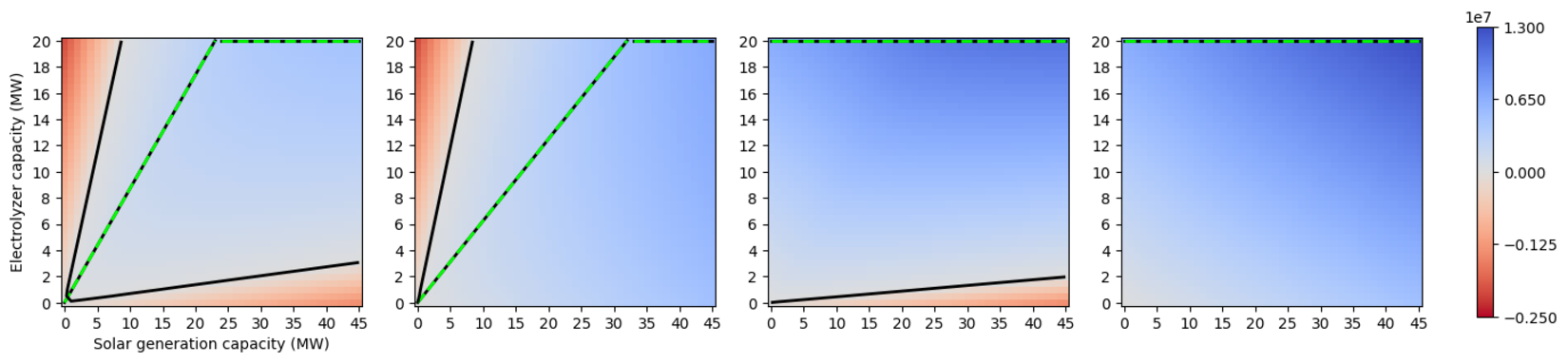}
    \caption{\scriptsize Annual operating profit in 2022 as a function of solar generation nameplate capacity (x-axis) and electrolyzer nameplate capacity (y-axis). Solid black: Break-even line. Green dashed: Optimal electrolyzer nameplate capacity as a function of solar generation nameplate capacity. (Top: hydrogen price of \$1/kg; bottom: \$4/kg.)}
    \label{fig:Heatmap}
    \vspace{-2em}
\end{figure*}

The discrepancy between the two colocation cases primarily arose from differences in renewable generation profiles. The wind RCHP had a higher average capacity factor of 0.310 compared to 0.229 for the solar RCHP, resulting in greater revenue from both hydrogen production and renewable electricity sales. However, as shown in Table \ref{tab:breakdown}, this advantage was less pronounced under M1-c and M2.

Under M0, the concentrated output peaks of solar generation frequently exceeded the hydrogen production capacity, leading to more frequent and severe curtailment compared to wind generation. Similarly, under M1-p, a greater portion of solar electricity exceeding electrolyzer capacity was sold during high-solar generation intervals, whereas the wind-colocated producer had greater potential to produce hydrogen across different periods. 
% This difference in renewable characteristics also explains the more significant profit gaps in the solar colocation setup in Fig. \ref{fig:Profit_PH}.
In contrast, under M1-c and M2, grid electricity imports compensated for the solar renewable shortage, making the revenue from hydrogen sales relatively similar between solar and wind setups. Besides, the covariance between renewable generation and electricity prices indicates that solar generation peaks aligned more closely with high electricity price intervals, allowing the solar RCHP to generate higher revenue from renewable sales.

\subsection{Effects of Nameplate Capacities}
We analyzed the effects of renewable and electrolyzer nameplate values $(Q_{\mbox{\tiny\sf R}},Q_{\mbox{\tiny\sf H}})$ on RCHP profitability. Fig. \ref{fig:Heatmap} illustrates the annual operating profit in 2022 as a function of the renewable plant (solar) capapcity $Q_{\mbox{\tiny\sf R}}$ (MW) on the x-axis and the electrolyzer capacity $Q_{\mbox{\tiny\sf H}}$ (MW) on the y-axis. The heatmaps depict outcomes for the four RCHP models at two hydrogen selling prices: a low price of \$1/kg (top row), and the prevailing price of \$4/kg (bottom row). The results for higher hydrogen prices closely resemble the \$4/kg case.

When the capacity parameter pair $(Q_{\mbox{\tiny\sf R}},Q_{\mbox{\tiny\sf H}})$ was set in the orange-red regions, the RCHP operated at a deficit due to mismatches between renewable and electrolyzer capacities. For instance, in the orange-red triangles on the upper left side of the heatmaps, where the renewable capacity was low, we observe that as the electrolyzer capacity increased, the fixed operating cost rose, and the mismatch became more pronounced, leading to a larger deficit.

In the blue regions, bounded by black break-even lines, the RCHP annual operation profit was non-negative. As shown in Fig. \ref{fig:Heatmap}, higher hydrogen prices expanded the profitable region across all four market participation models.

The green dashed lines in the blue regions represent the optimal electrolyzer capacities for the given renewable nameplate values. The slope of each green dashed line is influenced by market parameters, including the hydrogen price, credits, and variable cost, as well as the distribution of electricity prices and renewable capacity factors. From the top to the bottom row, the slope of the optimal electrolyzer capacity lines increased for each model, as higher hydrogen price made hydrogen sales more profitable, justifying investment in a larger electrolyzer.
% The slope of each green dashed line indicates the efficacy of hydrogen production; the flatter the slope, the better the utilization of the renewables at fixed electrolyzer input capacity. For example, consider the two left-most lower heatmaps on the bottom row of Fig. II. The green-dashed line under M1 is flatter than that under M0. At the electrolyzer capacity of 10MW, the optimally matched renewable capacity is approximately 16 MW under M0 and 22 MW under M1, implying that a 10 MW electrolyzer is suitable for renewables up to 22 MW under M1 and only 10 MW under M0.

Note that the slope of the green dashed line in the top row of Fig. \ref{fig:Heatmap} under M1-p is zero. At the hydrogen price of \$1/kg, the zero optimal electrolyzer capacity implies that investing in an electrolyzer and producing hydrogen was less profitable than exporting all renewables to the grid.
% In our parameter setting, hydrogen production becomes more profitable than exporting all renewables to the grid only when the hydrogen selling price exceeds \$1.75/kg. 

\subsection{Effects of Hydrogen Price}
We examined the impact of hydrogen prices $\pi^{\mbox{\tiny\sf H}}$ on RCHP's operating profit under different participation models using data from 2012-2022. The left panel of Fig. \ref{fig:Profit_PH} represents wind RCHP, while the right panel corresponds to solar RCHP. Both configurations employ the same electrolyzer capacity ($Q_{\mbox{\tiny\sf H}}=20$ MW) and renewable capacity ($Q_{\mbox{\tiny\sf R}}=45$ MW), thus the performance differences between the wind and solar RCHPs only came from the statistical characteristics of the respective renewable sources.
\begin{figure}[!htb]
    \centering
    \includegraphics[scale=0.28]{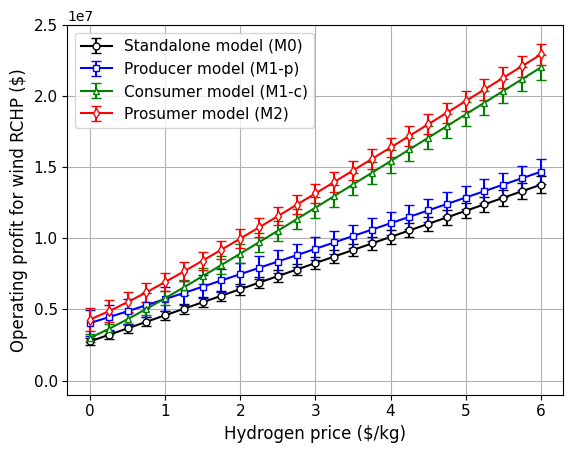}
    \includegraphics[scale=0.28]{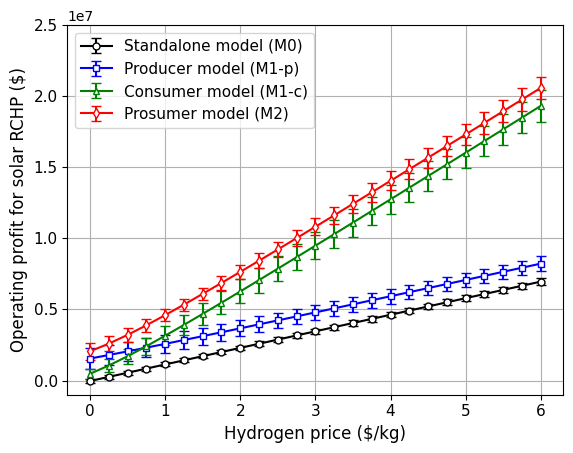}
    \caption{\scriptsize Mean annual operating profit of the RCHP under varying hydrogen prices (2012-2022), with error bars indicating inter-annual variability. (Left: (45 MW, 20 MW) wind-colocated hydrogen producer; right: (45 MW, 20 MW) solar-colocated hydrogen producer.)}
    \label{fig:Profit_PH}
    \vspace{-2em}
\end{figure}

The annual operating profit for the two types of renewable had similar characteristics.  First, the prosumer model M2 yielded the highest operating profit, and the standalone model M0 the lowest.
At the prevailing hydrogen price range of \$3-4.5/kg, the percentage gains of M1-p, M1-c, and M2 over M0 were significant. In wind colocation scenarios, these models achieved gains of up to 11.12\%, 57.51\%, and 65.47\%, respectively. The gains from solar colocation were more substantial, reaching 35.19\%, 181.26\%, and 215.95\%, respectively.

Second, both figures showed opposite trends for the producer and consumer models. As hydrogen price increased, M1-p trended away from the prosumer model M2 toward the standalone model M0, whereas M1-c trended away from the standalone model M0 to the prosumer model, which has simple explanations. As hydrogen price decreased toward zero, the economic value of hydrogen was diminishing. Both M1-p and M2 exported and profited from renewable the same way, while M0 and M1-c similarly suffered from the inability to export renewable.  As the economic value of hydrogen grew with its price,  M2 and M1-c benefited from grid-imported renewable while M0 and M1-p could not. The profit gaps between M2 and M1-c, and between M0 and M1-p, were due to high renewable cases where M1-c and M0 had to curtail renewable beyond the electrolyzer capacity, while M2 and M1-p could export the surplus renewable to the grid.

\subsection{Effects of Colocation and Subsidies}
To address the influence of the policy environment on economic outcomes, we conducted a sensitivity analysis by varying the environmental subsidy factor from 1.0 (current favorable conditions) down to 0.0. This factor proportionally scales all environmental credit values, including REC prices. A factor of 0 simulates a complete phase-out of subsidies, such as the expiration of green hydrogen credits (\eg the 45V tax credit introduced by the Inflation Reduction Act \cite{DOE2023}) and the absence of REC values.

Fig.~\ref{fig:Profit_Subsidy} illustrates the annual operating profit of the RCHP under varying environmental subsidy factors. We compared the operating profit achieved under the prosumer model with that in the non-colocation configuration, where the electrolyzer and the renewable generator operated independently without co-optimization. When the subsidy factor was zero, no financial incentive was provided for renewable electricity or green hydrogen, making green hydrogen production economically equivalent to purchasing grid electricity for hydrogen production while selling renewable output to the grid independently. As the subsidy factor increased, the profit gap between the prosumer model and the non-colocation model widened. This trend highlighted the proposed RCHP production plan's capacity to leverage environmental subsidies through co-optimization, leading to significantly higher profitability under generous policy support.

Given that overall profitability understandably declines as subsidies phase out, we further examined the relative performance among different market participation models. Fig.~\ref{fig:M2_additional_Subsidy} illustrates the additional operating profit generated by the M2 model compared to the M1-c and M1-p models under less favorable policy conditions. Notably, the M2 model consistently maintained a clear and strictly positive economic advantage over both M1-c and M1-p across all subsidy levels.

To identify the conditions under which the system remains economically viable in the complete absence of subsidies, \ie subsidy factor = 0, we evaluated the minimum hydrogen selling price required to achieve break-even. 
Our results indicated that the M2 prosumer model could sustain profitability without any policy support at a hydrogen price of only \$2.42/kg for the wind-colocated system and  \$2.45/kg for the solar-colocated system. In contrast, the M1 models demanded higher prices to survive: the producer model (M1-p) required \$3.22/kg for wind colocation and \$4.74/kg for solar, while the consumer model (M1-c) required \$2.89/kg for wind and \$3.25/kg for solar. These price thresholds provided essential insights for evaluating long-term project viability and risk management in a post-subsidy era.
\vspace{-0.5em}
\begin{figure}[!htb]
    \centering
    \includegraphics[scale=0.25]{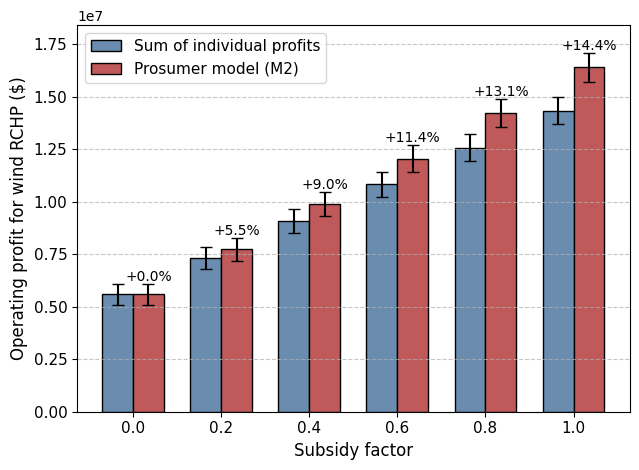}
    \includegraphics[scale=0.25]{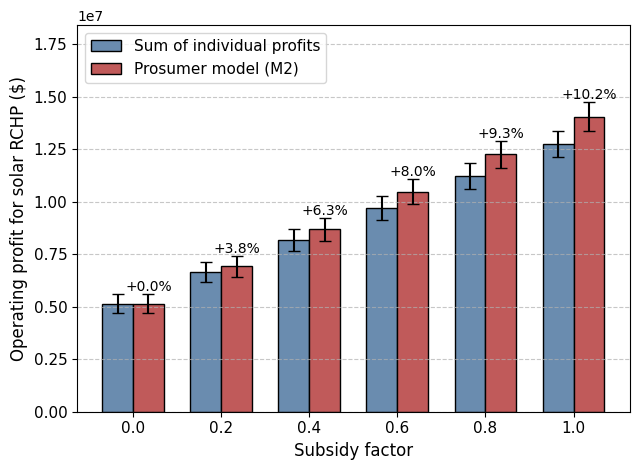}
    \caption{\scriptsize Mean annual operating profit of the RCHP under varying environmental subsidy factors (2012-2022), with error bars indicating inter-annual variability. (Left: (45 MW, 20 MW) wind-colocated hydrogen producer; right: (45 MW, 20 MW) solar-colocated hydrogen producer.)}
    \label{fig:Profit_Subsidy}
    \vspace{-1em}
\end{figure}
\begin{figure}[!htb]
    \centering
    \includegraphics[scale=0.25]{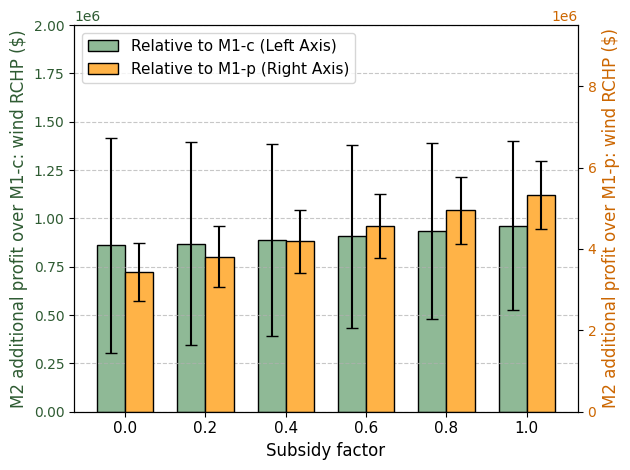}
    \includegraphics[scale=0.25]{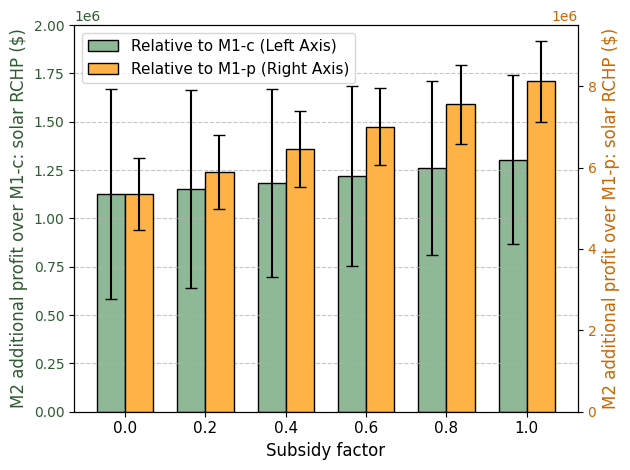}
    \caption{\scriptsize Mean additional annual operating profit of the M2 model relative to M1-c and M1-p under varying environmental subsidy factors (2012-2022), with error bars indicating inter-annual variability. (Left: (45 MW, 20 MW) wind-colocated hydrogen producer; right: (45 MW, 20 MW) solar-colocated hydrogen producer.)}
    \label{fig:M2_additional_Subsidy}
    \vspace{-1.5em}
\end{figure}

Beyond policy shifts, we also evaluated the robustness of our proposed models against environmental market uncertainties—specifically, REC price volatility. To achieve this, we conducted a Monte Carlo simulation. To isolate the impact of REC price volatility and avoid compounding uncertainties, we fixed the renewable generation profiles and real-time LMP inputs using data from the representative year of 2022, and performed a large-scale sampling of REC prices.

Recognizing that REC markets typically operate on longer trading cycles, we introduced a monthly variance to both the REC import and export prices. Specifically, we modeled the monthly REC prices as random variables drawn from a normal distribution, where the means ($\mu$) were set to the baseline REC import and export prices, as listed in Table~\ref{tab:para}. We scaled the relative standard deviation (RSD, defined as $\sigma/\mu$) to represent escalating market volatility. To preserve real-world market mechanics, we enforced strict boundary conditions: all generated REC prices were bounded to be non-negative, and the REC import price was constrained to be strictly greater than the export price at any given month. Under each volatility scenario, we generated 500 independent REC price paths and re-optimized the operations for both wind- and solar-colocated systems.
\begin{figure}[!htb]
    \centering
    \includegraphics[scale=0.25]{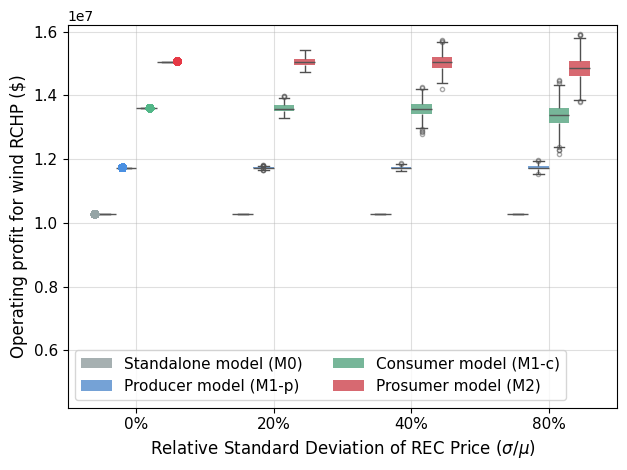}
    \includegraphics[scale=0.25]{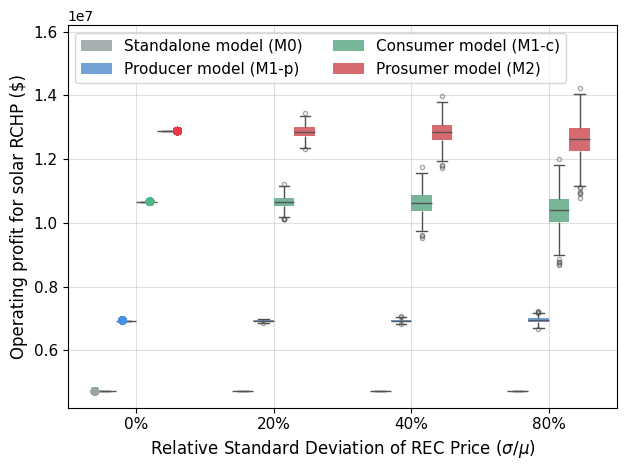}
    \caption{\scriptsize Annual operating profit distribution of the RCHP under varying REC price volatility based on 2022 data. (Left: (45 MW, 20 MW) wind-colocated hydrogen producer; right: (45 MW, 20 MW) solar-colocated hydrogen producer.)}
    \label{fig:Profit_REC}
    \vspace{-1.5em}
\end{figure}

As shown in Fig.~\ref{fig:Profit_REC}, the boxplots illustrate the empirical distribution of the RCHP's operating profits across different market participation models. The boxes demonstrate the core interquartile range (25th to 75th percentiles) of the profits, while the whiskers (representing standard statistical bounds) bound the robust operating profits under over 99\% of typical market conditions, isolating extreme tail-risk events as outliers.

Regarding the performance of the proposed models, while the standalone model M0 was completely unaffected by REC price variations, the grid-interactive models exhibited greater profit uncertainties as the RSD increased. For the consumer model M1-c, the overall expected profit experienced a slight downward trend due to the system's exposure to occasional spikes in REC purchasing costs. Conversely, any potential profit upside for the producer model M1-p driven by higher REC sales remained marginal. Despite the escalating market volatility, the prosumer model M2 consistently maintained the highest profitability among all configurations.

\subsection{Multi-ISO Simulations}\label{sec:multi-ISO}
To assess the operation and profitability of the RCHP across different regions, we conducted multi-ISO simulations. In addition to NYISO, we incorporated LMPs and renewable generation data from CAISO and MISO to determine RCHP's optimal real-time operational decisions and corresponding profits \cite{EIA_CAISO, EIA_MISO}. Fig. \ref{fig:MultiISO-CF} presents the operating profits in 2022 under a hydrogen price of \$4/kg, for deployments in these three regions and colocated with either solar or wind generation. Detailed revenue breakdowns are provided in the Appendix (Tables~\ref{tab:breakdown_CAISO}-\ref{tab:breakdown_MISO}).

Across all regions, model M2 achieved the highest operating profit among all market participation models, while also reducing profitability disparities between resources and regions. For an RCHP with fixed capacity, the greatest economic benefit was observed in MISO, where the average renewable generation level was the highest. The substantial revenue generated from selling abundant renewable energy in MISO also explains why, in this region, the RCHP earned higher profits under model M1-p than under M1-c. In contrast, the opposite trend was observed in NYISO and CAISO.

Although expected solar generation was higher in CAISO, the RCHP colocated with solar was more profitable in NYISO than in CAISO under models M1-c and M2. This is because the average electricity price in NYISO was significantly lower, allowing for cost-effective grid electricity purchases, which in turn reduced the cost of hydrogen production.

Fig. \ref{fig:MultiISO-renewable} illustrates the percentage allocations of onsite renewables generated by the RCHP across different regions and resources. As shown, wind resources generally exhibited higher utilization rates for hydrogen production across all regions compared to solar. This is because the concentrated output peaks of solar generation resulted in a higher proportion of curtailment and market sales of surplus renewable electricity.
\begin{figure}[htbp]
    \vspace{-0.5em}
    \centering
    \includegraphics[scale=0.26]{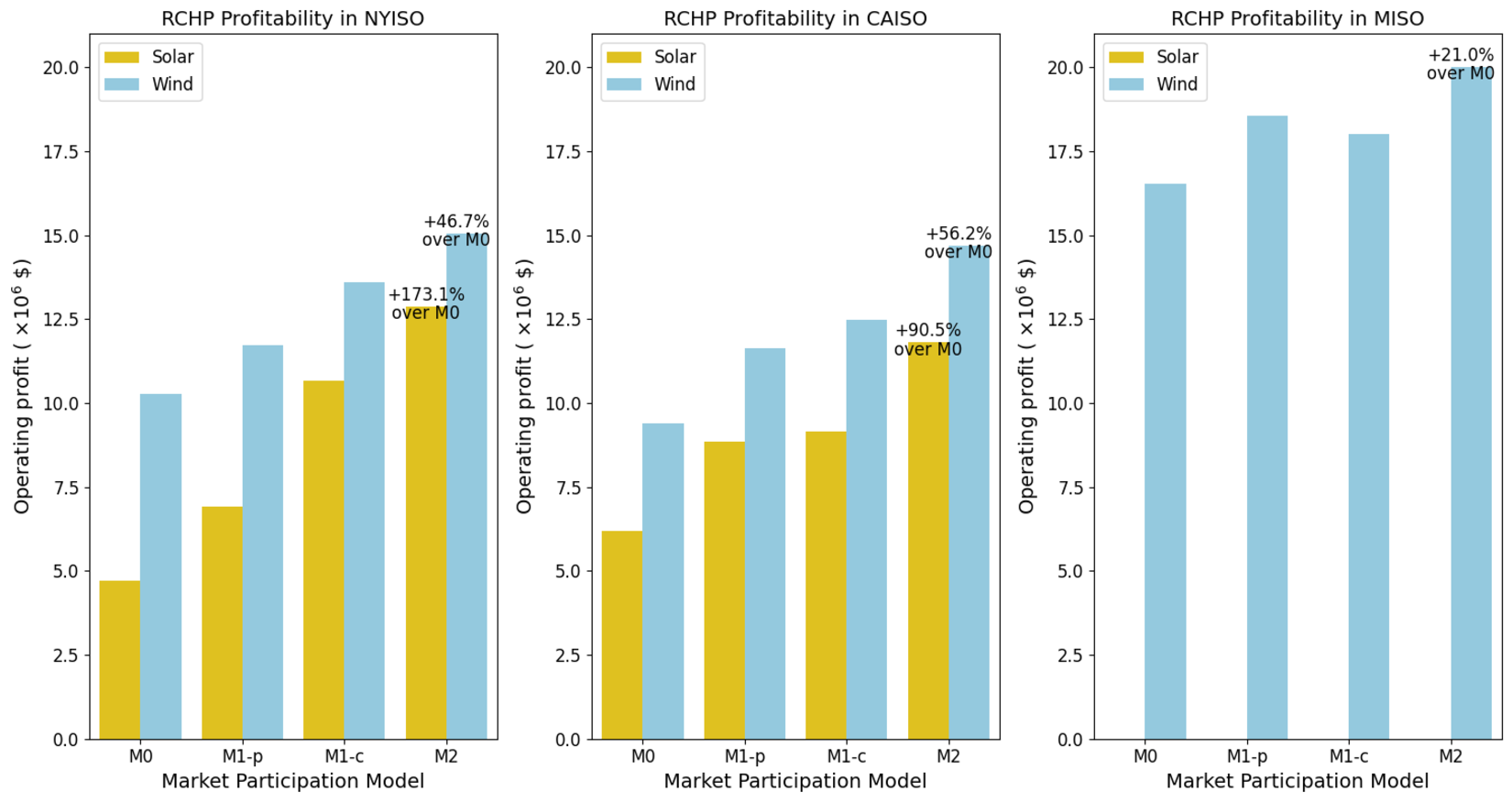}
    \caption{\scriptsize Annual operating profit of the (45MW, 20MW) RCHP in different regions in 2022. In NYISO, the mean capacity factors were 0.229 for solar and 0.310 for wind, while in CAISO, they were 0.252 for solar and 0.287 for wind. MISO had a mean wind capacity factor of 0.423. The mean electricity prices were \$0.055/kWh in NYISO, \$0.073/kWh in CAISO, and \$0.057/kWh in MISO.}
    \label{fig:MultiISO-CF}
    \vspace{-1.5em}
\end{figure}
\begin{figure}[htbp]
    \centering
    \includegraphics[scale=0.25]{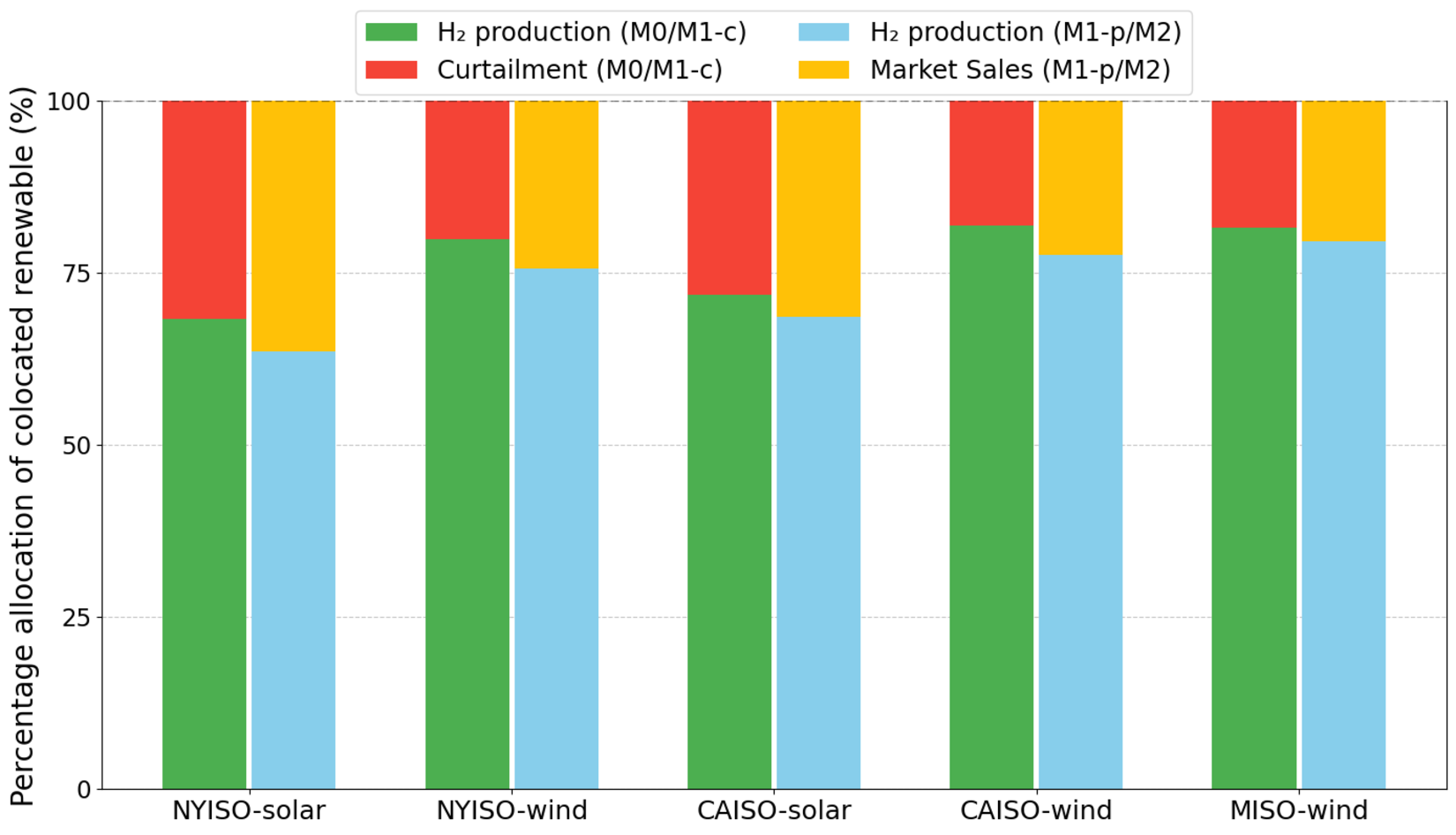}
    \caption{\scriptsize Percentage allocations of colocated renewables across different models, regions, and resources.}
    \label{fig:MultiISO-renewable}
\end{figure}

\section{Conclusion}\label{sec:conclusion}
The main contribution of this work is the methodology developed to analyze RCHP's operation and profitability, which is applicable to broader contexts, including integrated production and energy use in manufacturing, scheduling and energy management in data centers, as well as hydrogen production colocated with other generation assets. Specifically, we derive closed-form solutions for the RCHP's optimal production plan and provide analytical expressions for its operating profit. These results facilitate the rapid implementation of operational strategies and enable the assessment of RCHP profitability and optimal capacity sizing.
Empirical studies based on data from multiple ISOs show that RCHP profitability is sensitive to market prices, renewable generation profiles, and policy incentives. Optimal design choices, including electrolyzer and renewable capacity sizing, can enhance RCHP profitability. Furthermore, our cross-regional analyses reveal that different market characteristics favor different participation models, with wind generally achieving higher hydrogen utilization than solar. 

Future research will build upon this framework to explore several promising extensions. First, examining the impact of the hydrogen delivery system on RCHP profitability is important and warrants careful future study. Second, as fuel cell technologies advance toward higher efficiencies and lower costs, their integration can be considered to provide generation capability from stored hydrogen. Finally, incorporating risk-aversion metrics and relaxing the price-taker assumption to capture strategic market behavior will contribute to a more comprehensive profitability analysis.

% Use the style of numbering in square brackets.
% If nothing is used, default style will be taken.
%\begin{enumerate}[a)]
%\item 
%\item 
%\item 
%\end{enumerate}  

%\begin{itemize}
%\item 
%\item 
%\item 
%\end{itemize}  

%\begin{description}
%\item[]
%\item[] 
%\item[] 
%\end{description}  

% Uncomment and use as the case may be
%\begin{theorem} 
%\end{theorem}

% Uncomment and use as the case may be
%\begin{lemma} 
%\end{lemma}



% \bibliographystyle{cas-model2-names}
\bibliographystyle{elsarticle-num}

\begingroup
\small
\begin{spacing}{0.9}
\bibliography{refs}
\end{spacing}
\endgroup

\appendix
\vspace{-0.5em}
\section{Appendix}
% !TEX root = LiEtal26AE_extend.tex

\subsection{Numerical Example of Fuel Cell Energy Arbitrage}\label{appendix:fuel_cell}
We excluded fuel cells from the current real-time energy market framework primarily due to the prevailing economic realities of energy arbitrage. The round-trip efficiency of P2H2P systems is relatively low, making them economically uncompetitive for frequent cycling \cite{Zhang&etal:20FER}.

To justify deploying a fuel cell in the real-time energy market, the LMP must be exceptionally high to offset the opportunity cost of not selling the hydrogen directly to the market. Following the methodology in \cite{Campana&etal:25AE}, we assume a fuel cell system efficiency of 45\% and utilize the High Heating Value (HHV) of hydrogen (39.41 kWh/kg). Under these parameters, 1 kg of H$_2$ yields approximately 0.0177 MWh of electricity. If we consider a net profit of \$3/kg from direct hydrogen sales to external off-takers, the LMP would need to exceed \$169.2/MWh just to equalize the revenue from power generation with the lost profit from direct hydrogen sales. Historical ISO data indicates that such high-price events are too infrequent to justify the substantial capital investment and degradation costs associated with fuel cell arbitrage in the energy market. For instance, in the 2022 NYISO Central zone data, the LMP met or exceeded \$169.2/MWh during only 226 hours throughout the year, representing just 2.58\% of all time intervals.

\subsection{Ex-Post Validation of Finite Storage}\label{appendix:sto_limits}
To quantitatively justify our assumption on hydrogen storage and demonstrate the robustness of our model, we conducted an ex-post numerical validation. We simulated a practical operational scenario where the local hydrogen storage tank had a finite capacity. Consistent with standard industrial practices, the tank capacity was sized to accommodate 12 to 24 hours of the electrolyzer's maximum nominal production. In recent literature, a 24-hour maximum hydrogen production volume is widely considered a standard storage capacity baseline for hydrogen facilities \cite{Lu&etal:23, Saars&etal:24}.

Furthermore, in our numerical validation, the downstream hydrogen demand was modeled as a continuous, steady-state off-take. This formulation aligns with the stringent operational requirements of heavy industrial off-takers, which represent the primary consumers of large-scale green hydrogen \cite{ArmijoPhilibert:20IHE}. Specifically, the hourly demand was generated from a Gaussian distribution, where the mean was set to 95\% of the electrolyzer's maximum hourly production capacity, and the standard deviation was constrained to 10 kg/h to reflect the micro-fluctuations of a continuous base-load chemical process.

In this simulation, our derived threshold policy served as the baseline dispatch signal, subject only to an ex-post physical override, \ie halting hydrogen production solely when the tank hit maximum capacity.
%\vspace{-0.5em}
\begin{table}[!htb]
\caption{\small Ex-post Annual Profit Loss due to Finite Storage Constraints (Based on 2022 NYISO Data)} \label{tab:profit_loss}
    \centering
    \begin{tabular}{c c c c c}
    \toprule
     Scenario & M0& M1-c & M1-p& M2 \\
    \midrule
    12h storage, wind & 0.00\% & 2.30\% & 0.00\% & 1.41\% \\
    12h storage, solar & 0.00\% & 3.08\% & 0.00\% & 1.61\% \\
    24h storage, wind & 0.00\% & 1.86\% & 0.00\% & 1.07\% \\
    24h storage, solar & 0.00\% & 2.61\% & 0.00\% & 1.25\% \\
    \bottomrule
    \end{tabular}
\vspace{-1em}
\end{table}

As shown in Table \ref{tab:profit_loss}, deviations from our optimal schedule due to storage bottlenecks were rare, and the resulting operating profit loss compared to the unconstrained ideal case was marginal across all configurations.
This confirms that while storage limits introduce operational clipping, they do not invalidate our economic findings. The fundamental structure of our real-time threshold policy remains a highly robust and near-optimal dispatch baseline for practical engineering applications.

\subsection{Derivation of Amortized Per-Unit Fixed Costs}\label{sec:Appendix_AmortizedCost}
As discussed in Section~\ref{sec:formulation}, the fixed operating cost $C^{\mbox{\tiny\sf F}}$ of an RCHP is assumed to be a linear function of the renewable capacity $Q_{\mbox{\tiny\sf R}}$ and the electrolyzer capacity $Q_{\mbox{\tiny\sf H}}$, where the factors $\alpha^{\mbox{\tiny\sf R}}$ and $\alpha^{\mbox{\tiny\sf H}}$ represent the annual fixed operating costs per unit capacity of renewable and electrolyzer facilities, respectively.

To evaluate the RCHP's operating profit over $n$ periods, we define the amortized fixed costs $C_n^{\mbox{\tiny\sf F}}$ for the evaluation period.
% We assume that the renewable power plant and electrolyzer are financed through a fixed-rate loan with an annual interest rate $r$ and a total loan term of $Y$ years, which aligns with the facilities' lifespan.
Let $N$ denote the number of RCHP scheduling intervals per year. Then, the amortized fixed cost is given by
\begin{align}
    C_n^{\mbox{\tiny\sf F}}(Q_{\mbox{\tiny\sf R}}, Q_{\mbox{\tiny\sf H}}) &=  \frac{n}{N} \big( \alpha^{\mbox{\tiny\sf R}}Q_{\mbox{\tiny\sf R}} + \alpha^{\mbox{\tiny\sf H}}Q_{\mbox{\tiny\sf H}} \big) \nonumber \\[2pt]
    &= \alpha_n^{\mbox{\tiny\sf R}} Q_{\mbox{\tiny\sf R}} + \alpha_n^{\mbox{\tiny\sf H}} \kappa Q_{\mbox{\tiny\sf R}}.
\end{align}

\subsection{Sketch of the proof for Remark~\ref{rmk:convexity}}\label{sec:proof_rmk}
\begin{proof}
In the optimization problem \eqref{eq:maxprofit}, nonconvexity arises due to the bilinear constraint \eqref{eq:cons3}, which prevents simultaneous export and import of electricity. However, under the market condition where the REC purchase price strictly exceeds the REC selling price ($\tau_{\mbox{\tiny\sf REC}}^{\mbox{\tiny\sf IM}} > \tau_{\mbox{\tiny\sf REC}}^{\mbox{\tiny\sf EX}}$), any simultaneous power injection and withdrawal is economically suboptimal compared to its net equivalent. Specifically, replacing any simultaneous import and export with the equivalent net power exchange yields a strictly higher profit due to the REC price spread. Therefore, any optimal solution to the relaxed problem---where constraint \eqref{eq:cons3} is removed---will inherently satisfy $P^{\mbox{\tiny\sf IM}}_tP^{\mbox{\tiny\sf EX}}_t=0$. Consequently, the original non-convex problem can be exactly relaxed and solved as a convex program.
\end{proof}

\subsection{Proof of Theorem~\ref{Thm:solution} and Optimal Production Plan Including Negative LMP Scenarios}\label{sec:Appendix_OptimalPlan}
\begin{proof}[Proof of Theorem~\ref{Thm:solution}]
In the optimization problem \eqref{eq:maxprofit}, nonconvexity arises due to the bilinear constraint \eqref{eq:cons3}, which prevents simultaneous export and import of electricity. To address this, we decompose the problem into two cases: (1) $ P^{\mbox{\tiny\sf EX}}_{t}=0$ (no renewable electricity export), (2) $ P^{\mbox{\tiny\sf IM}}_{t}=0$ (no grid electricity import). Our approach involves solving the optimization problem separately for each case, formulating two linear programs (LPs). We then compare the optimal solutions from both cases to determine the globally optimal solution for \eqref{eq:maxprofit}.

(1) \underline{$ P^{\mbox{\tiny\sf EX}}_{t}=0$}: In this case, we assume that no renewable electricity generated by the RCHP is exported to the grid in time interval $t$. Substituting this condition into \eqref{eq:maxprofit} and excluding the term $\tau^{\mbox{\tiny\sf R}}\eta_tQ_{\mbox{\tiny\sf R}}$, which does not affect the operational decision, results in the following optimization.
\begin{align}\label{LP1} \nonumber
&\underset{\Pbf_t=(P^{\mbox{\tiny\sf H}}_{t}, P^{\mbox{\tiny\sf IM}}_t)}{\rm maximize}&& (\pi^{\mbox{\tiny\sf H}}+\tau^{\mbox{\tiny\sf H}}-c^{\mbox{\tiny\sf W}})(\gamma P^{\mbox{\tiny\sf H}}_{t})- (\pi^{\mbox{\tiny\sf LMP}}_t+\tau^{\mbox{\tiny\sf IM}}_{\mbox{\tiny\sf REC}})P_t^{\mbox{\tiny\sf IM}}\\ \nonumber
&\mbox{subject to}  && 0\le P^{\mbox{\tiny\sf H}}_t-P^{\mbox{\tiny\sf IM}}_t\le \eta_t Q_{\mbox{\tiny\sf R}},\\ \nonumber
&&& 0 \leq P^{\mbox{\tiny\sf H}}_{t} \leq Q_{\mbox{\tiny\sf H}},\\
& && 0 \leq  P^{\mbox{\tiny\sf IM}}_{t}  \leq Q_{\mbox{\tiny\sf H}}.
\end{align}
This LP yields the optimal solution $\mathbf{P}^{1*}_t=\Big[P^{\mbox{\tiny\sf H}*}_t,\ 0,\ P^{\mbox{\tiny\sf IM}*}_t\Big]$, subject to the constraint of no electricity export.
\begin{equation}\label{eq:sol_LP1}
\mathbf{P}^{1*}_t =
\begin{cases}
\Big[Q_{\mbox{\tiny\sf H}},\ 0,\ Q_{\mbox{\tiny\sf H}}\Big], & \pi^{\mbox{\tiny\sf LMP}}_t \le -\tau_{\mbox{\tiny\sf REC}}^{\mbox{\tiny\sf IM}};\\[4pt]
\Big[\min\{\eta_t Q_{\mbox{\tiny\sf R}}, Q_{\mbox{\tiny\sf H}}\},\ 0,\ 0\Big], & \pi^{\mbox{\tiny\sf LMP}}_t \ge \underline{\pi}^{\mbox{\tiny\sf LMP}};\\[4pt]
\Big[Q_{\mbox{\tiny\sf H}},\ 0,\ (Q_{\mbox{\tiny\sf H}}-\eta_t Q_{\mbox{\tiny\sf R}})^+\Big], & \text{otherwise.}
\end{cases}
\end{equation}
Furthermore, the corresponding objective value $V_t^{1*}$, is given by
\begin{equation}
V^{1*}_t =
\begin{cases}
(\underline{\pi}^{\mbox{\tiny\sf LMP}}-\pi^{\mbox{\tiny\sf LMP}}_t) Q_{\mbox{\tiny\sf H}},
\pi^{\mbox{\tiny\sf LMP}}_t \le -\tau_{\mbox{\tiny\sf REC}}^{\mbox{\tiny\sf IM}};\\[2pt]
\gamma(\pi^{\mbox{\tiny\sf H}}+\tau^{\mbox{\tiny\sf H}}-c^{\mbox{\tiny\sf W}}) Q_{\mbox{\tiny\sf H}},
\pi^{\mbox{\tiny\sf LMP}}_t >  -\tau_{\mbox{\tiny\sf REC}}^{\mbox{\tiny\sf IM}} ~\text{and}~  Q_{\mbox{\tiny\sf H}}\le \eta_t Q_{\mbox{\tiny\sf R}};\\[2pt]
\gamma(\pi^{\mbox{\tiny\sf H}}+\tau^{\mbox{\tiny\sf H}}-c^{\mbox{\tiny\sf W}}) \eta_t Q_{\mbox{\tiny\sf R}},
\pi^{\mbox{\tiny\sf LMP}}_t \ge \underline{\pi}^{\mbox{\tiny\sf LMP}} ~\text{and}~  Q_{\mbox{\tiny\sf H}}> \eta_t Q_{\mbox{\tiny\sf R}};\\[2pt]
(\underline{\pi}^{\mbox{\tiny\sf LMP}}-\pi^{\mbox{\tiny\sf LMP}}_t) Q_{\mbox{\tiny\sf H}}+(\pi^{\mbox{\tiny\sf LMP}}_t+\tau^{\mbox{\tiny\sf IM}}_{\mbox{\tiny\sf REC}})\eta_t Q_{\mbox{\tiny\sf R}}, \mbox{otherwise.}
\end{cases}
\end{equation}

(2) \underline{$P^{\mbox{\tiny\sf IM}}_t=0$}: In this case, we assume that electricity is not imported from the grid during time interval $t$. Similarly, we obtain the following LP.
\begin{align}\label{LP2} \nonumber
    &\underset{\Pbf_t=(P^{\mbox{\tiny\sf H}}_{t}, P^{\mbox{\tiny\sf EX}}_t)}{\rm maximize}&& (\pi^{\mbox{\tiny\sf H}}+\tau^{\mbox{\tiny\sf H}}-c^{\mbox{\tiny\sf W}})(\gamma P^{\mbox{\tiny\sf H}}_{t})+ (\pi^{\mbox{\tiny\sf LMP}}_t+\tau^{\mbox{\tiny\sf EX}}_{\mbox{\tiny\sf REC}})P_t^{\mbox{\tiny\sf EX}}\\ \nonumber
    &\mbox{subject to}  && 0\le P^{\mbox{\tiny\sf H}}_t+P^{\mbox{\tiny\sf EX}}_t\le \eta_t Q_{\mbox{\tiny\sf R}},\\ \nonumber
    &&& 0 \leq P^{\mbox{\tiny\sf H}}_{t} \leq Q_{\mbox{\tiny\sf H}},\\
    & && 0 \leq  P^{\mbox{\tiny\sf EX}}_{t}  \leq \eta_tQ_{\mbox{\tiny\sf R}}.
    \end{align}
The optimal solution, $\mathbf{P}^{2*}_t=\Big[P^{\mbox{\tiny\sf H}*}_t,\ P^{\mbox{\tiny\sf EX}*}_t,\ 0\Big]$, and the corresponding optimal value, $V_t^{2*}$, are also determined.
\begin{equation}\label{eq:sol_LP2}
\mathbf{P}^{2*}_t =
\begin{cases}
\Big[\min\{\eta_t Q_{\mbox{\tiny\sf R}}, Q_{\mbox{\tiny\sf H}}\},\ 0,\ 0\Big], &  \pi^{\mbox{\tiny\sf LMP}}_t \le -\tau_{\mbox{\tiny\sf REC}}^{\mbox{\tiny\sf EX}};\\[2pt]
\Big[0,\ \eta_t Q_{\mbox{\tiny\sf R}},\ 0\Big],& \pi^{\mbox{\tiny\sf LMP}}_t \ge \overline{\pi}^{\mbox{\tiny\sf LMP}};\\[2pt]
\Big[\min\{\eta_t Q_{\mbox{\tiny\sf R}}, Q_{\mbox{\tiny\sf H}}\},\ (\eta_t Q_{\mbox{\tiny\sf R}}-Q_{\mbox{\tiny\sf H}})^+,\ 0\Big], & \mbox{otherwise.}
 \end{cases}
\end{equation}
\begin{equation}
V^{2*}_t =
\begin{cases}
(\pi^{\mbox{\tiny\sf LMP}}_t+\tau^{\mbox{\tiny\sf EX}}_{\mbox{\tiny\sf REC}}) \eta_t Q_{\mbox{\tiny\sf R}},
\pi^{\mbox{\tiny\sf LMP}}_t \ge  \overline{\pi}^{\mbox{\tiny\sf LMP}};\\[2pt]
\gamma(\pi^{\mbox{\tiny\sf H}}+\tau^{\mbox{\tiny\sf H}}-c^{\mbox{\tiny\sf W}}) \eta_t Q_{\mbox{\tiny\sf R}},
\pi^{\mbox{\tiny\sf LMP}}_t < \overline{\pi}^{\mbox{\tiny\sf LMP}} ~\text{and}~Q_{\mbox{\tiny\sf H}}> \eta_t Q_{\mbox{\tiny\sf R}};\\[2pt]
\gamma(\pi^{\mbox{\tiny\sf H}}+\tau^{\mbox{\tiny\sf H}}-c^{\mbox{\tiny\sf W}}) Q_{\mbox{\tiny\sf H}},
\pi^{\mbox{\tiny\sf LMP}}_t \le -\tau_{\mbox{\tiny\sf REC}}^{\mbox{\tiny\sf EX}} ~\text{and}~  Q_{\mbox{\tiny\sf H}}\le \eta_t Q_{\mbox{\tiny\sf R}};\\[2pt]
(\overline{\pi}^{\mbox{\tiny\sf LMP}}-\pi^{\mbox{\tiny\sf LMP}}_t) Q_{\mbox{\tiny\sf H}}+(\pi^{\mbox{\tiny\sf LMP}}_t+\tau^{\mbox{\tiny\sf EX}}_{\mbox{\tiny\sf REC}})\eta_t Q_{\mbox{\tiny\sf R}}, \mbox{otherwise.}
\end{cases}
\end{equation}

Note that four electricity price thresholds determine the optimal solution: $-\tau_{\mbox{\tiny\sf REC}}^{\mbox{\tiny\sf IM}}$, $-\tau_{\mbox{\tiny\sf REC}}^{\mbox{\tiny\sf EX}}$, $\underline{\pi}^{\mbox{\tiny\sf LMP}}$, and $\overline{\pi}^{\mbox{\tiny\sf LMP}}$. Typically, the first two thresholds are negative, while the last two are positive, satisfying the ordering $-\tau_{\mbox{\tiny\sf REC}}^{\mbox{\tiny\sf IM}}<-\tau_{\mbox{\tiny\sf REC}}^{\mbox{\tiny\sf EX}}<\underline{\pi}^{\mbox{\tiny\sf LMP}}<\overline{\pi}^{\mbox{\tiny\sf LMP}}$. We derive the optimal solution under this assumption. Solutions for special parameter settings that result in a different ordering of these thresholds can be obtained by similar arguments.

If $\pi^{\mbox{\tiny\sf LMP}}_t\le -\tau_{\mbox{\tiny\sf REC}}^{\mbox{\tiny\sf IM}}$, then
\begin{align}\nn
V^{1*}_t=&(\underline{\pi}^{\mbox{\tiny\sf LMP}}-\pi^{\mbox{\tiny\sf LMP}}_t) Q_{\mbox{\tiny\sf H}}\\ \nn
=&\gamma(\pi^{\mbox{\tiny\sf H}}+\tau^{\mbox{\tiny\sf H}}-c^{\mbox{\tiny\sf w}})Q_{\mbox{\tiny\sf H}}-(\pi^{\mbox{\tiny\sf LMP}}_t+\tau_{\mbox{\tiny\sf REC}}^{\mbox{\tiny\sf IM}})Q_{\mbox{\tiny\sf H}}\\
\geq &\gamma(\pi^{\mbox{\tiny\sf H}}+\tau^{\mbox{\tiny\sf H}}-c^{\mbox{\tiny\sf w}})\min\{\eta_t Q_{\mbox{\tiny\sf R}}, Q_{\mbox{\tiny\sf H}}\} =V^{2*}_t,
\end{align}
indicating that the optimal solution for \eqref{eq:maxprofit} is given by $\mathbf{P}^{*}_t=\mathbf{P}^{1*}_t=\Big[Q_{\mbox{\tiny\sf H}},\ 0,\ Q_{\mbox{\tiny\sf H}}\Big]$.

If $-\tau_{\mbox{\tiny\sf REC}}^{\mbox{\tiny\sf IM}}<\pi^{\mbox{\tiny\sf LMP}}_t \le -\tau_{\mbox{\tiny\sf REC}}^{\mbox{\tiny\sf EX}}$, then for the case $Q_{\mbox{\tiny\sf H}}\le \eta_t Q_{\mbox{\tiny\sf R}}$, both \eqref{LP1} and \eqref{LP2} yield the same solution: $\mathbf{P}^{*}_t=\Big[Q_{\mbox{\tiny\sf H}},\ 0,\ 0\Big]$. However, when $Q_{\mbox{\tiny\sf H}}> \eta_t Q_{\mbox{\tiny\sf R}}$, we have
\begin{align}\nn
V^{1*}_t-V^{2*}_t=&(\underline{\pi}^{\mbox{\tiny\sf LMP}}-\pi^{\mbox{\tiny\sf LMP}}_t) Q_{\mbox{\tiny\sf H}}+(\pi^{\mbox{\tiny\sf LMP}}_t-\underline{\pi}^{\mbox{\tiny\sf LMP}})\eta_t Q_{\mbox{\tiny\sf R}}\\ \nn
=&(\underline{\pi}^{\mbox{\tiny\sf LMP}}-\pi^{\mbox{\tiny\sf LMP}}_t)(Q_{\mbox{\tiny\sf H}}-\eta_t Q_{\mbox{\tiny\sf R}})>0.
\end{align}
Thus, the optimal solution is $\mathbf{P}^{*}_t=\Big[Q_{\mbox{\tiny\sf H}},\ 0,\ Q_{\mbox{\tiny\sf H}}-\eta_t Q_{\mbox{\tiny\sf R}}\Big]$.

If $-\tau_{\mbox{\tiny\sf REC}}^{\mbox{\tiny\sf EX}}<\pi^{\mbox{\tiny\sf LMP}}_t \le \underline{\pi}^{\mbox{\tiny\sf LMP}}$, then for the case $Q_{\mbox{\tiny\sf H}}\le \eta_t Q_{\mbox{\tiny\sf R}}$,
\begin{align}\nn
    V^{1*}_t-V^{2*}_t=&-(\pi^{\mbox{\tiny\sf LMP}}_{t}+\tau^{\mbox{\tiny\sf EX}}_{\mbox{\tiny\sf REC}})(\eta_t Q_{\mbox{\tiny\sf R}}-Q_{\mbox{\tiny\sf H}})\le 0,
\end{align}
and we should adopt the optimal solution $\mathbf{P}^{*}_t$ as described in $\mathbf{P}^{2*}_t=\Big[Q_{\mbox{\tiny\sf H}},\ \eta_t Q_{\mbox{\tiny\sf R}}-Q_{\mbox{\tiny\sf H}},\ 0\Big]$. For the case $Q_{\mbox{\tiny\sf H}}> \eta_t Q_{\mbox{\tiny\sf R}}$, we obtain
\begin{align}\nn
    V^{1*}_t-V^{2*}_t=&(\underline{\pi}^{\mbox{\tiny\sf LMP}}-\pi^{\mbox{\tiny\sf LMP}}_t)(Q_{\mbox{\tiny\sf H}}-\eta_t Q_{\mbox{\tiny\sf R}})\ge 0.
\end{align}
Therefore, the optimal solution is given when $ P^{\mbox{\tiny\sf EX}}_{t}=0$, \ie $\mathbf{P}^{*}_t=\Big[Q_{\mbox{\tiny\sf H}},\ 0,\ Q_{\mbox{\tiny\sf H}}-\eta_t Q_{\mbox{\tiny\sf R}}\Big]$.

If $\underline{\pi}^{\mbox{\tiny\sf LMP}}<\pi^{\mbox{\tiny\sf LMP}}_t< \overline{\pi}^{\mbox{\tiny\sf LMP}}$, then for the case $Q_{\mbox{\tiny\sf H}}> \eta_t Q_{\mbox{\tiny\sf R}}$, both \eqref{LP1} and \eqref{LP2} yield the same solution: $\mathbf{P}^{*}_t=\Big[\eta_t Q_{\mbox{\tiny\sf R}},\ 0,\ 0\Big]$. Conversely, when $Q_{\mbox{\tiny\sf H}}\le \eta_t Q_{\mbox{\tiny\sf R}}$,
\begin{align}\nn
V^{1*}_t-V^{2*}_t=&-(\pi^{\mbox{\tiny\sf LMP}}_{t}+\tau^{\mbox{\tiny\sf EX}}_{\mbox{\tiny\sf REC}})(\eta_t Q_{\mbox{\tiny\sf R}}-Q_{\mbox{\tiny\sf H}})\le 0.
\end{align}
It follows that the optimal solution for \eqref{eq:maxprofit} is expressed as $\mathbf{P}^{*}_t=\Big[Q_{\mbox{\tiny\sf H}},\ \eta_t Q_{\mbox{\tiny\sf R}}-Q_{\mbox{\tiny\sf H}},\ 0\Big]$.

If $\pi^{\mbox{\tiny\sf LMP}}_t\ge \overline{\pi}^{\mbox{\tiny\sf LMP}}$, then
\begin{align}\nn
V^{2*}_t=&(\pi^{\mbox{\tiny\sf LMP}}_{t}+\tau^{\mbox{\tiny\sf EX}}_{\mbox{\tiny\sf REC}})\eta_t Q_{\mbox{\tiny\sf R}}\\ \nn
=&\gamma(\pi^{\mbox{\tiny\sf H}}+\tau^{\mbox{\tiny\sf H}}-c^{\mbox{\tiny\sf W}})\eta_t Q_{\mbox{\tiny\sf R}}+(\pi^{\mbox{\tiny\sf LMP}}_t-\overline{\pi}^{\mbox{\tiny\sf LMP}})\eta_t Q_{\mbox{\tiny\sf R}}\\
\ge&\gamma(\pi^{\mbox{\tiny\sf H}}+\tau^{\mbox{\tiny\sf H}}-c^{\mbox{\tiny\sf W}})\min\{\eta_t Q_{\mbox{\tiny\sf R}}, Q_{\mbox{\tiny\sf H}}\}=V^{1*}_t,
\end{align}
implying that $\mathbf{P}^{*}_t=\Big[0,\ \eta_t Q_{\mbox{\tiny\sf R}},\ 0\Big]$.

Combining all these results, we derive the closed-form solution for the original optimization problem \eqref{eq:maxprofit}. The solution under positive LMPs is expressed in the form of \eqref{eq:optsol}, while the complete solution, including negative LMP cases, is visualized in Fig. \ref{fig:Optimalplan_negQH} (a).
\end{proof}

In determining the RCHP's optimal operational decision as a standalone hydrogen producer (M0), renewable producer (M1-p), or price-elastic consumer (M1-c), the optimization model \eqref{eq:maxprofit} simplifies to a LP due to the implicit constraints imposed by market participation rules. Therefore, the optimal solution can be readily obtained, as illustrated in Fig. \ref{fig:Optimalplan_negQH} (b)-(d). We also present the optimal production plans under different market models when $Q_{\mbox{\tiny\sf H}}>Q_{\mbox{\tiny\sf R}}$, as shown in Fig. \ref{fig:Optimalplan_negQR}.
\begin{figure}[!htb]
    \centering
    \includegraphics[scale=0.3]{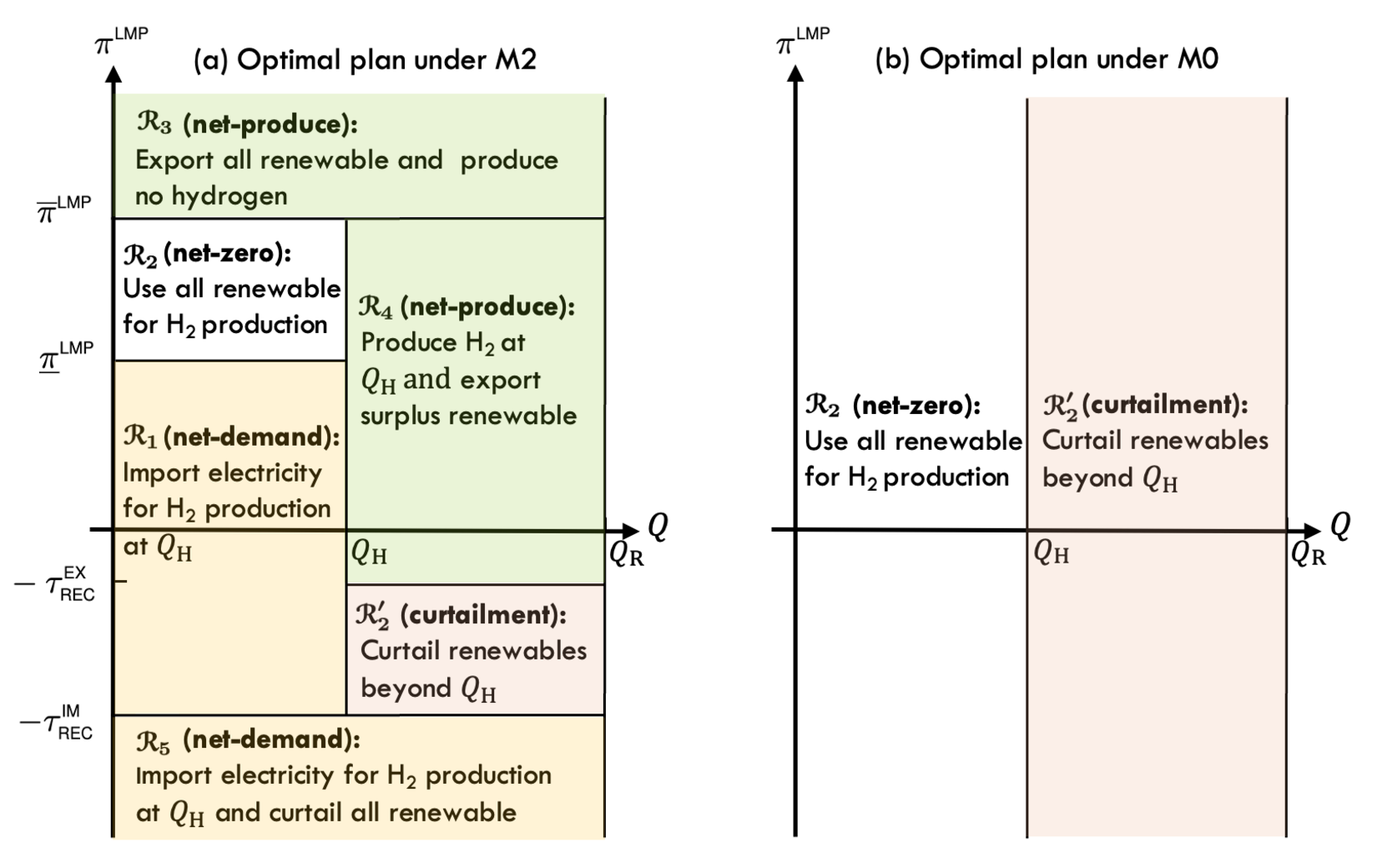}
    \includegraphics[scale=0.3]{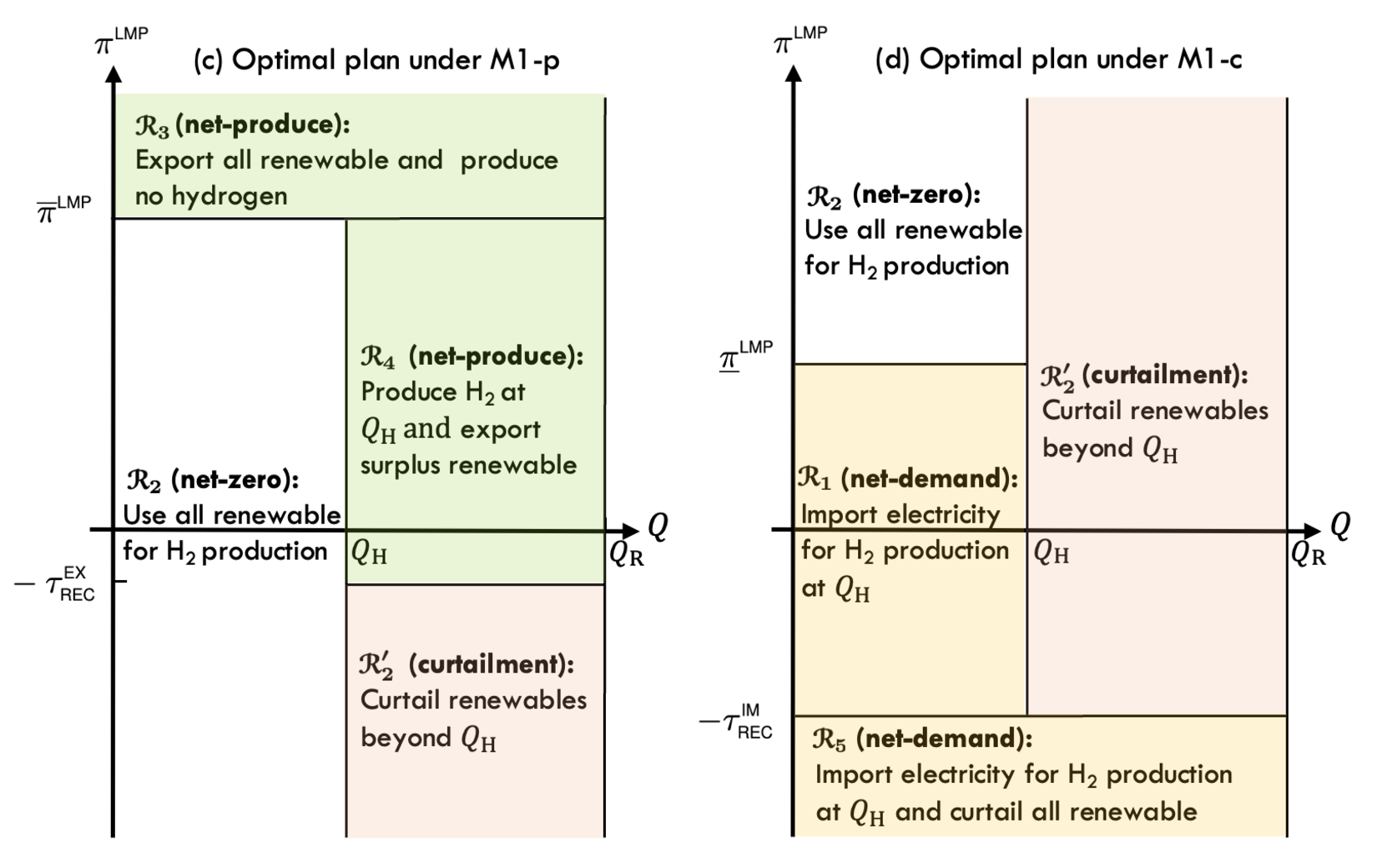}
    \caption{\scriptsize Optimal production plans for RCHP when $Q_{\mbox{\tiny\sf H}}<Q_{\mbox{\tiny\sf R}}$ (including negative LMP cases).}
    \label{fig:Optimalplan_negQH}
\end{figure}
\begin{figure}[!htb]
    \centering
    \includegraphics[scale=0.3]{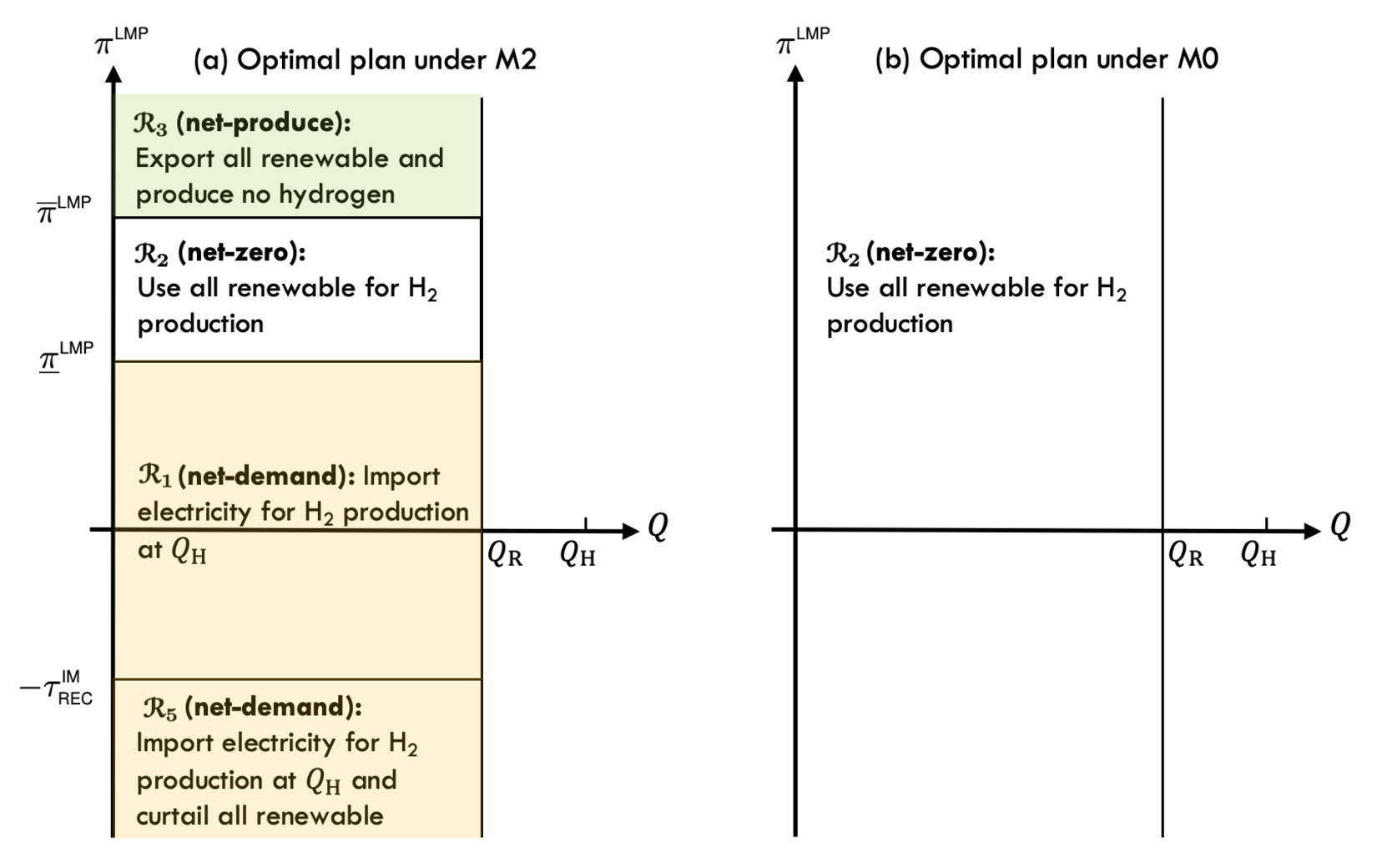}
    \includegraphics[scale=0.3]{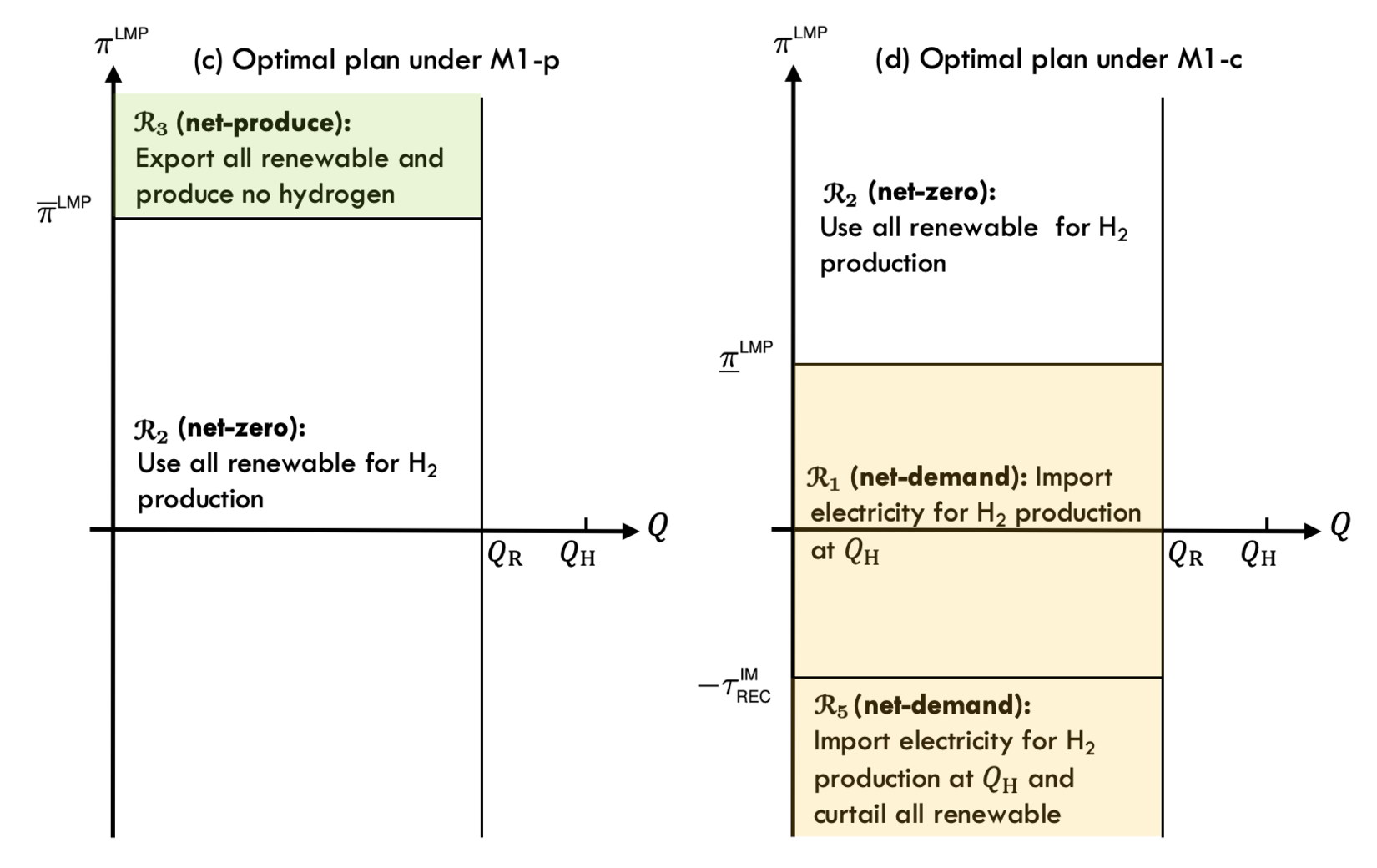}
    \caption{\scriptsize Optimal production plans for RCHP when $Q_{\mbox{\tiny\sf H}}>Q_{\mbox{\tiny\sf R}}$ (including negative LMP cases).}
    \label{fig:Optimalplan_negQR}
\end{figure}

\subsection{Optimal Production Plan with Piecewise Linear Hydrogen Production Function}\label{appendix:piecewise_prod}
\begin{figure}[!htb]
    \centering
    \includegraphics[scale=0.4]{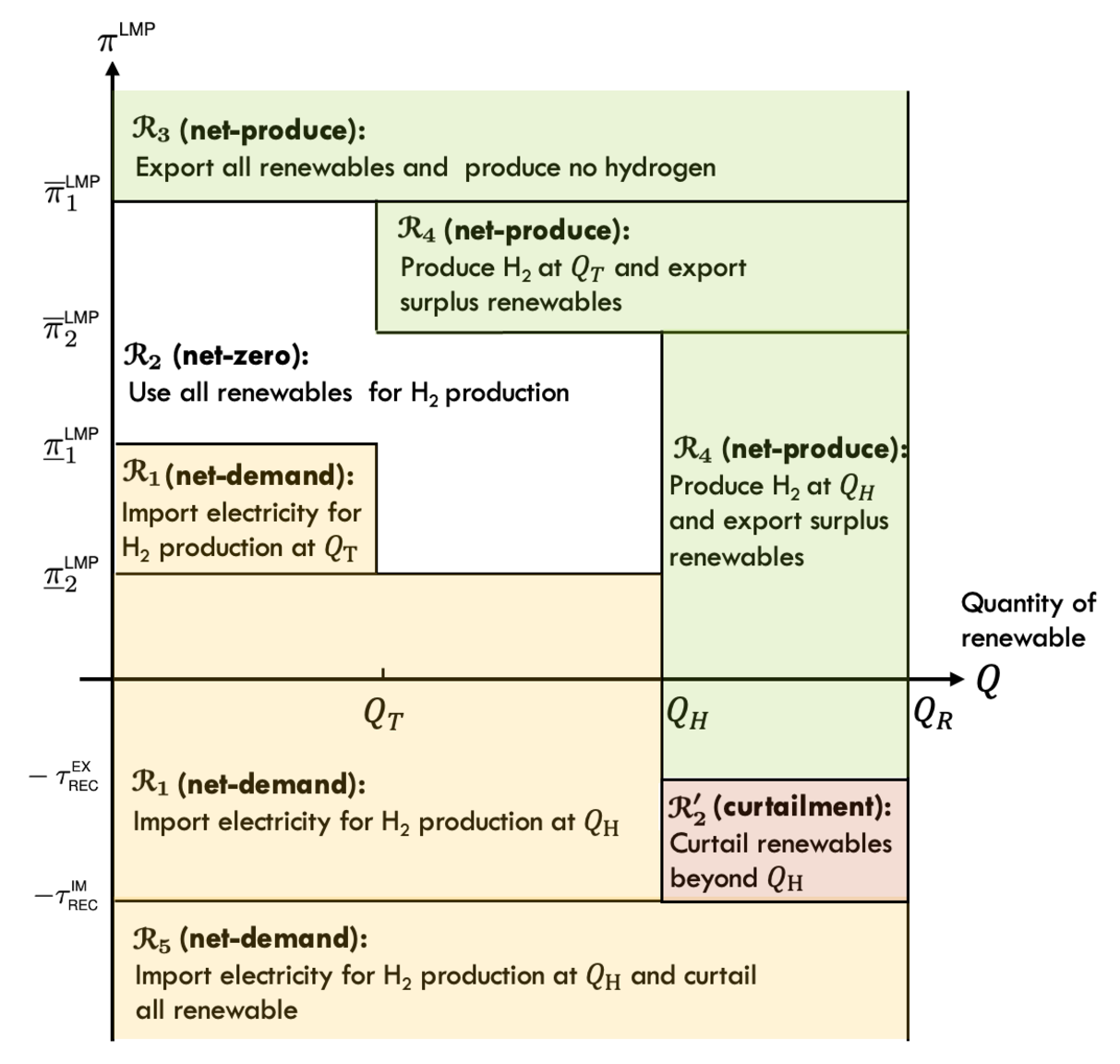}
    \includegraphics[scale=0.4]{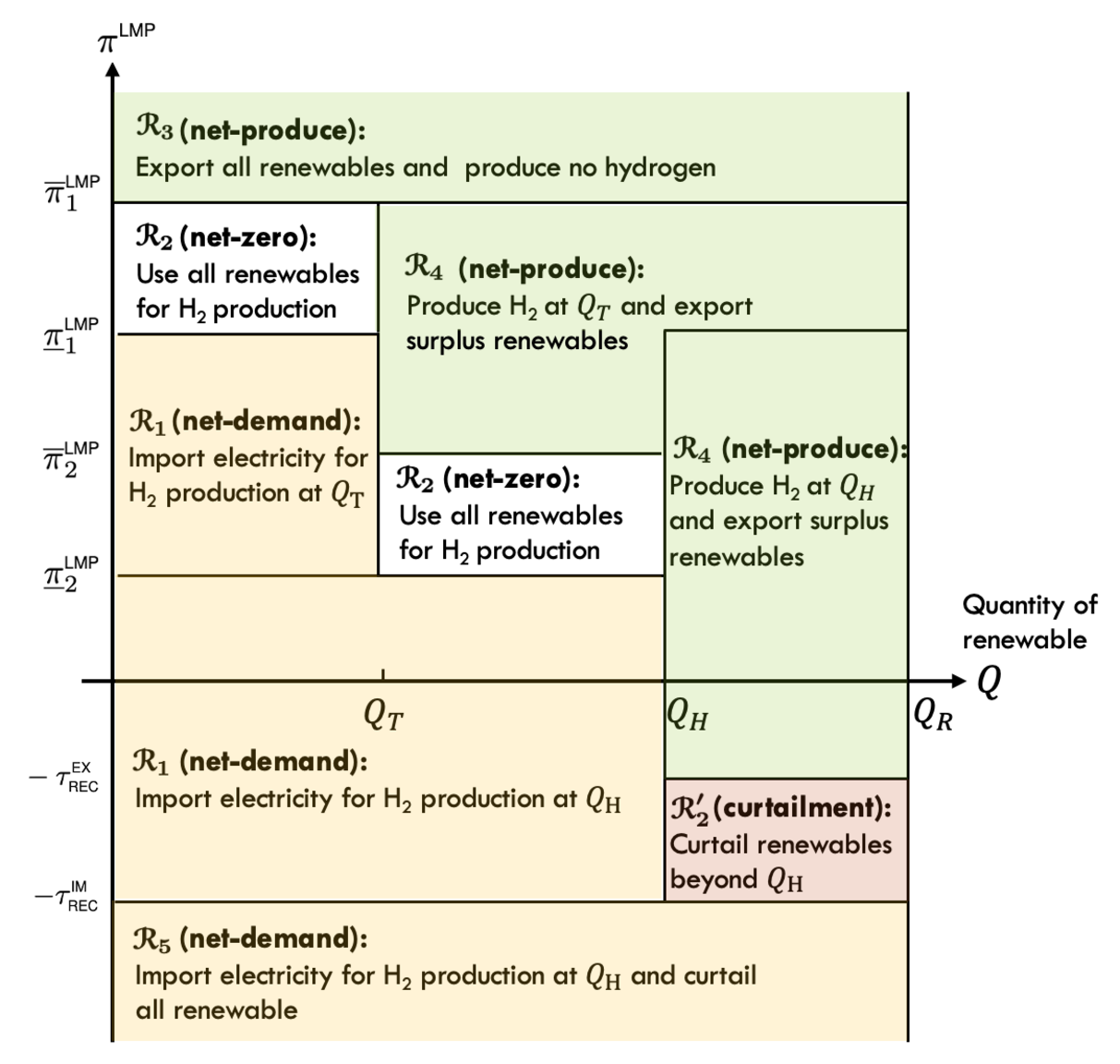}
    \caption{\scriptsize Optimal production plan for RCHP with a two-segment piecewise linear production function (prosumer model M2).}
    \label{fig:piecewise}
\end{figure}
A more precise modeling of the electrolyzer production function can be achieved using a piecewise linear approximation of the hydrogen production curve \cite{DhankarChenTong:24PESGM,lesniak2024}.

Consider the electrolyzer's hydrogen production at time $t$, which is constrained by a concave piecewise linear production function with $K$ segments, as follows:
\begin{equation}\label{eq:piecewise}
H_t\leq (\alpha_k P^{\mbox{\tiny\sf H}}_t+\beta_k)\Delta T \quad \forall k\in \{1,\cdots, K\},
\end{equation}
where $\alpha_k$ and $\beta_k$ are the slope and intercept of the $k$-th linear segment, respectively.
For simplicity and consistent with the main context, we assume $\Delta T=1$ and $\beta_1=0$.
The optimization problem \eqref{eq:maxprofit}, with the objective function \eqref{eq:objwREC}, can be reformulated by replacing the hydrogen production function \eqref{eq:linear_prod} with \eqref{eq:piecewise}.

Following the approach outlined in the proof of Theorem~\ref{Thm:solution}, the optimal solution for the RCHP's real-time operation can be derived. Despite the increased complexity, a threshold-based closed-form solution remains attainable, although it involves more thresholds than in the case of a linear production function.
Fig. \ref{fig:piecewise} illustrates the optimal production plan for an RCHP prosumer with a two-segment piecewise linear production function. The model yields six electricity price thresholds, which can be pre-determined based on system parameters. Depending on these parameters, the ordering of certain thresholds may vary. In particular, we present two distinct production plans where either $\overline{\pi}^{\mbox{\tiny\sf LMP}}_{\mbox{\tiny 2}}>\underline{\pi}^{\mbox{\tiny\sf LMP}}_{\mbox{\tiny 1}}$ or $\overline{\pi}^{\mbox{\tiny\sf LMP}}_{\mbox{\tiny 2}}<\underline{\pi}^{\mbox{\tiny\sf LMP}}_{\mbox{\tiny 1}}$.
\begin{equation}\label{eq:thresholds-piecewise}
\begin{array}{cl}
    \underline{\pi}^{\mbox{\tiny\sf LMP}}_{\mbox{\tiny 1}} &= \alpha_{\mbox{\tiny 1}}(\pi^{\mbox{\tiny\sf H}}+\tau^{\mbox{\tiny\sf H}}-c^{\mbox{\tiny\sf w}})-\tau_{\mbox{\tiny\sf REC}}^{\mbox{\tiny\sf IM}},\\[1pt]
    \underline{\pi}^{\mbox{\tiny\sf LMP}}_{\mbox{\tiny 2}} &= \alpha_{\mbox{\tiny 2}}(\pi^{\mbox{\tiny\sf H}}+\tau^{\mbox{\tiny\sf H}}-c^{\mbox{\tiny\sf w}})-\tau_{\mbox{\tiny\sf REC}}^{\mbox{\tiny\sf IM}},\\[1pt]
    \overline{\pi}^{\mbox{\tiny\sf LMP}}_{\mbox{\tiny 1}} &= \alpha_{\mbox{\tiny 1}}(\pi^{\mbox{\tiny\sf H}}+\tau^{\mbox{\tiny\sf H}}-c^{\mbox{\tiny\sf w}})-\tau_{\mbox{\tiny\sf REC}}^{\mbox{\tiny\sf EX}},\\[1pt]
    \overline{\pi}^{\mbox{\tiny\sf LMP}}_{\mbox{\tiny 2}} &= \alpha_{\mbox{\tiny 2}}(\pi^{\mbox{\tiny\sf H}}+\tau^{\mbox{\tiny\sf H}}-c^{\mbox{\tiny\sf w}})-\tau_{\mbox{\tiny\sf REC}}^{\mbox{\tiny\sf EX}}.
\end{array}
\end{equation}

Additionally, the renewable generation thresholds, denoted as $(Q_{\mbox{\tiny\sf T}}, Q_{\mbox{\tiny\sf H}}, Q_{\mbox{\tiny\sf R}})$, correspond to the threshold where the electrolyzer's efficiency changes, \ie $Q_{\mbox{\tiny\sf T}}=(\beta_{\mbox{\tiny 2}}-\beta_{\mbox{\tiny 1}})/(\alpha_{\mbox{\tiny 1}}-\alpha_{\mbox{\tiny 2}})$, the electrolyzer's capacity, and the renewable generation capacity, respectively.

\subsection{Proofs of Proposition~\ref{Prop:linear} and Theorem~\ref{Thm:profitability}}\label{appendix:thm2}
\begin{proof}[Proof of Proposition~\ref{Prop:linear}]
According to the optimal production plan, which includes four operational regions, if the RCHP prosumer operates optimally in a certain region during time interval $t$, the gross profit corresponding to  $\reg{1}$- $\reg{4}$ can be calculated as follows.\footnote{The inclusion of  $\reg{5}$ and  $\reg{6}$ in cases with negative LMP is straightforward. The conclusions and the logic of the proof remain unchanged.}
\begin{flalign}\nn
V_t^{(1)}&=(\underline{\pi}^{\mbox{\tiny\sf LMP}}-\pi_t^{\mbox{\tiny\sf LMP}})Q_{\mbox{\tiny\sf H}}+(\pi_t^{\mbox{\tiny\sf LMP}}+\tau_{\mbox{\tiny\sf REC}}^{\mbox{\tiny\sf IM}}+\tau^{\mbox{\tiny\sf R}})\eta_t Q_{\mbox{\tiny\sf R}}, &\\ \nn
V_t^{(2)}&=\big(\gamma(\pi^{\mbox{\tiny\sf H}}+\tau^{\mbox{\tiny\sf H}}-c^{\mbox{\tiny\sf W}})+\tau^{\mbox{\tiny\sf R}}\big)\eta_t Q_{\mbox{\tiny\sf R}}, &\\ \nn
V_t^{(3)}&=(\pi_t^{\mbox{\tiny\sf LMP}}+\tau_{\mbox{\tiny\sf REC}}^{\mbox{\tiny\sf EX}}+\tau^{\mbox{\tiny\sf R}})\eta_t Q_{\mbox{\tiny\sf R}},&\\
V_t^{(4)}&=(\overline{\pi}^{\mbox{\tiny\sf LMP}}-\pi_t^{\mbox{\tiny\sf LMP}})Q_{\mbox{\tiny\sf H}}+(\pi_t^{\mbox{\tiny\sf LMP}}+\tau_{\mbox{\tiny\sf REC}}^{\mbox{\tiny\sf EX}}+\tau^{\mbox{\tiny\sf R}})\eta_t Q_{\mbox{\tiny\sf R}}.&
\end{flalign}
Here, we denote $V_t^{(i)}$ without explicitly listing its arguments. The full expression is $V_t^{(i)}(\pi_t^{\mbox{\tiny\sf LMP}},\eta_t; Q_{\mbox{\tiny\sf R}},Q_{\mbox{\tiny\sf H}})$.

By taking the conditional expectation of the gross profit in each region, we obtain the conditional expected gross profit $\Pi_{t,\kappa}^{(i)}(Q_{\mbox{\tiny\sf R}},Q_{\mbox{\tiny\sf H}})$ for  $\reg{1}$- $\reg{4}$, where $\E_{t,\kappa}^{(i)}[\cdot]$ denotes the conditional expectation operator in interval $t$ on  $\reg{i}$, and $\kappa=Q_{\mbox{\tiny\sf H}}/ Q_{\mbox{\tiny\sf R}}$.
\begin{flalign}\nn
\Pi_{t,\kappa}^{(1)}(Q_{\mbox{\tiny\sf R}},Q_{\mbox{\tiny\sf H}})&=\Big((\tau^{\mbox{\tiny\sf IM}}_{\mbox{\tiny\sf REC}}+\tau^{\mbox{\tiny\sf R}})\E_{t,\kappa}^{(1)}[\eta_t]+\E_{t,\kappa}^{(1)}[\eta_t\pi_{t}^{\mbox{\tiny\sf LMP}}]\Big)Q_{\mbox{\tiny\sf R}}&\\ \nn
&+\Big(\underline{\pi}^{\mbox{\tiny\sf LMP}}- \E_{t,\kappa}^{(1)}[\pi_t^{\mbox{\tiny\sf LMP}}]\Big) Q_{\mbox{\tiny\sf H}}, &\\ \nn
\Pi_{t,\kappa}^{(2)}(Q_{\mbox{\tiny\sf R}},Q_{\mbox{\tiny\sf H}})&=\Big(\big(\gamma(\pi^{\mbox{\tiny\sf H}}+\tau^{\mbox{\tiny\sf H}}-c^{\mbox{\tiny\sf W}})+\tau^{\mbox{\tiny\sf R}}\big)\E_{t,\kappa}^{(2)}[\eta_t] \Big)Q_{\mbox{\tiny\sf R}}, &\\ \nn
\Pi_{t,\kappa}^{(3)}(Q_{\mbox{\tiny\sf R}},Q_{\mbox{\tiny\sf H}})&=\Big((\tau^{\mbox{\tiny\sf EX}}_{\mbox{\tiny\sf REC}}+\tau^{\mbox{\tiny\sf R}})\E_{t,\kappa}^{(3)}[\eta_t]+\E_{t,\kappa}^{(3)}[\eta_t\pi_t^{\mbox{\tiny\sf LMP}}] \Big)Q_{\mbox{\tiny\sf R}}, &\\ \nn
\Pi_{t,\kappa}^{(4)}(Q_{\mbox{\tiny\sf R}},Q_{\mbox{\tiny\sf H}})&=\Big((\tau^{\mbox{\tiny\sf EX}}_{\mbox{\tiny\sf REC}}+\tau^{\mbox{\tiny\sf R}})\E_{t,\kappa}^{(4)}[\eta_t]+\E_{t,\kappa}^{(4)}[\eta_t\pi_{t}^{\mbox{\tiny\sf LMP}}]\Big)Q_{\mbox{\tiny\sf R}}&\\ 
&+\Big(\overline{\pi}^{\mbox{\tiny\sf LMP}}- \E_{t,\kappa}^{(4)}[\pi_t^{\mbox{\tiny\sf LMP}}]\Big) Q_{\mbox{\tiny\sf H}}. &
\end{flalign}

Summing the conditional expected gross profit over all regions, weighted by the probabilities of each region $P_{t,\kappa}^{(i)}$, we derive the expected gross profit for the RCHP prosumer as a function of the nameplate capacities of the electrolyzer and renewable plant, given by
\beq\label{eq:Expect_NR}
\Pi_t(Q_{\mbox{\tiny\sf R}},Q_{\mbox{\tiny\sf H}})= \sum_{i=1}^4 P_{t,\kappa}^{(i)}\Pi_{t,\kappa}^{(i)}(Q_{\mbox{\tiny\sf R}},Q_{\mbox{\tiny\sf H}})=A^{\mbox{\tiny\sf R}}_{t,\kappa}Q_{\mbox{\tiny\sf R}}+A^{\mbox{\tiny\sf H}}_{t,\kappa}Q_{\mbox{\tiny\sf H}}.
\eeq

Subtracting the amortized fixed costs from the expected gross profit over $n$ periods, the expected $n$-period operating profit is
\begin{flalign}\nn
J^{\mbox{\tiny\sf OP}}_n(Q_{\mbox{\tiny\sf R}},Q_{\mbox{\tiny\sf H}})&=\sum_{t=1}^n \Pi_t(Q_{\mbox{\tiny\sf R}},Q_{\mbox{\tiny\sf H}})-(\alpha^{\mbox{\tiny\sf R}}_nQ_{\mbox{\tiny\sf R}}+\alpha^{\mbox{\tiny\sf H}}_nQ_{\mbox{\tiny\sf H}})&\\ \nn
&=\big(\sum_{t=1}^nA^{\mbox{\tiny\sf R}}_{t,\kappa}-\alpha_n^{\mbox{\tiny\sf R}}\big)Q_{\mbox{\tiny\sf R}}+\big(\sum_{t=1}^nA^{\mbox{\tiny\sf H}}_{t,\kappa}-\alpha_n^{\mbox{\tiny\sf H}}\big)Q_{\mbox{\tiny\sf H}}.&
\end{flalign}
\end{proof}

\begin{proof}[Proof of Theorem~\ref{Thm:profitability}]
From the expected gross profit expression \eqref{eq:Expect_NR}, we observe that for RCHPs with different capacity pairs $(Q_{\mbox{\tiny\sf H}}, Q_{\mbox{\tiny\sf R}})$, if the capacity ratio $\kappa=Q_{\mbox{\tiny\sf H}}/ Q_{\mbox{\tiny\sf R}}$ remains constant, the values of $A^{\mbox{\tiny\sf R}}_{t,\kappa}$ and $A^{\mbox{\tiny\sf H}}_{t,\kappa}$ are constants. Consequently, the expected gross profit $\Pi_t(Q_{\mbox{\tiny\sf R}},Q_{\mbox{\tiny\sf H}})$ is linear with respect to the capacity pair $(Q_{\mbox{\tiny\sf H}}, Q_{\mbox{\tiny\sf R}})$.

Furthermore, the amortized fixed cost is also linear with respect to the electrolyzer and renewable plant capapcities. Therefore, as shown in \eqref{eq:nCF1}, if $\kappa$ is fixed, the expected $n$-period operating profit $J^{\mbox{\tiny\sf OP}}_n(Q_{\mbox{\tiny\sf R}},Q_{\mbox{\tiny\sf H}})$ is a linear function of the nameplate capacities $(Q_{\mbox{\tiny\sf H}}, Q_{\mbox{\tiny\sf R}})$.

In the nameplate capacity plane $Q_{\mbox{\tiny\sf H}}$ vs $Q_{\mbox{\tiny\sf R}}$, for any capacity pair lying on a line with a fixed slope $\kappa$, the expected operating profit either increases or decreases linearly away from the origin. This implies that the break-even points of the RCHP capacity form a union of linear lines in the nameplate capacity plane, determined by the slope $\kappa^0$ satisfying $J^{\mbox{\tiny\sf OP}}_n(Q_{\mbox{\tiny\sf R}},\kappa^0Q_{\mbox{\tiny\sf R}})=0$.

Now, considering the RCHP operation under fixed system parameters, we fix the renewable capcacity $Q_{\mbox{\tiny\sf R}}$ and examine the impact of the electrolyzer capacity $Q_{\mbox{\tiny\sf H}}=\kappa Q_{\mbox{\tiny\sf R}}$ on the expected operating profit. Taking the partial derivative of $J^{\mbox{\tiny\sf OP}}_n(Q_{\mbox{\tiny\sf R}},Q_{\mbox{\tiny\sf H}})$ with respect to $Q_{\mbox{\tiny\sf H}}$, we obtain
\begin{flalign}
\nn
\frac{\partial J^{\mbox{\tiny\sf OP}}_n(Q_{\mbox{\tiny\sf R}},Q_{\mbox{\tiny\sf H}})}{\partial Q_{\mbox{\tiny\sf H}}}&=\sum_{t=1}^n \!\Big(\!\int_{0}^{\underline{\pi}^{\mbox{\tiny\sf LMP}}}\!\!\!\!\!\!\int_{0}^{\kappa}\!\!\!(\underline{\pi}^{\mbox{\tiny\sf LMP}}\!\!-\!\!\pi_t^{\mbox{\tiny\sf LMP}})\rho_t(\pi_t^{\mbox{\tiny\sf LMP}}\!,\eta_t)d\eta_t d\pi_t^{\mbox{\tiny\sf LMP}}&\\ \nn
&+\!\! \int_{0}^{\overline{\pi}^{\mbox{\tiny\sf LMP}}}\!\!\!\!\!\!\int_{\kappa}^{1}\!\!\!(\overline{\pi}^{\mbox{\tiny\sf LMP}}\!\!-\!\!\pi_t^{\mbox{\tiny\sf LMP}})\rho_t(\pi_t^{\mbox{\tiny\sf LMP}}\!,\eta_t)d\eta_t d\pi_t^{\mbox{\tiny\sf LMP}} \Big)\!\!-\!\alpha_n^{\mbox{\tiny\sf H}}&\\
&=\sum_{t=1}^n A_{t,\kappa}^{\mbox{\tiny\sf H}}-\alpha_n^{\mbox{\tiny\sf H}},&
\end{flalign}
where $\rho_t(\pi_t^{\mbox{\tiny\sf LMP}},\eta_t)$ is the joint probability density function of distribution $(\pi_t^{\mbox{\tiny\sf LMP}},\eta_t)$.

Consider two electrolyzer capacity values, $\tilde{Q}_{\mbox{\tiny\sf H}}=\tilde{\kappa}Q_{\mbox{\tiny\sf R}}$ and $\tilde{Q}'_{\mbox{\tiny\sf H}}=(\tilde{\kappa}+\delta)Q_{\mbox{\tiny\sf R}}$, where $0\leq \tilde{Q}_{\mbox{\tiny\sf H}}< \tilde{Q}'_{\mbox{\tiny\sf H}}$.
\begin{flalign}
\nn
&\frac{\partial J^{\mbox{\tiny\sf OP}}_n}{\partial Q_{\mbox{\tiny\sf H}}}(Q_{\mbox{\tiny\sf R}},\tilde{Q}'_{\mbox{\tiny\sf H}})-\frac{\partial J^{\mbox{\tiny\sf OP}}_n}{\partial Q_{\mbox{\tiny\sf H}}}(Q_{\mbox{\tiny\sf R}},\tilde{Q}_{\mbox{\tiny\sf H}})&\\ \nn
=&\sum_{t=1}^n \Big( \int_{0}^{\underline{\pi}^{\mbox{\tiny\sf LMP}}}\!\!\!\!\!\!\int_{\tilde{\kappa}}^{\tilde{\kappa}+\delta}(\underline{\pi}^{\mbox{\tiny\sf LMP}}-\pi_t^{\mbox{\tiny\sf LMP}})\rho_t(\pi_t^{\mbox{\tiny\sf LMP}},\eta_t)d\eta_t d\pi_t^{\mbox{\tiny\sf LMP}}&\\ \nn
-& \int_{0}^{\overline{\pi}^{\mbox{\tiny\sf LMP}}}\!\!\!\!\!\!\int_{\tilde{\kappa}}^{\tilde{\kappa}+\delta}(\overline{\pi}^{\mbox{\tiny\sf LMP}}-\pi_t^{\mbox{\tiny\sf LMP}})\rho_t(\pi_t^{\mbox{\tiny\sf LMP}},\eta_t)d\eta_t d\pi_t^{\mbox{\tiny\sf LMP}} \Big)&\\ \nn
=&\sum_{t=1}^n\Big( \int_{0}^{\underline{\pi}^{\mbox{\tiny\sf LMP}}}\!\!\!\!\!\!\int_{\tilde{\kappa}}^{\tilde{\kappa}+\delta}(\underline{\pi}^{\mbox{\tiny\sf LMP}}-\overline{\pi}^{\mbox{\tiny\sf LMP}})\rho_t(\pi_t^{\mbox{\tiny\sf LMP}},\eta_t)d\eta_t d\pi_t^{\mbox{\tiny\sf LMP}}&\\
-& \int_{\underline{\pi}^{\mbox{\tiny\sf LMP}}}^{\overline{\pi}^{\mbox{\tiny\sf LMP}}}\!\!\!\!\!\!\int_{\tilde{\kappa}}^{\tilde{\kappa}+\delta} (\overline{\pi}^{\mbox{\tiny\sf LMP}}-\pi_t^{\mbox{\tiny\sf LMP}})\rho_t(\pi_t^{\mbox{\tiny\sf LMP}},\eta_t)d\eta_t d\pi_t^{\mbox{\tiny\sf LMP}} \Big)\le 0.&
\end{flalign}
By analyzing the difference between their partial derivatives, we establish that the partial derivative $\partial J^{\mbox{\tiny\sf OP}}_n(Q_{\mbox{\tiny\sf R}},Q_{\mbox{\tiny\sf H}})/\partial Q_{\mbox{\tiny\sf H}}$ decreases as $Q_{\mbox{\tiny\sf H}}$ (or $\kappa$) increases for $0 \leq \kappa < 1$
\footnote{Under the mild assumption that $\underline{\pi}^{\mbox{\tiny\sf LMP}} > 0$, and for any $\eta\in [0,1]$, there exists a probability density function $\rho_t(\pi_t^{\mbox{\tiny\sf LMP}}\in [0, \underline{\pi}^{\mbox{\tiny\sf LMP}}],\eta_t=\eta)>0$, this monotonicity is strict.}.
For $\kappa \geq 1$, this derivative remains constant.

From this, three cases arise:
\begin{itemize}
    \item If $\partial J^{\mbox{\tiny\sf OP}}_n(Q_{\mbox{\tiny\sf R}},Q_{\mbox{\tiny\sf H}})/\partial Q_{\mbox{\tiny\sf H}}< 0$ for all $Q_{\mbox{\tiny\sf H}}=\kappa Q_{\mbox{\tiny\sf R}}$ with $\kappa\ge 0$, then the expected operating profit decreases with increasing electrolyzer capacity. The RCHP is profitable for all electrolyzer capacities with $0\le \kappa<\kappa^0$ and in deficit for $\kappa>\kappa^0$.

    \item If $\partial J^{\mbox{\tiny\sf OP}}_n(Q_{\mbox{\tiny\sf R}},Q_{\mbox{\tiny\sf H}})/\partial Q_{\mbox{\tiny\sf H}}> 0$ for all $Q_{\mbox{\tiny\sf H}}=\kappa Q_{\mbox{\tiny\sf R}}$ with $\kappa\ge 0$, then the expected operating profit increases with electrolyzer capacity. The RCHP is in deficit for $0 \leq \kappa < \kappa^0$ and profitable for $\kappa > \kappa^0$.

    \item If there exists $\kappa^*$ such that $\frac{\partial J^{\mbox{\tiny\sf OP}}_n}{\partial Q_{\mbox{\tiny\sf H}}}(Q_{\mbox{\tiny\sf R}},Q^*_{\mbox{\tiny\sf H}})= 0$ at $Q^*_{\mbox{\tiny\sf H}}=\kappa^*Q_{\mbox{\tiny\sf R}}$, then the expected operating profit is maximized at $Q^*_{\mbox{\tiny\sf H}}$. For $0\le \kappa<\kappa^*$, the expected operating profit increases with respect to $\kappa$, whereas for $\kappa>\kappa^*$, it decreases. A special case arises when $\partial J^{\mbox{\tiny\sf OP}}_n(Q_{\mbox{\tiny\sf R}},Q_{\mbox{\tiny\sf H}})/\partial Q_{\mbox{\tiny\sf H}}= 0$ for all $Q_{\mbox{\tiny\sf H}}\ge Q_{\mbox{\tiny\sf R}}$, \ie when $\kappa\ge 1$. In this case, any electrolyzer capacity $Q_{\mbox{\tiny\sf H}}\ge Q_{\mbox{\tiny\sf R}}$ yields the same maximum operating profit.
\end{itemize}

Since the renewable capacity $Q_{\mbox{\tiny\sf R}}$ is arbitrarily chosen, the conclusions hold for all capacity pairs $(Q_{\mbox{\tiny\sf R}}, Q_{\mbox{\tiny\sf H}})$ in the nameplate capacity plane. Therefore, both the profitable and deficit regions form convex cones bounded by the break-even lines, and the optimal matching electrolyzer capacity $Q^*_{\mbox{\tiny\sf H}}$ is linearly proportional to $ Q_{\mbox{\tiny\sf R}}$, as stated in Theorem \ref{Thm:profitability}.
\end{proof}

\subsection{Sketch of the proof for Proposition~\ref{Prop:colocation}}
\begin{proof}
We establish this by comparing the optimal operation of a colocated prosumer (M2) to a strictly constrained subset where the renewable generator and electrolyzer operate independently. In the separated configuration, the electrolyzer must satisfy all its energy demand via grid imports, while the renewable must export all its generation.

The colocated M2 model, however, can directly utilize onsite renewable power for the electrolyzer, which bypasses the grid entirely, thereby avoiding the non-negative REC price spread on the self-consumed power. Because the objective function of the colocated system is identical to the separated system plus these non-negative cost savings, the expected operating profit of the colocated RCHP must be greater than or equal to the sum of the separated operations.
\end{proof}

\subsection{Adaptability to Nonlinear Fixed Costs}\label{appendix:nonlinear_cost}
To demonstrate that our framework can flexibly accommodate nonlinear fixed costs, we extended the linear cost assumption used in Sec.~\ref{sec:simulation} by applying piecewise linear cost functions for both the renewable plant and the electrolyzer, as shown in Fig.~\ref{fig:capacity_cost}. Using the consumer model (M1-c) as an illustrative example, we calculated the RCHP's annual operating profit in 2022 as a function of the solar generation and electrolyzer capacities. The comparative results under linear and piecewise linear fixed costs are illustrated in Fig.~\ref{fig:heatmap_piecewise}. Under the nonlinear cost structure (right panel), the break-even boundaries and the optimal capacity matching line naturally transitioned into piecewise linear frontiers. Specifically, we observed that the optimal capacity ratio dynamically shifted to favor the technology exhibiting steeper cost reductions in the corresponding capacity ranges.
\begin{figure}[!htb]
        \centering
        \includegraphics[scale=0.32]{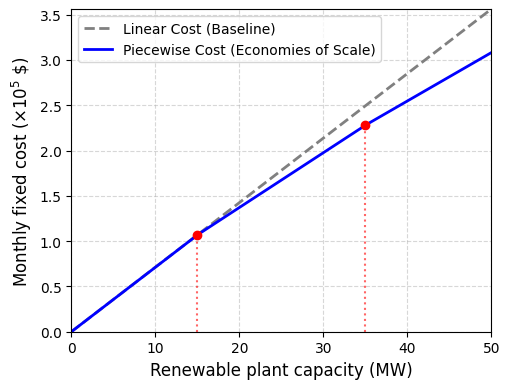}
        \includegraphics[scale=0.32]{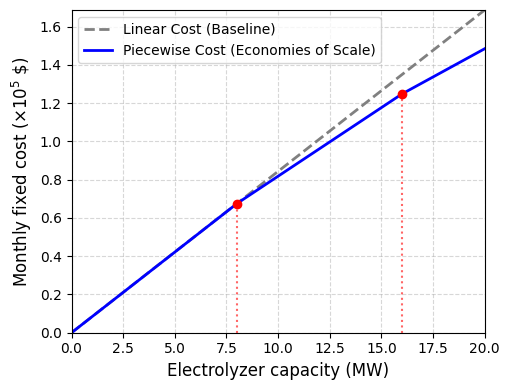}
        \caption{\scriptsize Monthly fixed cost of facilities at different capacity levels. (Left: renewable plant fixed cost; right: electrolyzer fixed cost.)}
        \label{fig:capacity_cost}
        \vspace{-2em}   
\end{figure}   
\begin{figure}[!htb]
    \centering
    \includegraphics[scale=0.38]{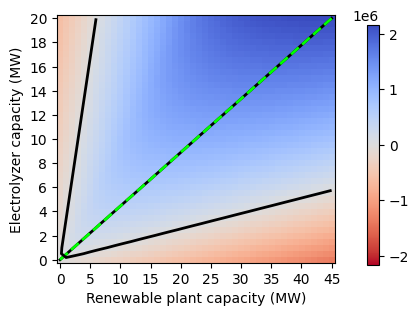}
    \includegraphics[scale=0.38]{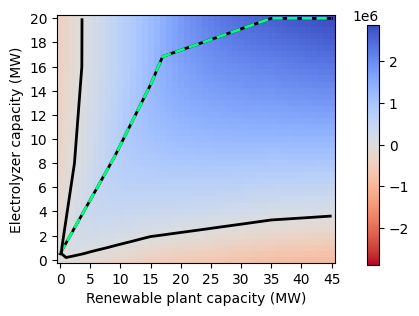}
    \caption{\scriptsize Annual operating profit in 2022 as a function of solar generation nameplate capacity (x-axis) and electrolyzer nameplate capacity (y-axis). Solid black: Break-even line. Green dashed: Optimal electrolyzer nameplate capacity as a function of solar generation nameplate capacity. (Left: linear fixed cost function; right: piecewise linear fixed cost function.)} \label{fig:heatmap_piecewise}
\vspace{-1em}     
\end{figure}

These results confirm that while sizing configurations adjust to nonlinear cost trends, the underlying optimization remains highly tractable and our fundamental operational insights are completely robust.

\subsection{Proof of Theorem~\ref{Thm:nameplate}}\label{appendix_thm3}
\begin{proof}
We apply the Lagrangian method to the optimization problem \eqref{eq:stochastic}. The Lagrangian function is defined as
\begin{flalign}
\mathcal{L}(Q_{\mbox{\tiny\sf R}},Q_{\mbox{\tiny\sf H}},\lambda)=&-J^{\mbox{\tiny\sf OP}}_n(Q_{\mbox{\tiny\sf R}},Q_{\mbox{\tiny\sf H}})+\lambda(\alpha^{\mbox{\tiny\sf R}}_nQ_{\mbox{\tiny\sf R}}+\alpha^{\mbox{\tiny\sf H}}_nQ_{\mbox{\tiny\sf H}}-B_n).&
\end{flalign}

Following the proof of Theorem~\ref{Thm:profitability}, we compute the partial derivative of $J^{\mbox{\tiny\sf OP}}_n(Q_{\mbox{\tiny\sf R}},Q_{\mbox{\tiny\sf H}})$ with respect to $Q_{\mbox{\tiny\sf R}}$:
\begin{flalign}
\nn
\frac{\partial J^{\mbox{\tiny\sf OP}}_n(Q_{\mbox{\tiny\sf R}},\!Q_{\mbox{\tiny\sf H}})}{\partial Q_{\mbox{\tiny\sf R}}}&\!=\!\!\sum_{t=1}^n \!\Big(\!\!\int_{0}^{\underline{\pi}^{\mbox{\tiny\sf LMP}}}\!\!\!\!\!\!\!\int_{0}^{\kappa}\!\!(\pi_t^{\mbox{\tiny\sf LMP}}\!\!+\!\!\tau^{\mbox{\tiny\sf IM}}_{\mbox{\tiny\sf REC}})\eta_t\rho_t(\pi_t^{\mbox{\tiny\sf LMP}}\!,\eta_t)d\eta_t d\pi_t^{\mbox{\tiny\sf LMP}}&\\ \nn
&\!+\!\! \int_{\underline{\pi}^{\mbox{\tiny\sf LMP}}}^{\overline{\pi}^{\mbox{\tiny\sf LMP}}}\!\!\!\!\!\!\int_{0}^{\kappa}\!\!\!\gamma(\pi^{\mbox{\tiny\sf H}}+\tau^{\mbox{\tiny\sf H}}-c^{\mbox{\tiny\sf w}})\eta_t\rho_t(\pi_t^{\mbox{\tiny\sf LMP}}\!,\eta_t)d\eta_t d\pi_t^{\mbox{\tiny\sf LMP}}&\\ \nn
&\!+\!\!\int_{\overline{\pi}^{\mbox{\tiny\sf LMP}}}^{+\infty}\!\!\!\int_{0}^{1}\!\!(\pi_t^{\mbox{\tiny\sf LMP}}\!\!+\!\!\tau^{\mbox{\tiny\sf EX}}_{\mbox{\tiny\sf REC}})\eta_t\rho_t(\pi_t^{\mbox{\tiny\sf LMP}}\!,\eta_t)d\eta_t d\pi_t^{\mbox{\tiny\sf LMP}}&\\ \nn
&\!+\!\!\int_{0}^{\overline{\pi}^{\mbox{\tiny\sf LMP}}}\!\!\!\!\!\!\int_{\kappa}^{1}\!\!(\pi_t^{\mbox{\tiny\sf LMP}}\!\!+\!\!\tau^{\mbox{\tiny\sf EX}}_{\mbox{\tiny\sf REC}})\eta_t\rho_t(\pi_t^{\mbox{\tiny\sf LMP}}\!,\eta_t)d\eta_t d\pi_t^{\mbox{\tiny\sf LMP}}&\\ \nn
&\!+\!\!\int_{0}^{+\infty}\!\!\!\int_{0}^{1} \tau^{\mbox{\tiny\sf R}} \eta_t\rho_t(\pi_t^{\mbox{\tiny\sf LMP}}\!,\eta_t)d\eta_t d\pi_t^{\mbox{\tiny\sf LMP}}\Big)\!\!-\!\alpha_n^{\mbox{\tiny\sf R}} &\\ \nn
&\!=\sum_{t=1}^n A_{t,\kappa}^{\mbox{\tiny\sf R}}-\alpha_n^{\mbox{\tiny\sf R}}.&
\end{flalign}

The necessary conditions for optimality are obtained by setting the gradients of $\mathcal{L}$ with respect to $Q_{\mbox{\tiny\sf R}}$ and $Q_{\mbox{\tiny\sf H}}$ to zero:
\begin{flalign}\nn
-\frac{\partial J^{\mbox{\tiny\sf OP}}_n}{\partial Q_{\mbox{\tiny\sf R}}}(Q^*_{\mbox{\tiny\sf R}},Q^*_{\mbox{\tiny\sf H}})+\lambda^* \alpha^{\mbox{\tiny\sf R}}_n &=-\sum_{t=1}^n A_{t,\kappa^*}^{\mbox{\tiny\sf R}}+(1+\lambda^*)\alpha_n^{\mbox{\tiny\sf R}}=0, &\\
-\frac{\partial J^{\mbox{\tiny\sf OP}}_n}{\partial Q_{\mbox{\tiny\sf H}}}(Q^*_{\mbox{\tiny\sf R}},Q^*_{\mbox{\tiny\sf H}})+\lambda^* \alpha^{\mbox{\tiny\sf H}}_n &=-\sum_{t=1}^n A_{t,\kappa^*}^{\mbox{\tiny\sf H}}+(1+\lambda^*)\alpha_n^{\mbox{\tiny\sf H}}=0, &
\end{flalign}
where $\kappa^*=Q^*_{\mbox{\tiny\sf H}}/Q^*_{\mbox{\tiny\sf R}}$.

Dividing these two equations, we obtain
\begin{align}\label{eq:ratio}
\frac{\sum_{t=1}^n  A^{\mbox{\tiny\sf H}}_{t,\kappa^*}}{\sum_{t=1}^n  A^{\mbox{\tiny\sf R}}_{t,\kappa^*}}=\frac{\alpha_n^{\mbox{\tiny\sf H}}}{\alpha_n^{\mbox{\tiny\sf R}}}.
\end{align}

The optimal values of $Q^*_{\mbox{\tiny\sf R}}$ and $Q^*_{\mbox{\tiny\sf H}}$ can be determined by solving \eqref{eq:ratio} together with the budget constraint, as described in \eqref{eq:optcond}.

We have proved that $\partial J^{\mbox{\tiny\sf OP}}_n(Q_{\mbox{\tiny\sf R}},Q_{\mbox{\tiny\sf H}})/\partial Q_{\mbox{\tiny\sf H}}$ decreases as $Q_{\mbox{\tiny\sf H}}$ (or $\kappa$) increases. Similarly, we fix the electrolyzer capacity and consider two renewable capacity values: $\tilde{Q}_{\mbox{\tiny\sf R}}=Q_{\mbox{\tiny\sf H}}/\tilde{\kappa}$ and $\tilde{Q}'_{\mbox{\tiny\sf R}}=Q_{\mbox{\tiny\sf H}}/(\tilde{\kappa}+\delta)$ for a small $\delta>0$.
\begin{flalign}
\nn
&\frac{\partial J^{\mbox{\tiny\sf OP}}_n}{\partial Q_{\mbox{\tiny\sf R}}}(\tilde{Q}'_{\mbox{\tiny\sf R}},Q_{\mbox{\tiny\sf H}})-\frac{\partial J^{\mbox{\tiny\sf OP}}_n}{\partial Q_{\mbox{\tiny\sf R}}}(\tilde{Q}_{\mbox{\tiny\sf R}},Q_{\mbox{\tiny\sf H}})&\\ \nn
=&\sum_{t=1}^n \Big( \int_{0}^{\underline{\pi}^{\mbox{\tiny\sf LMP}}}\!\!\!\!\!\!\int_{\tilde{\kappa}}^{\tilde{\kappa}+\delta}(\pi_t^{\mbox{\tiny\sf LMP}}+\tau^{\mbox{\tiny\sf IM}}_{\mbox{\tiny\sf REC}})\eta_t\rho_t(\pi_t^{\mbox{\tiny\sf LMP}},\eta_t)d\eta_t d\pi_t^{\mbox{\tiny\sf LMP}}&\\ \nn
+&\int_{\underline{\pi}^{\mbox{\tiny\sf LMP}}}^{\overline{\pi}^{\mbox{\tiny\sf LMP}}}\!\!\!\!\!\!\int_{\tilde{\kappa}}^{\tilde{\kappa}+\delta}\gamma(\pi^{\mbox{\tiny\sf H}}+\tau^{\mbox{\tiny\sf H}}-c^{\mbox{\tiny\sf w}})\eta_t\rho_t(\pi_t^{\mbox{\tiny\sf LMP}},\eta_t)d\eta_t d\pi_t^{\mbox{\tiny\sf LMP}}&\\ \nn
-& \int_{0}^{\overline{\pi}^{\mbox{\tiny\sf LMP}}}\!\!\!\!\!\!\int_{\tilde{\kappa}}^{\tilde{\kappa}+\delta}(\pi_t^{\mbox{\tiny\sf LMP}}\!\!+\!\!\tau^{\mbox{\tiny\sf EX}}_{\mbox{\tiny\sf REC}})\eta_t\rho_t(\pi_t^{\mbox{\tiny\sf LMP}},\eta_t)d\eta_t d\pi_t^{\mbox{\tiny\sf LMP}} \Big)&\\ \nn
=&\sum_{t=1}^n\Big( \int_{0}^{\underline{\pi}^{\mbox{\tiny\sf LMP}}}\!\!\!\!\!\!\int_{\tilde{\kappa}}^{\tilde{\kappa}+\delta}(\tau^{\mbox{\tiny\sf IM}}_{\mbox{\tiny\sf REC}}-\tau^{\mbox{\tiny\sf EX}}_{\mbox{\tiny\sf REC}})\eta_t\rho_t(\pi_t^{\mbox{\tiny\sf LMP}},\eta_t)d\eta_t d\pi_t^{\mbox{\tiny\sf LMP}}&\\
+& \int_{\underline{\pi}^{\mbox{\tiny\sf LMP}}}^{\overline{\pi}^{\mbox{\tiny\sf LMP}}}\!\!\!\!\!\!\int_{\tilde{\kappa}}^{\tilde{\kappa}+\delta} (\overline{\pi}^{\mbox{\tiny\sf LMP}}-\pi_t^{\mbox{\tiny\sf LMP}})\eta_t\rho_t(\pi_t^{\mbox{\tiny\sf LMP}},\eta_t)d\eta_t d\pi_t^{\mbox{\tiny\sf LMP}} \Big)\ge 0.&
\end{flalign}
This allows us to conclude that $\partial J^{\mbox{\tiny\sf OP}}_n(Q_{\mbox{\tiny\sf R}},Q_{\mbox{\tiny\sf H}})/\partial Q_{\mbox{\tiny\sf R}}$ is an increasing function of $\kappa$ for $0 \leq \kappa < 1$. For $\kappa \geq 1$, this derivative remains constant.

Since $\alpha_n^{\mbox{\tiny\sf R}}$ and $\alpha_n^{\mbox{\tiny\sf H}}$ are constants, it follows that $\sum_{t=1}^n A^{\mbox{\tiny\sf R}}_{t,\kappa}$ and $\sum_{t=1}^n A^{\mbox{\tiny\sf H}}_{t,\kappa}$ are increasing and decreasing functions of $\kappa$, respectively. Then, the ratio $\sum_{t=1}^n  A^{\mbox{\tiny\sf H}}_{t,\kappa}/\sum_{t=1}^n  A^{\mbox{\tiny\sf R}}_{t,\kappa}$ appearing on the left hand-side of \eqref{eq:ratio} is a monotonically decreasing function of $\kappa$. For each capacity pair $(Q_{\mbox{\tiny\sf R}}, Q_{\mbox{\tiny\sf H}})$ that satisfies the budget constraint, there corresponds a unique capacity ratio $\kappa$. This monotonicity implies that searching for the nameplate capacity values along the budget constraint provides an efficient approach for guiding the RCHP to the optimal nameplate capacity pair, if an optimal solution exists.

Finally, we briefly discuss the existence of the optimal nameplate capacity pair. The budget constraint is a linear equation in the nameplate capacity plane, subject to the non-negativity constraints $Q_{\mbox{\tiny\sf R}}\ge 0$ and $Q_{\mbox{\tiny\sf H}}\ge 0$, which define a compact feasible set. By the Weierstrass theorem, the existence of an optimal solution is guaranteed if the expected operating profit function, $J^{\mbox{\tiny\sf OP}}_n(Q_{\mbox{\tiny\sf R}},Q_{\mbox{\tiny\sf H}})$, is continuous, which in turn requires the continuity of the probability density function $\rho_t(\pi_t^{\mbox{\tiny\sf LMP}},\eta_t)$. However, if $\rho_t(\pi_t^{\mbox{\tiny\sf LMP}},\eta_t)$ is discontinuous, the supremum of the expected operating profit function may not be attained. In such case, one can construct a sequence of nameplate capacity pairs that approach the optimal solution arbitrarily closely.
\end{proof}

\subsection{Empirical Example for Operating Profit Forecasting} \label{sec:empirical_example}
This example illustrates how Proposition~\ref{Prop:linear} can be applied to estimate the RCHP's expected operating profit.
\begin{figure}[htbp]
    \centering
    \includegraphics[scale=0.4]{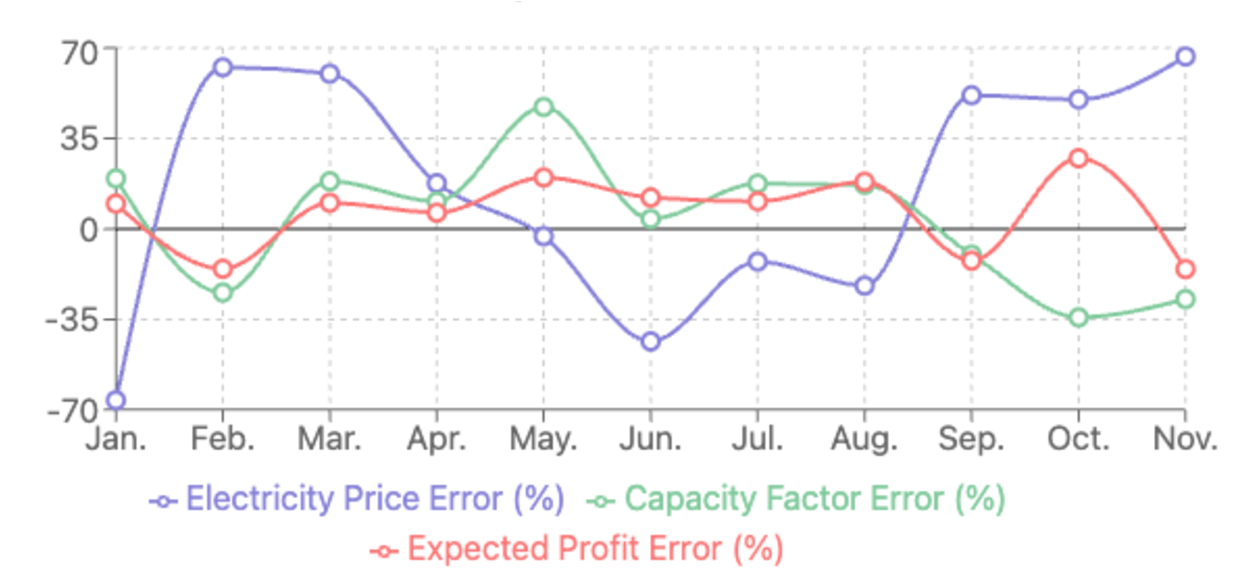}
    \vspace{-1em}  
    \caption{\scriptsize Prediction errors of LMP, capacity factor, and expected operating profit.}
    \label{fig:error}
\end{figure}

In practice, theoretical probabilities and expectations can be replaced with empirical counterparts derived from forecasted LMP and renewable trajectories. Since our focus is not on forecasting methodology, we do not discuss trajectory generation in detail; standard techniques such as Monte Carlo simulation or historical bootstrapping can be employed.

As a concrete example, for the 2022 New York case, we used a naive approach: historical LMP and renewable generation trajectories from 2021 were used to construct empirical probability distributions and conditional expectations as required by Proposition~\ref{Prop:linear}. These empirical values were then substituted into \eqref{eq:nCF1} to compute the expected operating profit for 2022.

Fig.~\ref{fig:error} shows the monthly prediction errors for this example, indicating that the accuracy of operating profit forecasts is comparable to that of renewable generation forecasts.

\subsection{Additional Numerical Results}\label{appendix:breakdown}
Table~\ref{tab:breakdown_CAISO} and Table~\ref{tab:breakdown_MISO} present the detailed annual revenue breakdowns for RCHPs deployed in CAISO and MISO under a hydrogen price of \$4/kg. These results complement the multi-ISO simulations discussed in Sec.~\ref{sec:multi-ISO}.
\begin{table*}[!t]
    \centering
    \footnotesize
    \renewcommand{\arraystretch}{1.5}
    \setlength{\tabcolsep}{2.5pt}
    \caption{\small RCHP Revenue Breakdown in 2022 (CAISO)}
    \label{tab:breakdown_CAISO}
    \begin{tabular}{l c c c c c c c c }
        \toprule
        & \multicolumn{2}{c}{M0: Standalone}
        & \multicolumn{2}{c}{M1-p: Producer}  & \multicolumn{2}{c}{M1-c: Consumer} & \multicolumn{2}{c}{M2: Prosumer} \\
        \midrule
        Renewable Type & Solar & Wind & Solar & Wind & Solar & Wind & Solar & Wind \\
        Total renewable generation ( MWh) & 0.9916$\times 10^5$ & 1.1332$\times 10^5$ & 0.9916$\times 10^5$ & 1.1332$\times 10^5$ & 0.9916$\times 10^5$ & 1.1332$\times 10^5$ & 0.9916$\times 10^5$ & 1.1332$\times 10^5$\\
        Renewable in hydrogen production (\%) & 71.91 & 81.86 & 68.65 & 77.62 & 71.91 & 81.86 & 68.65 & 77.62 \\
        Hydrogen produced (kg) & 1.3547$\times 10^6$ & 1.7626$\times 10^6$ & 1.2934$\times 10^6$ & 1.6713$\times 10^6$ & 3.0588$\times 10^6$ & 3.1612$\times 10^6$ & 2.9974$\times 10^6$ & 3.0699$\times 10^6$\\
        Revenue from hydrogen sales (\$) & 9.4832$\times 10^6$ & 1.2338$\times 10^7$ & 9.0537$\times 10^6$ & 1.1699$\times 10^7$ & 2.1411$\times 10^7$ & 2.2129$\times 10^7$ & 2.0982$\times 10^7$ & 2.1489$\times 10^7$\\
        Renewable sold in the market (\%) & 0 & 0 & 31.35 & 22.38 & 0 & 0 & 31.35 & 22.38 \\
        Revenue from renewable sales (\$) & 0 & 0 & 3.0771$\times 10^6$ & 2.8548$\times 10^6$ & 0 & 0 & 3.0771$\times 10^6$ & 2.8548$\times 10^6$ \\
        Renewable curtailed (\%) & 28.09 & 18.14 & 0 & 0 & 28.09 & 18.14 & 0 & 0 \\
        Revenue lost due to curtailment (\$) & 2.1000$\times 10^6$ & 1.7074$\times 10^6$ & 0 & 0 & 2.1000$\times 10^6$ & 1.7074$\times 10^6$ & 0 & 0 \\
        Annual operating profit (\$) & 6.2021$\times 10^6$ & 9.4056$\times 10^6$ & 8.8558$\times 10^6$ & 1.1630$\times 10^7$ & 9.1588$\times 10^6$ & 1.2470$\times 10^7$ & 1.1812$\times 10^7$ & 1.4695$\times 10^7$ \\
        \bottomrule
    \end{tabular}
\end{table*}

\begin{table*}[!t]
    \centering
    \footnotesize
    \renewcommand{\arraystretch}{1.5}
    \caption{\small RCHP Revenue Breakdown in 2022 (MISO)}
    \label{tab:breakdown_MISO}
     \begin{tabular}{l c c c c c c c c }
        \toprule
        & \multicolumn{2}{c}{M0: Standalone}
        & \multicolumn{2}{c}{M1-p: Producer}  & \multicolumn{2}{c}{M1-c: Consumer} & \multicolumn{2}{c}{M2: Prosumer} \\
         \midrule
        Renewable Type & Solar & Wind & Solar & Wind & Solar & Wind & Solar & Wind \\
        Total renewable generation ( MWh) & ~~/~~ & 1.6659$\times 10^5$ &~~/~~  & 1.6659$\times 10^5$ & ~~/~~  & 1.6659$\times 10^5$ & ~~/~~ & 1.6659$\times 10^5$\\
        Renewable in hydrogen production (\%) & ~~/~~ & 81.65 & ~~/~~ & 79.67 & ~~/~~ & 81.65 & ~~/~~ & 79.67 \\
        Hydrogen produced (kg) & ~~/~~ & 2.5844$\times 10^6$ & ~~~/~~ & 2.5218$\times 10^6$ & ~~~/~~ & 3.2627$\times 10^6$ & ~~/~~ & 3.2002$\times 10^6$\\
        Revenue from hydrogen sales (\$) & ~~/~~ & 1.8090$\times 10^7$ & ~~~/~~ & 1.7653$\times 10^7$ & ~~~/~~ & 2.2839$\times 10^6$ & ~~/~~ & 2.2402$\times 10^7$\\
        Renewable sold in the market (\%) & ~~/~~ & 0 & ~~~/~~ & 20.33 & ~~~/~~ & 0 & ~~/~~ & 20.33 \\
        Revenue from renewable sales (\$) & ~~/~~ & 0 & ~~/~~ & 2.4421$\times 10^6$ & ~~~/~~ & 0 & ~~/~~ & 2.4421$\times 10^6$ \\
        Renewable curtailed (\%) & ~~/~~ & 18.35 & ~~~/~~ & 0 & ~~~/~~ & 18.35 & ~~/~~ & 0 \\
        Revenue lost due to curtailment (\$) & ~~/~~ & 1.7848$\times 10^6$ & ~~~/~~ & 0 & ~~~/~~ & 1.7848$\times 10^6$ & ~~/~~ & 0 \\
        Annual operating profit (\$) & ~~/~~ & 1.6541$\times 10^7$ & ~~/~~ & 1.8552$\times 10^7$ & ~~/~~ & 1.8003$\times 10^6$ & ~~/~~ & 2.0013$\times 10^7$ \\
        \bottomrule
    \end{tabular}
\end{table*}



\end{document}